\documentclass[preprint]{imsart}

\RequirePackage{amsthm,amsmath,amsfonts,amssymb}
\RequirePackage[numbers]{natbib}

\usepackage[utf8]{inputenc}

\usepackage{amsmath,amssymb,graphics,dsfont,bm,xcolor,float}

\usepackage{graphicx}
\usepackage[caption = false]{subfig} 

\usepackage{array,multirow}
\newcommand{\STAB}[1]{\begin{tabular}{@{}c@{}}#1\end{tabular}}

\startlocaldefs

\newtheorem{theorem}{Theorem}
\newtheorem{lemma}{Lemma}
\newcommand{\R}{\mathbb{R}}

\newcommand{\p}{\mathbb{P}}

\newcommand{\E}{\mathbb{E}}

\newcommand{\B}{\mathcal{B}}

\newcommand{\X}{\mathcal{X}}

\newcommand{\1}{\mathds{1}}

\newcommand{\dd}{\mathrm{d}}

\endlocaldefs

\begin{document}

\begin{frontmatter}
\title{A presmoothing approach for estimation in {semiparametric} mixture cure models}
\runtitle{A presmoothing approach for  mixture cure models}
\runauthor{Musta, Patilea, Van Keilegom}
\begin{aug}
\author[A]{\fnms{Eni} \snm{Musta}\ead[label=e1,mark]{e.musta@uva.nl}},
\author[B]{\fnms{Valentin} \snm{Patilea}\ead[label=e2]{valentin.patilea@ensai.fr}}
\and
\author[A]{\fnms{Ingrid} \snm{Van Keilegom}\ead[label=e3,mark]{ingrid.vankeilegom@kuleuven.be}}
\address[A]{ORSTAT, KU Leuven, Belgium, \printead{e1,e3}}

\address[B]{CREST, Ensai, France, \printead{e2}}
\end{aug}

\begin{abstract}
A challenge {when} dealing with survival analysis data is accounting for  a cure fraction,  meaning that some subjects will never experience the event of interest.  Mixture cure models have been frequently used to estimate both the probability of being cured and the time to event for the susceptible subjects, by usually assuming a parametric (logistic) form of the incidence.  
We propose a new estimation procedure for a parametric cure rate that relies on  a preliminary smooth estimator and is independent of the model assumed for the latency.  We investigate the {theoretical properties} 
of the estimators and show through simulations that, in the logistic/Cox model, presmoothing leads to more accurate results compared to the maximum likelihood estimator. To illustrate the practical use, we apply the new estimation procedure to two studies of melanoma survival data.
\end{abstract}

\begin{keyword}
\kwd{cure models}
\kwd{ kernel smoothing}
\kwd{logistic model}
\kwd{survival analysis}
\end{keyword}

\end{frontmatter}


\section{Introduction}
There are many situations in survival analysis problems where some of the subjects will never experience the event of interest. For instance, as significant progress is being made for treatment of different types of cancers, many of the patients get cured of the disease and do not experience recurrence or cancer-related death. Other examples include study of time to natural conception,
time to default in finance and risk management,
time to early failure of integrated circuits in {engineering, time to find a job after a layoff.}
However, because of the finite duration of the studies and censoring, the cured subjects (for which the event never takes place) cannot be distinguished from the `susceptible' ones. We can just get an indication of the presence of a cure fraction from the context of the study and  a long plateau (containing many censored observations) with height greater than zero in the Kaplan-Meier estimator of the survival function. Predicting the probability of being cured given a set of characteristics is often of particular interest in order to make better decisions in terms of {treatment, management strategies or public policies.}
This lead to the development of mixture cure models.

Mixture cure models were first proposed by \cite{boag49} and \cite{berkson52}. They assume that the population is a mixture of two groups: the cured  and the susceptible subjects. Within this very wide class of models, various approaches have been considered in the literature for modelling and {estimating}  the incidence (probability of being uncured) and the latency (survival function of the uncured subjects). Initially, fully parametric models with a logistic regression form of the incidence and various parametric distributions for the latency  were used in \cite{farewell82,yamaguchi92,kuk92}.  Later on, more flexible semi-parametric approaches were proposed for the latency based on the Cox proportional hazards model \cite{ST2000,peng2000} or accelerated failure time models \cite{li2002,zhang2007}. However, they still maintain the logistic regression model for the incidence.  More recently, nonparametric methods have been developed for both or one of the model components in \cite{XP2014,PK2019,AKL19}. 
{In this} wide range of models, probably the most commonly used one in practice is the {logistic/Cox mixture} cure model  \citep{stringer2016cure,wycinka2017,lee2017extinct}.

There have been different proposals for estimation in the logistic/Cox mixture cure model. The presence of a latent variable (the unknown cure status), does not allow for a `direct' approach as in the classical Cox proportional hazards model. \cite{kuk92} adapted a marginal likelihood approach computed through Monte Carlo approximations, whereas \cite{peng2000} and \cite{ST2000} computed the  maximum likelihood estimator via the Expectation-Maximization algorithm. Asymptotic properties of the latter estimators are investigated in \cite{Lu2008}, while the procedure is implemented in the package \texttt{smcure} \cite{cai_smcure}. 
One concern about the previous estimators is that they are obtained by iterative procedures which {could} be unstable  in practice. {In particular, when the sample size is small there are situations in which the EM algorithm fails to converge (even though the \texttt{smcure} package can still provide without error the estimates obtained when the maximum number of iterations is reached). Such problems are for example reported in \cite{han2017statistical}. In addition, the maximum likelihood estimator for the incidence component depends on which variables are included in the latency model (see for example the illustration in Section~\ref{sec:application}) and this instability might in practice lead to unobserved effects (when the effect is not very strong). In particular, if the latency model is misspecified, even the estimators of the incidence parameters suffer from induced bias (see for example \cite{BP18}).}

 In this paper, we introduce an alternative {estimation method which applies {very broadly and}, in particular, for the} logistic/Cox mixture cure model. {Our approach focuses on } direct estimation of the cure probability without using distributional assumptions on the latency and iterative algorithms.  It relies on  a preliminary nonparametric estimator for the incidence which is then `projected' on {a parametric class of functions (like logistic functions)}.  The idea of constructing a parametric estimator by nonparametric estimation has been previously proposed for the classical linear regression by \cite{cristobal1987class}. {Later on it was shown to be effective also in the context of variable selection and functional linear regression  \citep{presmoothing_var_sel,ferraty2012presmoothing}. However, its }
extension to nonlinear setups has been very little investigated. Here we show that in the context of mixture cure models, even when a parametric form is  assumed for the incidence, the use of a presmoothed estimator as an intermediate step for obtaining the parameter estimates often leads to more accurate results. Once the cure fraction is estimated, we estimate the survival distribution of the uncured subjects. {In the case of the logistic/Cox cure model, this is done by maximizing the Cox component of the likelihood.}  In this step, {an iterative algorithm}
is used to compute the estimators of the baseline cumulative hazard and the regression parameters. This new approach is of practical relevance given  the popularity {of the semiparametric logistic/Cox mixture cure model.} However, the method can be applied more in general to a mixture cure model with a parametric form of the incidence and {other type of models for the uncured subjects, such as the semiparametric proportional odds model or the semiparametric AFT model. Our findings} suggest that presmoothing has potential {to improve parameter estimation for small and moderate sample size.}

The paper is organized as follows. In Sections \ref{sec:model} and \ref{sec:method} we describe the model and the estimation procedure.   {Section \ref{sec:Cox} focuses on the estimation method in the case of the logistic/Cox mixture cure model.}
Consistency and asymptotic normality of the estimators are shown in Section~\ref{sec:asymptotics}. {Thanks to the presmoothing, we are able to present theoretical results under more reasonable assumptions and thus we contribute to fill a gap between unrealistic technical conditions and applications. 
	The finite sample performance of the method is investigated through a simulation study and results are reported in {Section~\ref{sec:simulations}}.  
	{For practical purposes, we propose to make simple and commonly used choices for the bandwidth and the kernel function in the presmoothing step, and we show that these choices provide satisfactory results.}
	  The proposed estimation procedure is applied to two medical datasets about studies of patients with melanoma cancer (see Section~\ref{sec:application}).   {We conclude in Section \ref{sec:disc} with some discussion and ideas for further research.  Finally, some of the proofs can be found in Section \ref{sec:appendix}, while the remaining proofs {and additional simulation results} are collected in the online Supplementary Material.}

	\section{Model description}
	\label{sec:model}
	In the mixture cure model  the survival time $T$ can be decomposed as 
	\[
	T=BT_0+(1-B)\infty,
	\]
	where $T_0$ represents {the finite survival time} for an uncured individual and  {$B$ is an unobserved $0$-$1$ random variable giving the uncured status: $B=1$ for uncured individuals and $B=0$ otherwise}. By convention $0\cdot \infty = 0$.  Let $C$ be the censoring time and $(X',Z')'$ a $(p+q)$-dimensional vector of  covariates, where $x'$ denotes the transpose of the vector $x$. Let $\X$ and $\mathcal{Z}$ be the supports of $X$ and $Z$ respectively. Observations consist {of} $n$ i.i.d. realizations of $(Y,\Delta,X,Z)$, where $Y=\min(T,C)$ is {the finite follow-up time} and $\Delta=\1_{\{T\leq C\}}$ is the censoring indicator. {Since $Y$ is finite, then necessarily $\p(C<\infty)=1$, that means} the censoring times are finite ({which makes sense given the} limited duration of the studies). As a result, {censored survival times of the uncured subjects} cannot be distinguished from the cured ones. 
	
	The covariates included in $X$ are those used to model the cure rate, while the ones in $Z$ affect the survival conditional on the uncured status. This allows in general to use different variables for modelling the incidence and the latency but does not exclude situations in which the two vectors $X$ and $Z$ share some components or are exactly the same.  
	Apart from the standard assumption in survival analysis that $T_0\perp (C,X) | Z$, here we also need 
	\begin{equation}\label{eqn:bb_ind}
		B\perp (C,T_0,Z)|X.
	\end{equation} 
	This implies in particular that 
	\begin{equation}
		\label{eqn:CI1}
		T\perp C| (X,Z)
	\end{equation}
	(see {Lemma~1 in the Suplementary Material}). Moreover, \eqref{eqn:bb_ind} implies 
	\begin{equation}
		\label{eqn:condition_X_Z}
		\p(T=\infty|X,Z)=\p(T=\infty|X). 
	\end{equation}
	In addition, in the cure model context we need that the event time $T_0$ has support $[0,\tau_0]$, {i.e.  $\{T>\tau_0\}=\{T=\infty\}$, such} that   
	\begin{equation}
		\label{eqn:CI2}
		{\inf_{x}\p(C>\tau_0|X=x)>0.}
	\end{equation}
	{(If the support of $T_0$ given $Z=z$ depends on $z$, then we let $\tau_0 = \sup \tau_0(z)$, where $\tau_0(z)$ is the right endpoint of this support.)}   This condition tells us that all the observations with $Y>\tau_0$ are cured. Even if it might {seem restrictive,} it is reasonable when a cure model is justified by a `good' follow-up beyond the time when most of the events occur and it is commonly accepted in the cure model literature in order for the mixture cure model to be identifiable and not to  overestimate the cure rate. Since $T_0\perp X | Z$, we have 
	\begin{equation*}
		\p(T_0\leq t |X,Z)=\p(T_0\leq t|Z), \quad \forall t\in [0,\tau_0].
	\end{equation*}
	
	
	We assume a parametric model for the cure rate and we denote by $\pi_0(x)$ the cure probability of a subject with covariate $x$, i.e
	\[
	\pi_0(x)=\p(T=\infty|X=x)=1-\phi(\gamma_0,x),
	\]
	for some {parametric model $\{\phi (\gamma,x): \gamma \in G\}$ and $\gamma_0\in G$.} The first component of $X$ is equal to one and {the first component of $\gamma$} corresponds to the intercept. In order for $\gamma$ to be identifiable we need the following condition
	\begin{equation}
		\label{eqn:CI3}
		\p\left(\phi(\gamma,X)=\phi(\tilde{\gamma},X)\right)=1\qquad\text{ implies that }\qquad \gamma=\tilde{\gamma}.
	\end{equation}
{Choosing a parametric model for the incidence seems quite standard in the literature of mixture cure models (\cite{PK2019,BP18,sposto2002cure}) because of its simplicity and ease of interpretability (particularly for multiple covariates).  To check the fit of this model in practice, one can compare the prediction error with that of a more flexible single-index model as done in \cite{AKL19} and for our real data application  in Section~\ref{sec:application}. It is also possible to test whether this assumption is reasonable using the test proposed in \cite{muller2019goodness}, but this is currently developed only for one covariate. Among the parametric models for the incidence component, } the most common example is the logistic model, where 
	\begin{equation}
		\label{eqn:logistic}
		\phi(\gamma,x)=1/(1+{\exp(-\gamma' x)}).
	\end{equation}
	{We state} the results in Section \ref{sec:asymptotics} for a general parametric model for the incidence, but then we focus on the logistic function in the simulation study in Section \ref{sec:simulations} since it is more of interest in practice. 
	For the uncured subjects, we {can consider a general semiparametric model defined through the survival function 
		\begin{equation} \label{Su}
			S_u(t|z) =S_u(t|z;\beta,\Lambda) = \p(T_0>t| {Z=z}, B=1) \quad\text{and}\quad S_u(\tau_0|z) =0,
		\end{equation}
		{where the conditional survival function $S_u$ is allowed to depend on a finite-dimensional parameter, denoted by $\beta\in \mathcal B$, and/or an infinite-dimensional parameter, denoted by $\Lambda\in \mathcal H$, with $\mathcal B$ and $\mathcal H$ the respective parameter sets.} Let $\beta_0\in\mathcal  B$ and $\Lambda_0 \in\mathcal H$ be the true values of these parameters. As a result, the conditional survival function corresponding to $T$  is then
		$$
		S(t|x,z)=\p(T>t|{X=x, Z=z})=1-\phi(\gamma_0,x)+\phi(\gamma_0,x)S_u(t|z).
		$$
		The main example we keep in mind is the Cox proportional hazards (PH) model where  $\Lambda_0$ is the baseline cumulative hazard. In this case 
		\begin{equation} \label{SuCox}
			S_u(t|z) 
			=S_0(t)^{\exp(\beta'_0z)} = \exp(-\Lambda_0(t)\exp(\beta'_0z)),
		\end{equation}
		where $S_0 $ is the baseline survival and $\beta_0 $ does not contain an intercept.}


	\section{{Presmoothing estimation approach}}
	\label{sec:method}
	The estimation method we propose is based on a two step procedure. We first estimate nonparametrically the cure probability for each observation and then compute an estimator of $\gamma$ as the maximizer of the logistic likelihood, ignoring {the model for the uncured subjects.}
	In the second step, we plug-in this estimator of $\gamma$ in the {full  likelihood of the
		mixture cure model and fit the 
		latency model 
		using maximum likelihood estimation.} 
	In what follows, we describe in more details these two steps. 
	
	\textit{Step 1.} Even though a parametric model is assumed for the incidence, we start by computing a nonparametric estimator of the cure probability for each subject. One possibility is to use  the method {followed by \cite{PK2019} (see also \cite{XP2014})}, {but other estimators are possible as well, as long as the conditions given in Section \ref{sec:asymptotics} are satisfied.  The estimator of \cite{PK2019} is defined as follows:}
	\begin{equation}
		\label{def:hat_pi}
		\hat\pi(x)=\prod_{t\in\R} \left(1-\frac{\hat{H}_1(dt|x)}{\hat{H}([t,\infty)|x)}\right){,}
	\end{equation}
	where  {$\hat{H}([t,\infty)|x)=\hat{H}_1([t,\infty)|x)+\hat{H}_0([t,\infty)|x)$,} $\hat{H}_1(dt|x) = \hat{H}_1((t-dt,t]|x)$ for small $dt$ and 
	\[
	\hat{H}_k([t,\infty)|x)=\sum_{i=1}^n\frac{\tilde{K}_b(X_i-x)}{\sum_{j=1}^n\tilde{K}_b(X_j-x)}\1_{\{Y_i\geq t, \Delta_i=k\}},\quad k=0,1{,}
	\] 
	are estimators of
	\[
	H_k([t,\infty)|x)=\p\left(Y\geq t,\Delta=k|X=x\right), 
	\] 
	$H([t,\infty)|x)=H_1([t,\infty)|x)+H_0([t,\infty)|x)$. Here $\tilde{K}_{{b}}$ is a multidimensional kernel function defined in the following way. If $X$ is composed of continuous and discrete components, {$X=(X_c,X_d)\in\X_c\times\X_d\subset\R^{p_c}\times\R^{p_d}$ with $p_c+p_d=p$}, then
	\[
	\tilde{K}_b(X_i-x)=K_b(X_{c,i}-x_c)\1_{\{X_{d,i}=x_d\}},
	\] 
	where $b=b_n$ is a bandwidth sequence, $K_b(\cdot)=K(\cdot/b)/b^{p_c}$ and $K(u)=\prod_{j=1}^{p_c}k(u_j)$, {with $k$ a kernel.} 
	Note that, one can compute this estimator with any covariate but here we only use $X$ because of our assumption~\eqref{eqn:condition_X_Z}.
	{The estimator} $\hat\pi(x)$ coincides with the Beran estimator of the {conditional survival function} $S$ at the largest observed event time $Y_{(m)}$ {and does not require any specification of $\tau_0$}. {Since $\hat{H}_1(dt|x)$ is different from zero only at the observed event times, computation of $\hat\pi(x)$ requires only a product over $t$ in the set of the observed event times.} Afterwards, we consider the logistic likelihood
	\[
	\hat L_{n,1}(\gamma)=\prod_{i=1}^n \phi(\gamma,X_i)^{1-\hat\pi(X_i)}(1-\phi(\gamma,X_i))^{\hat\pi(X_i)}{,}
	\]
	and define $\hat\gamma_n$ as the maximizer of 
	\begin{equation}
		\label{def:hat_L_gamma}
		\log \hat{L}_{n,1}(\gamma)=\sum_{i=1}^n\Big\{\left[1-\hat\pi(X_i)\right] \log  \phi(\gamma,X_i)+\hat\pi(X_i)\log \left[1-\phi(\gamma,X_i)\right]\Big\}.
	\end{equation}
	Existence and uniqueness of $\hat\gamma_n$ holds under the same conditions as for the maximum likelihood estimator in the binary outcome regression model where $1-\hat\pi(X_i)$ is replaced by {the outcome $B_i$}. For example, in the logistic model,   it is required that $p<n$ and the matrix of the variables $X$ has full rank.

	\textit{Step 2.} {Now we consider the  likelihood of the 
		mixture cure model.  Let {$$f_u(t|z;\beta,\Lambda) = -(\partial/\partial t) S_u(t|z;\beta,\Lambda)$$ with $S_u(t|z;\beta,\Lambda)$ as defined in (\ref{Su}),} {denote an element in the model for the conditional density of} $T_0$ given $Z=z$, which is supposed to exist {and belong to the model}.} Assuming  non informative censoring and that the distribution of the covariates does not carry information on the parameters $\beta$, $\Lambda$, the likelihood criterion is then
	{\begin{equation}
			\label{eqn:likelihood2}
			L_{n,2}( \beta,\Lambda,\gamma)=\prod_{i=1}^{n}\left\{\phi(\gamma,X_i)f_u(Y_i|Z_i;\beta,\Lambda)
			\right\}^{\Delta_i}\!\left\{1\!-\phi(\gamma,X_i)+\phi(\gamma,X_i)S_u(Y_i|Z_i;\beta,\Lambda) \right\}^{1-\Delta_i},
		\end{equation}
		and we maximize it w.r.t. $\beta$ and $\Lambda$ for $\gamma=\hat\gamma_n$, i.e.  $(\hat\beta_n,\hat\Lambda_n)$  are the maximizers of 
		\begin{equation}\label{eqn:likelihood2b}
			\hat{l}_n(\beta,\Lambda,\hat\gamma_n) 
			= \frac{1}{n}\sum_{i=1}^n \ell (Y_i,\Delta_i,X_i,Z_i;\beta,\Lambda,\hat\gamma_n),
		\end{equation}
		over a set of possible values for $\beta $ and $\Lambda$, where
		\begin{equation}
			\label{def:hat_l_cox_g}
			\begin{split}
				\ell (Y_i,\Delta_i,X_i,Z_i;\beta,\Lambda,\gamma) &=\Delta_i  \log f_u(Y_i|Z_i;\beta,\Lambda) \\
				&\quad+ 
				(1-\Delta_i)\log\left\{1-\phi(\gamma,X_i)+\phi(\gamma,X_i)S_u(Y_i|Z_i;\beta,\Lambda)\right\}.
			\end{split}
		\end{equation}
	}

	\section{{Presmoothing estimation for the parametric/Cox mixture cure model}} 
	\label{sec:Cox}

	In the sequel we focus on the case of a Cox PH model {defined in (\ref{SuCox})} for the conditional law of $T_0$. 
	The criterion defined in \eqref{eqn:likelihood2b} becomes 
	\begin{multline}
		\label{def:hat_l_cox}
		\hat{l}_n(\beta,\Lambda,\hat\gamma_n)  
		=\frac{1}{n}\sum_{i=1}^n \Delta_i  \left\{{\1_{\{Y_i<\tau_0\}}}[ \log \Delta\Lambda(Y_i)+\beta'Z_i]-\Lambda(Y_i)e^{\beta'Z_i} \right\} 
		\\
		+\frac{1}{n}\sum_{i=1}^n(1-\Delta_i )\log\left\{1-\phi(\hat\gamma_n,X_i)+\phi(\hat\gamma_n,X_i)\exp\left(-\Lambda(Y_i)e^{\beta'Z_i}\right)\right\},\,
	\end{multline}
	and has to be maximized with respect to $\beta$ and $\Lambda$ in  
	the class of step functions $\Lambda$  defined on $[0,\tau_0]$ (thus by definition $\Lambda(t)={\infty}$ if $t>\tau_0$), with jumps of size $\Delta \Lambda$ at the event times. {The indicator of the event $\{Y_i<\tau_0\}$ in the first term is needed in case the distribution of the event times has a jump at $\tau_0$ meaning that $\p(T_0=\tau_0|Z)>0$. In such a case $f_u(\tau_0|Z;\beta,\Lambda)=\exp(-\Lambda(\tau_0)e^{\beta'Z})$ where $\Lambda(\tau_0)=\lim_{t\uparrow\tau_0}\Lambda(t)$. Otherwise, if $\p(T_0=\tau_0|Z)=0$, then  for all uncensored observations we have $\1_{\{Y<\tau_0\}}=1$ with probability one. Thus, the presence of the indicator function can be neglected. }
	As in \cite{Lu2008}, it can be shown that  
	\begin{equation}\label{def_MLE}
		(\hat\beta_n, \hat\Lambda_n)= \arg\max_{\beta, \Lambda }
		\hat{l}_n(\beta,\Lambda,\hat\gamma_n)
	\end{equation}
	exists and it is finite. Moreover, for any given $\beta$ and $\gamma$, the ${\Lambda}_{n,\beta,\gamma}$ which maximizes $\hat{l}_n(\beta,\Lambda,\gamma)$ {in \eqref{def:hat_l_cox}}, 
	with respect to $\Lambda$  with jumps at the event times,
	can be characterized as 
	\begin{equation}
		\label{eqn:hat_lambda_n}
		{\Lambda}_{n,\beta,\gamma}(t)=\frac{1}{n}\sum_{i=1}^{n}\frac{\Delta_i\1_{\{Y_i\leq  {t {, Y_i<\tau_0} }\}}}{\frac{1}{n}\sum_{j=1}^n\1_{\{{{Y_i\leq Y_j\leq\tau_0}  }\}}\exp(\beta'Z_j)\left\{{\Delta_j}+(1-{\Delta_j})g_j(Y_j,{\Lambda}_{n,\beta},\beta,\gamma)\right\}} ,
	\end{equation}
	where
	\begin{equation}
		\label{def:g_j}
		g_j(t,\Lambda,\beta,\gamma)=\frac{\phi(\gamma,X_j)\exp\left(-\Lambda(t) \exp\left(\beta' Z_j\right)\right)}{1-\phi(\gamma,X_j)+\phi(\gamma,X_j)\exp\left(-\Lambda(t) \exp\left(\beta' Z_j\right)\right)}.
	\end{equation}
	Next, we could define
	$$
	\hat\beta_n = \arg\max_{\beta}\hat{l}_n(\beta,{\Lambda}_{n,\beta,\hat\gamma_n} ,\hat\gamma_n) \quad \text{and} \quad \hat\Lambda_n  = {\Lambda}_{n,\hat\beta_n,\hat\gamma_n}.
	$$
	
	To compute $(\hat\beta_n,\hat\Lambda_n)$ we use an iterative algorithm based on profiling.  
	To be precise, we start with initial values which are the maximum partial likelihood estimator and the Breslow estimator (as if there was no cure fraction) and we iterate between the next two steps until convergence:
	\begin{itemize}
		\item[{a)}] Compute the weights 
		\[
		w_j^{(m)}={\Delta_j}+(1-{\Delta_j})\frac{\phi(\hat\gamma_n,X_j)\hat{S}^{(m)}_u(Y_j|Z_j)}{1-\phi(\hat\gamma_n,X_j)+\phi(\hat\gamma_n,X_j)\hat{S}^{(m)}_u(Y_j|Z_j)},
		\]
		where
		\[
		\hat{S}^{(m)}_u(Y_j|Z_j)=\exp\left(-\hat{\Lambda}^{(m)}_n(Y_j) \exp\left(\hat\beta^{(m)'}_n Z_j\right)\right),
		\]
		using the estimators $\hat{\Lambda}^{(m)}_n$, $\hat\beta^{(m)}_n $ of the previous step. 
		\item[{b)}] Using the previous weights, update the estimators for $\Lambda$  and $\beta$, i.e. $\hat\beta^{(m+1)}_n$ is the maximizer of
		\[
		\prod_{i=1}^n \left\{\frac{e^{\beta'Z_i}}{\sum_{{Y_k\geq Y_i }}w_k^{(m)}e^{\beta'Z_k}}\right\}^{\Delta_i
		}
		\]
		and
		\begin{equation}
			\label{eqn:lambda}
			\hat\Lambda_n^{(m+1)}(t)=\sum_{i=1}^n\frac{\Delta_i\1_{\{{Y_i\leq t{, Y_i<\tau_0}}\}}}{\sum_{j=1}^n\1_{\{{{Y_i\leq Y_j\leq\tau_0}}\}} w_j^{(m)}\exp\left(\hat{\beta}_n^{(m+1)'} Z_j\right)}.
		\end{equation}
	\end{itemize}
	{The update of $\Lambda$ an $\beta$ in {Step (b)} coincides with the maximization step of the EM algorithm and the weights $w^{(m)}$ correspond to the expectation of the latent variable $B$ given the observed data and the current parameter values. However, unlike the maximum likelihood estimation  \citep{ST2000}, we are  keeping $\hat\gamma_n$ fixed while performing this iterative algorithm.}  {The estimator $\hat\Lambda_n$ seems to depend on the unknown $\tau_0$. However,  with data at hand, one could easily proceed without knowing $\tau_0$. Indeed, if there are ties at the last uncensored observation, then $\tau_0$ is revealed by the data. On the other hand, if there are no ties, all uncensored observations will be smaller than $\tau_0$, hence no need to know $\tau_0$.}
	
	
	As suggested in \cite{taylor95,ST2000}, we impose the zero-tail constraint, meaning that $\hat{S}_u^{(m)}$ is forced to be equal to zero beyond the last event. In this way, all censored observations in the plateau are assigned to the cured group. 
	

	\section{Asymptotic results}
	\label{sec:asymptotics}
	
	We first explain {why} presmoothing  {allows for more realistic} asymptotic results in semiparametric  mixture cure models. 
	Next, we 
	show consistency and asymptotic normality of the proposed estimators $\hat{\gamma}_n$, $\hat\beta_n$ and $\hat\Lambda_n$ 
	{for the parametric/Cox 
		mixture cure model}
	when, in Step 1, we  use a general nonparametric estimator $\hat\pi$ of $\pi_0$ that satisfies certain assumptions. Afterwards, we verify these conditions for the particular estimator $\hat\pi$ in \eqref{def:hat_pi}. 
	Some of the proofs can be found in Section~\ref{sec:appendix} and the rest in  the online Supplementary Material.
	The assumptions mentioned in  Section \ref{sec:model} are assumed to be satisfied throughout this section. {In addition $Var(Z)$ is supposed to have full rank. }
	
	
	\subsection{A challenge with mixture cure models}
	
	To derive asymptotic results, in most of the existing literature it has been  assumed  that
	\begin{equation}
		\label{eqn:jump_cond}
		\inf_{z}\p(T_0\geq\tau_0|Z=z)>0, 
	\end{equation}
	\citep{Lu2008,PK2019}. In nonparametric approaches such a condition keeps the denominators away from zero. In the parametric/Cox 
	mixture cure model, it guarantees that  the baseline distribution stays bounded on the compact support $[0,\tau_0]$.
	{However, condition \eqref{eqn:jump_cond} implies that $\inf_{z}\mathbb P (Y=\tau_0{,\Delta=1}|Z = z)  >0$, a condition which is not frequently satisfied in real-data applications.

		One could imagine that, instead of imposing condition \eqref{eqn:jump_cond}, it could be possible to proceed as follows: first restrict to events on $[0,\tau^*]$ for some 
		$\tau^* < \tau_0$ such that 
		\begin{equation}
			\label{eqn:no_jump_cond}
			\inf_{z}\p(Y\geq\tau^*{,\Delta=1}|Z=z)>0, 
		\end{equation}
		next derive the asymptotics, and finally let $\tau^*$ tend to $\tau_0$. This idea is used, for instance, in Cox PH model, see  \cite{FleHar} chapter 8, or \cite{AndGill82}. However, this idea does not seem to work for  mixture cure models without suitable adaptation. This is because it implicitly requires 
		that $\beta_0$ and $\Lambda_0$ are identifiable from the restricted  data. Here, identifiability means that the true values $\beta_0$ and $\Lambda_0$ of the parameters maximize the expectation of the criterion maximized to obtain the estimators. 
		Two aspects have to be taken into account when analyzing this identifiability. The first aspect is related to the parameter identifiability in the semiparametric model for $T_0$ when the events are restricted to $[0,\tau^*]$. This property is satisfied in the common models, in particular it holds true in the Cox PH model as soon as $Var(Z)$ has full rank. 
		The second aspect is the additional complexity induced by the mixture with a cure fraction. If the cure fraction is unknown and one decides to restrict to events on  $[0,\tau^*]$, the parameter identifiability is likely lost because  the events $\{T_0\in (\tau^*, \tau_0]\}$ and $\{T=\infty\}$ are not distinguishable. The usual remedy for this is to impose \eqref{eqn:jump_cond}, so that  $\tau^*$ could be taken equal to $\tau_0$.

		Presmoothing allows  to avoid  condition \eqref{eqn:jump_cond}  and thus to fill the gap between the technical conditions and the reality of the data. This is possible because, when using the presmoothing, the conditional probability of the event $\{T=\infty\}$ is identified by other means. We are thus able to prove the consistency of $\hat \beta$ and $\hat\Lambda$ without imposing   \eqref{eqn:jump_cond}. Deriving the asymptotic normality without \eqref{eqn:jump_cond} remains an open problem which will be addressed elsewhere.

		\subsection{Consistency}
		
		We first prove  consistency of $\hat\gamma_n$ and then use that result to obtain consistency of $\hat\Lambda_n$ and $\hat\beta_n$. In order to proceed with our results, the following conditions will be used.
		\begin{itemize}
			\item[(AC1)] $\sup_{x\in\X}\left|\hat\pi(x)-\pi_0(x)\right|\to 0$ almost surely. 
			\item[(AC2)] The parameters $\beta_0$ and $\gamma_0$ lie in the interior of compact sets $B\subset\R^q$, $G\subset\R^p$.
			\item[(AC3)]There exist some constants $a>0$, $c>0$ such that
			\[
			\left|\phi(\gamma_1,x)-\phi(\gamma_2,x)\right|\leq c\Vert\gamma_1-\gamma_2\Vert^a,\qquad\forall\gamma_1,\gamma_2\in G,\,\forall x\in\X,
			\]
			where $\Vert\cdot\Vert$ denotes the Euclidean distance.
			\item[(AC4)] $\inf_{\gamma\in G}\inf_{x\in\X}\phi(\gamma,x)>0$ and $\inf_{\gamma\in G}\inf_{x\in\X}\phi(\gamma,x)<1$.
			\item[(AC5)] The covariates are bounded: $\p\left(\Vert Z\Vert<{m} \text{ and } \Vert X\Vert<m\right)=1$ for some $m>0$.
			\item[(AC6)] {The baseline hazard function $\lambda_0(t)=\Lambda'_0(t)$} is strictly positive and continuous on   $[0,{\tau_0)}$.  
			\item[(AC7)] With probability one,   the conditional distribution function of the censoring times $F_C(t|x,z)$  is continuous in $t$ on $[0,{\tau_0}]$ and   
			there exists a constant $C>0$ such that
			$$
			\inf_{0\leq t_1<t_2\leq\tau_0} \inf_{x,z} \frac{F_C(t_2|x,z)- F_C(t_1|x,z)}{t_2-t_1} >C.
			$$
		\end{itemize}
		{(AC1) is a minimal assumption given that we want to match $\phi(\gamma,\cdot)$ to $\hat \pi (\cdot)$. 
			(AC2) to (AC4) are mild conditions satisfied by usual binary regression models, like for instance the logistic one, and (AC5) is always satisfied in practice for large $m$. }
		\begin{theorem}
			\label{theo:consistency}
			{Let the estimator $\hat\gamma_n$ be defined as in (\ref{def:hat_L_gamma}).}
			Assume that 
			(AC1)-(AC4) hold. Then, $\hat\gamma_n\to\gamma_0$ almost surely.
		\end{theorem}
		\begin{theorem}
			\label{theo:consistency2}
			Let the estimators $\hat\beta_n$ and $\hat\Lambda_n$ be defined as in Section~\ref{sec:Cox}.
			Assume that
			(AC1)-(AC7) hold.  Then, with probability one,
			$\Vert\hat\beta_n-\beta_0\Vert\to 0$,
			where $\Vert\cdot\Vert$ denotes the Euclidean distance. Moreover, 
			for any $\tau^*\leq  \tau_0$
			satisfying  \eqref{eqn:no_jump_cond}, with probability one,
			\[
			\sup_{t\in[0,{\tau^*}]}\left|\hat\Lambda_n(t)-\Lambda_0(t)\right|\to 0.
			\]
		\end{theorem}
		{When condition \eqref{eqn:jump_cond} is satisfied and $\tau^*=\tau_0$ in the previous Theorem, we are referring to the continuous version of $\Lambda_0$, i.e.  $\Lambda_0(\tau_0)=\lim_{t\uparrow\tau_0}\Lambda_0(t)$. Note that, by definition,  we also have $\hat\Lambda_n(\tau_0)=\lim_{t\uparrow\tau_0}\hat\Lambda_n(t)$. } 
		\subsection{Asymptotic normality}
		We first derive asymptotic normality of $\hat\gamma_n$ following the approach in \cite{chen2003}. Theorem 2 in that paper provides sufficient conditions for the $\sqrt{n}$ normality of parametric estimators obtained by minimizing an objective function that depends on a preliminary infinite dimensional estimator $\hat\pi$. In our case, since $\hat\gamma_n$ solves
		\[
		\frac{1}{n}\nabla_\gamma\log \hat{L}_{n,1}(\gamma)=0,
		\] 
		where $\nabla_\gamma$ denotes the vector-valued partial differentiation operator with respect to the components of $\gamma$, it follows that $\hat\gamma_n$ minimizes the function 
		\[
		\left\Vert\frac{1}{n}\nabla_\gamma\log \hat{L}_{n,1}(\gamma)\right\Vert=\left\Vert\frac{1}{n}\sum_{i=1}^n m(X_i;\gamma,\hat\pi)\right\Vert,
		\]
		where
		\begin{equation}
			\label{def:m}
			m(x;\gamma,\pi)=\left[\frac{1-\pi(x)}{\phi(\gamma,x)}-\frac{\pi(x)}{1-\phi(\gamma,x)}\right]\nabla_\gamma\phi(\gamma,x).
		\end{equation}
		Hence, we only need to check that the conditions of Theorem 2 in \cite{chen2003} are satisfied.           
		To do that, we need the following assumptions which are stronger than the previous (AC1)-(AC4).
		\begin{itemize}
			\item[(AN1)] The parameter $\gamma_0$  lies in the interior of a compact set $G\subset\R^p$ and, for each $x\in\X$, the function $\gamma\mapsto\phi(\gamma,x)$ is twice continuously differentiable with uniformly bounded derivatives in $G\times\X$ and satisfies (AC4).
			\item[(AN2)] $\pi_0(\cdot)$ belongs to a class of functions $\Pi$ such that 
			\begin{equation*}
				\label{eqn:entropy}
				\int_0^\infty \sqrt{\log N(\epsilon,\Pi,\Vert\cdot\Vert_{\infty})}\,\dd\epsilon<\infty,
			\end{equation*}
			where $N(\epsilon,\Pi,\Vert\cdot\Vert_{\infty})$ denotes the $\epsilon$-covering number of the space $\Pi$ with respect to $\Vert\pi\Vert_{{\infty}}=\sup_{x\in\X}|\pi(x)|$.
			\item[(AN3)] The matrix $\E\left[\nabla_\gamma\phi(\gamma_0,X)\nabla_\gamma\phi(\gamma_0,X)'\right]$ is positive definite.
			\item[(AN4)] The estimator $\hat\pi(\cdot)$ satisfies the following properties:
			\begin{itemize}
				\item[(i)] $\p\left(\hat\pi(\cdot)\in\Pi\right)\to 1$.
				\item[(ii)] $\left\Vert\hat\pi(x)-\pi_0(x)\right\Vert_{{\infty}}=o_P(n^{-1/4})$. 
				\item[(iii)] There exists a function $\Psi$ such that 
				\[
				\begin{split}
					&\E^*\left[\left(\hat\pi(X)-\pi_0(X)\right)\left(\frac{1}{\phi(\gamma_0,X)}+\frac{1}{1-\phi(\gamma_0,X)}\right)\nabla_\gamma\phi(\gamma_0,X)\right]\qquad\qquad\\
					&\qquad\qquad\qquad\qquad\qquad\qquad\qquad\qquad\qquad=\frac{1}{n}\sum_{i=1}^n\Psi(Y_i,\Delta_i,X_i)+R_n,
				\end{split}
				\]
				where $\E^*$ denotes the conditional expectation given the sample, taken with respect to the generic variable $X$. Moreover, $\E[\Psi(Y,\Delta,X)]=0$ and $\Vert R_n\Vert=o_P(n^{-1/2})$.
			\end{itemize}
		\end{itemize}
		\begin{theorem}
			\label{theo:asymptotic_normality}
			{Let the estimator $\hat\gamma_n$ be defined as in (\ref{def:hat_L_gamma}).}
			Assume that
			(AN1)-(AN4) hold. Then, \[
			n^{1/2}\left(\hat\gamma_n-\gamma_0\right)\xrightarrow{d}N(0,\Sigma_\gamma)
			\]
			with covariance matrix $\Sigma_\gamma$ defined { in~(A28).}
		\end{theorem}
		{For deriving the asymptotic distribution of $\hat\beta_n$ and $\hat\Lambda_n$ we assume, for simplicity,  that condition \eqref{eqn:jump_cond} is satisfied. In such case, in Theorem~\ref{theo:consistency2} we can take $\tau^*=\tau_0$ and obtain uniform strong consistency of $\hat\Lambda_n$ on the whole support $[0,\tau_0]$. We believe that, at the price  of additional technicalities, asymptotic distributional theory can be obtained also without imposing \eqref{eqn:jump_cond}, as we did for the consistency in Theorem~\ref{theo:consistency2}. This conjecture is supported by simulations but we leave the problem   to be addressed by future research. }
		\begin{theorem}
			\label{theo:asymptotic_normality2}
			Let the estimators $\hat\beta_n$ and $\hat\Lambda_n$ be defined as in Section~\ref{sec:Cox}. 	Assume that {condition~\eqref{eqn:jump_cond},}
			(AN1)-(AN4) and {(AC2)}, (AC5)-(AC7) hold. Then, 
			\[
			\left	\langle\sqrt{n}\left(\hat\Lambda_n-\Lambda_0\right),\sqrt{n}\left(\hat\beta_n-\beta_0\right)\right\rangle\to G 
			\]
			weakly in $l^\infty(\mathcal{H}_\mathfrak{m})$, where {$\mathcal{H}_\mathfrak{m}$ is a functional space defined in Section \ref{sec:appendix_Cox}}, { $l^\infty(\mathcal{H}_\mathfrak{m})$ denotes the  space of bounded real-valued functions on $\mathcal{H}_\mathfrak{m}$,}  $G$ is a tight Gaussian process in $l^\infty(\mathcal{H}_\mathfrak{m})$ with mean zero and covariance process given {in~(A39)} { and  for $h=(h_1,h_2)\in \mathcal{H}_\mathfrak{m}$
				\[
				\langle \Lambda,\beta\rangle(h)=\int_0^{\tau_0}h_1(t)\,\dd\Lambda(t)+h_2'\beta.
				\]	
			} 	
		\end{theorem}
	{The asymptotic variances of each component of $\hat\beta_n$ and of $\hat\Lambda_n(t)$ can be obtained from the covariance process in~(A39) by taking $h_1(t)=0$ for all $t$ and $h_2=e_i$ (the $i$th unit vector) or $h_2=0$ and $h_1(s)=\1_{\{s\leq t\}}$. We leave the details about these covariance matrices in the Supplementary Material because they have quite complicated expressions that require definitions of several other quantities. Even though it could be possible in principle to estimate the asymptotic standard errors through plug-in estimators and numerical inverse, we think that this is not feasible in practice and we  do not intend to exploit it further. Instead, we use a bootstrap procedure for estimation of the standard errors in the application discussed in Section~\ref{sec:application}.  However,  the maximum likelihood estimators are not more favorable in this regard. For example, in the logistic/Cox model, the proposed estimators of the asymptotic variance in \cite{Lu2008} also involve solving numerically complicated nonlinear equalitons. For this reason,  bootstrap is used in practice to estimate the standard errors even for the maximum likelihood estimators. }
	
	{By considering a two-step procedure, where estimation of the incidence parameters is performed independently of the latency model, we expect to loose efficiency of the estimators. However, this does not cause major concern because our  purpose is to provide an alternative estimation method that performs better than the maximum likelihood estimation with sample sizes usually encountered in practice. Efficiency is a key concept for the asymptotics of the estimators, and in general there is no particular need for another method since the MLE would be the best choice. However, in many nonlinear models, like the mixture cure models, the asymptotic approximation is poor and the efficiency becomes a less relevant purpose for real data sample sizes. Hence, we choose to trade efficiency for better performance in a wider range of applications. }
		\subsection{Verification of assumptions for $\hat\pi$}
		Next we show that our assumptions 
		(AN1)-(AN4) of the asymptotic theory  are satisfied for  the nonparametric estimator $\hat\pi$ defined in \eqref{def:hat_pi} and the logistic model in \eqref{eqn:logistic}. For reasons of simplicity, since we use results available in the literature only for a one-dimensional covariate,  we consider only cases with one continuous covariate. In order for assumption (AN4) to be satisfied we need the following conditions:
		\begin{itemize}
			\item[(C1)] The bandwidth $b$ is such that $nb^4\to 0$ and $nb^{3+\xi}/(\log b^{-1}) \to \infty$ for some $\xi>0$.
			\item[(C2)] The support $\X$ of $X$ is a compact subset of $\R$. The density $f_X(\cdot)$ of $X$ is bounded away from zero and twice differentiable with bounded second derivative.
			\item[(C3)] The kernel $k$ is a twice continuously differentiable, symmetric  probability density function with compact support and $\int uk(u)\,\dd u=0$.
			\item[(C4)] (i) The functions $H([0,t]|x)$, $H_1([0,t]|x)$ are twice differentiable with respect to $x$, with uniformly bounded derivatives for all $t\leq \tau_0$, $x\in\X$. Moreover, there exist continuous nondecreasing functions $L_1$, $L_2$, $L_3$ such that $L_i(0)=0$, $L_i(\tau_0)<\infty$ and 	for all $t, s\in[0,\tau_0]$, $x\in\X$,
			\[
			\begin{split}
				\left|H_c(t|x)-H_c(s|x)\right|\leq \left|L_1(t)-L_1(s)\right|,&\quad \left|H_{1c}(t|x)-H_{1c}(s|x)\right|\leq \left|L_1(t)-L_1(s)\right|\\
				\left|\frac{\partial H_c(t|x)}{\partial x}-\frac{\partial H_c(s|x)}{\partial x}\right|&\leq \left|L_2(t)-L_2(s)\right|\\
				\left|\frac{\partial H_{1c}(t|x)}{\partial x}-\frac{\partial H_{1c}(s|x)}{\partial x}\right|&\leq \left|L_3(t)-L_3(s)\right|,
			\end{split}
			\]
			where the subscript c denotes the continuous part of a function.
			
			(ii) The jump points for the distribution function $G(t|x)$ of the censoring times given the covariate, are {finite and} the same for all $x$. The partial derivative of $G(t|x)$ with respect to $x$ exists and is uniformly bounded for all $t\leq\tau_0$, $x\in\X$. Moreover, 
			the partial derivative with respect to $x$ of $F(t|x)$ {(distribution function  of the survival times $T$ given $X=x$)}  exists and is uniformly bounded for all $t\leq\tau_0$, $x\in\X$.   
			\item[(C5)]{The survival time $T$ and the censoring time $C$ are independent given $X$.}
		\end{itemize}
		{(C1) to (C5) are conditions guaranteeing the rates of convergence and the i.i.d. representation  \citep{DA2002}}.
		In case of discrete covariates we also need to have only a finite number of atoms.  Assumption (C5) is needed because we are dealing with the distribution of $T$ conditional only on the covariate $X$ (since the cure rate depends only on $X$).
		
		\begin{theorem}
			\label{theo:pi_hat}
			Under the conditions (C1)-(C5), the assumptions (AN1)-(AN4) hold true for the logistic model {and the estimator $\hat\pi(x)$ defined in (\ref{def:hat_pi}).}
		\end{theorem}

		\section{Simulation study}
		\label{sec:simulations}
		In this section we focus on the logistic/Cox mixture cure model and evaluate the finite sample performance of the proposed method. Comparison is made with the maximum likelihood estimator implemented in the package \texttt{smcure}. 
		
		{
	We first illustrate through a brief example the convergence problems of the smcure estimator. 	We consider a model where the incidence depends on four independent covariates: $X_1\sim N(0,2)$, $X_2\sim \text{Uniform}(-1,1)$, $X_3\sim\text{Bernoulli}(0.8)$, $X_4\sim\text{Bernoulli}(0.2)$. The latency depends on $Z_1=X_1$, $Z_2=X_3$ and $Z_3=X_4$.  We generate the cure status $B$ as a Bernoulli random variable with success probability  $\phi(\gamma,X)$ where $\phi$ is the logistic function and $\gamma=(0.6,-1,1,2.5,1.2)$. The survival times for the uncured observations are generated according to a Weibull proportional hazards model
	\[
	S_u(t|z)=\exp\left(-\mu t^\rho\exp(\beta'z)\right),
	\]
	and are truncated at $\tau_0=14$ for $\rho=1.75$, $\mu=1.5$, $\beta=(-0.8,0.9,0.5)$. The censoring times are independent from $X$ and $T$. They are generated from the exponential distribution with parameter $\lambda_C=0.22$ and are truncated at $\tau=16$. We generate $1000$ datasets according to this model with sample size $n=100$, and we observe that smcure fails to converge in $43\%$ of the cases. Convergence fails mainly in the $\gamma$ parameter, with only $17\%$ of the cases failing to converge also for the $\beta$ parameter (because of the unreasonable $\gamma$ estimators). On the other hand, there was no convergence problem in the second step of the presmoothing approach. In addition, even among the cases where smcure converged, the presmoothing approach showed significantly better behavior, as can be seen in Table~\ref{tab:convergence}. 
	}
\begin{table}[h!]
	\caption{\label{tab:convergence}Bias, variance and MSE of $\hat\gamma$ and $\hat\beta$ for \texttt{smcure}  and our approach  among the iterations that converged for smcure.}
	\centering
	\scalebox{0.85}{
		\fbox{
			\begin{tabular}{rrrrrrr}
				&	& & & & &  \\[-7pt]
		&	\multicolumn{3}{c}{presmoothing}&\multicolumn{3}{c}{smcure}\\
			Par. &  Bias & Var. & MSE & Bias & Var. & MSE \\[2pt]
				\hline
				& & & & & &  \\[-8pt]
				 $\gamma_1 $ & $-0.113 $  & $0.620 $ & $0.633  $ & $0.200  $ & $8.318 $ & $8.358 $ \\
			$\gamma_2$	& $-0.073 $  & $0.156$ & $ 0.162 $ & $-0.388 $ & $3.085$ & $3.236 $\\
			 $\gamma_3 $ & $-0.071$  & $0.546$ & $0.551 $ & $ 0.280$ & $1.957 $ & $2.035 $ \\
				$\gamma_4$ & $ 0.037$  & $1.326$ & $1.327 $ & $0.704$ & $14.395 $ & $14.891 $ \\
					$\gamma_5$ & $ -0.250$  & $8.398$ & $8.461 $ & $1.621$ & $36.450 $ & $36.945 $ \\
				 $\beta_1 $ & $-0.014 $  & $0.011 $ & $0.012  $ & $-0.017  $ & $ 0.012$ & $ 0.012$\\
				$\beta_2$ & $0.024 $  & $0.064 $ & $ 0.065 $ & $ 0.026 $ & $0.065 $ & $0.065$ \\
				$\beta_3$ & $-0.053 $  & $0.165 $ & $ 0.168 $ & $- 0.053 $ & $0.166 $ & $0.169$ \\	
			\end{tabular}
	}}
\end{table}

{Hence, in the cases in which smcure exhibits very poor behavior, the presmoothing is obviously superior. Next, we focus on models for which smcure behaves reasonable (there are convergence problems in less then $3\%$ of the cases) and show that, even in such scenarios presmoothing can lead to more accurate results. }

		We consider four different models and for each of them various choices of the parameters in order to cover a wide range of scenarios. 
		The models are as follows. 
		
		\textit{Model 1.} Both incidence and latency depend on one covariate $X$, which is uniform on $(-1,1)$. We generate the cure status $B$ as a Bernoulli random variable with success probability  $\phi(\gamma,X)$ where $\phi$ is the logistic function. The survival times for the uncured observations are generated according to a Weibull proportional hazards model
		\[
		S_u(t|x)=\exp\left(-\mu t^\rho\exp(\beta x)\right),
		\]
		and are truncated at $\tau_0$ for $\rho=1.75$, $\mu=1.5$, $\beta=1$ and $\tau_0=4$. The censoring times are independent from $X$ and $T$. They are generated from the exponential distribution with parameter $\lambda_C$ and are truncated at $\tau=6$. 
		
		\textit{Model 2.} Both incidence and latency depend on one covariate $X$ with standard normal distribution. The cure status  and the survival times for the uncured observations are generated as in Model 1 for $\rho=1.75$, $\mu=1.5$, $\beta=1$ and $\tau_0=10$. The censoring times are  generated according to a Weibull proportional hazards model
		\[
		S_C(t|x)=\exp\left(-\nu\mu t^{\rho}\exp(\beta_Cx)\right),
		\] 
		for $\beta_C=1$ and various choices of $\nu$ and are truncated at $\tau=15$. 
		
		\textit{Model 3.} For the incidence we consider three independent covariates: $X_1$ is normal with mean zero and standard deviation $2$, $X_2$ and $X_3$ are Bernoulli random variables with parameters $0.6 $ and  $0.4$ respectively. The latency also depends on three covariates: $Z_1=X_1$, $Z_2$ is a uniform random variable on $(-3,3)$ independent of the previous ones  and $Z_3=X_2$.  The cure status  and the survival times for the uncured observations are generated as in Model 1 for $\rho=1.75$, $\mu=1.5$ and different choices of the other parameters.  The censoring times are  generated independently of the previous variables from an exponential distribution with parameter $\lambda_C$ 
		and are truncated at $\tau$, for given choices of  $\lambda_C$ and $\tau$. 
		
		\textit{Model 4.} This setting is obtained by adding an additional continuous covariate to the incidence component of Model 3. To be precise,  $X_1$ is normal with mean zero and standard deviation $2$, $X_2$ is uniform on $(-1,1)$ {independent of the other variables}, $X_3$ and $X_4$ are Bernoulli random variables with parameters $0.6 $ and  $0.4$ respectively. As in Model 3, $Z_1=X_1$, $Z_2$ is a uniform random variable on $(-3,3)$ independent of the previous ones  and $Z_3=X_3$. The event and censoring times are generated as in the previous model.
		
		For the four models we choose the values of the unspecified parameters in such a way that the cure rate is around $20\%$, $30\%$, $50\%$ (corresponding respectively to scenarios $1$, $2$ and $3$) and the censoring rate corresponds to three levels (with a difference of $5\%$ between each other). The specification of the parameters and the corresponding censoring rate and percentage of the observations in the plateau are given Table~\ref{tab:models}. {Note that, within each scenario, the fraction of the observations in the plateau decreases as the censoring rate increases because more cured observations are censored earlier and as a result are not observed in the plateau. This makes the estimation of the cure rate more difficult.} The truncation of the  survival  and censoring times on $[0,\tau_0]$ and $[0,\tau]$ is made in such a way that $\tau_0<\tau$ and  condition \eqref{eqn:jump_cond} is satisfied { but in practice
			it is unlikely to observe event times at $\tau_0$. In this way, we try to find a compromise between theoretical assumptions and  real-life scenarios.}
		
		\begin{table}
			\caption{	\label{tab:models}Parameter values and model characteristics for each scenario.}
			\centering
			\addtolength{\tabcolsep}{-4pt}
			\fbox{%
				\begin{tabular}{ccccccc}
					Model & Parameters & Scenario & Cens.  & Cens. & Cens.  & Plateau\\
					&  &  &  level & parameters &  rate & \\
					\hline
					& & & & & &  \\[-10pt]
					&  &  & $1$  & $\lambda_C=0.1 $ & $25\%$ & $15\% $\\
					& $\gamma=(1.75, 2)$ & $1 $ & $2$  & $\lambda_C=0.2 $ & $30\%$ & $11\% $\\
					&  &  & $3$  & $\lambda_C=0.3 $ & $35\%$ & $9\% $\\
					\cline{2-7}
					& &  & $1$  & $\lambda_C=0.1 $ & $34\%$ & $22\% $\\
					1 & $\gamma=(1, 1.5)$ & $2$ & $2$  & $\lambda_C=0.25 $ & $40\%$ & $15\% $\\
					&  &  & $3$  & $\lambda_C=0.4 $ & $46\%$ & $10\% $\\
					\cline{2-7}
					& &  & $1$  & $\lambda_C=0.2 $ & $54\%$ & $32\% $\\
					& $\gamma=(0.1, 5)$ & $3 $ & $2$  & $\lambda_C=0.4 $ & $59\%$ & $23\% $\\
					&  &  & $3$  & $\lambda_C=0.7 $ & $65\%$ & $15\% $\\
					\hline
					& & & & & &  \\[-10pt]
					&  &  & $1$  & $\nu=1/15 $ & $25\%$ & $7\% $\\
					& $\gamma=(1.5, 0.5)$ & $1$ & $2$  & $\nu=1/7 $ & $30\%$ & $4\% $\\
					&  &  & $3$  & $\nu=1/4 $ & $35\%$ & $2\% $\\
					\cline{2-7}
					& &  & $1$  & $\nu=1/13 $ & $35\%$ & $14\% $\\
					2 & $\gamma=(1, 1)$ & $2 $ & $2$  & $\nu=1/10 $ & $40\%$ & $9\% $\\
					&  &  & $3$  & $\nu=5/18 $ & $45\%$ & $6\% $\\
					\cline{2-7}
					& &  & $1$  & $\nu=1/9 $ & $56\%$ & $38\% $\\
					& $\gamma=(-0.1, 5)$ & $3 $ & $2$  & $\nu=1/4$ & $60\%$ & $30\% $\\
					&  &  & $3$  & $\nu=2/5 $ & $65\%$ & $25\% $\\
					\hline
					& & & & & &  \\[-10pt]
					& $\gamma=(0.5, -1,2.5, 1.2)$ &  & $1$  & $\lambda_C=0.12$ & $25\%$ & $10\% $\\
					& $\beta=( -1,0.5, 1.5)$ & $1$ & $2$  & $\lambda_C=0.25 $ & $30\%$ & $6\% $\\
					& $\tau_0=30$, $\tau=35$ &  & $3$  & $\lambda_C=0.45 $ & $35\%$ & $4\% $\\
					\cline{2-7}
					& $\gamma=(1,2,1.8, 0.5)$&  & $1$  & $\lambda_C=0.2 $ & $35\%$ & $16\% $\\
					3 & $\beta=(1,0.5, 2)$ & $2 $ & $2$  & $\lambda_C=0.5 $ & $40\%$ & $9\% $\\
					& $\tau_0=6$, $\tau=8$ &  & $3$  & $\lambda_C=0.8 $ & $45\%$ & $6\% $\\
					\cline{2-7}
					&$\gamma=(-0.8,1.3,1.5,-0.2) $&  & $1$  & $\lambda_C=0.3 $ & $55\%$ & $24\% $\\
					& $\beta=(1,-0.1,0.8)$ & $3 $ & $2$  & $\lambda_C=0.7 $ & $59\%$ & $14\% $\\
					& $\tau_0=5$, $\tau=7$ &  & $3$  & $\lambda_C=1.3 $ & $65\%$ & $8\% $\\
					\hline
					& & & & & &  \\[-10pt]
					& $\gamma=(0.6,-1,1,2.5,1.2)$ &  & $1$  & $\lambda_C=0.1$ & $25\%$ & $11\% $\\
					& $\beta=( -0.8,0.3,0.5)$ & $1$ & $2$  & $\lambda_C=0.22 $ & $30\%$ & $7\% $\\
					& $\tau_0=14$, $\tau=16$ &  & $3$  & $\lambda_C=0.35 $ & $35\%$ & $5\% $\\
					\cline{2-7}
					& $\gamma=(0.45,0.5,2{,}1,0.5)$&  & $1$  & $\lambda_C=0.15 $ & $35\%$ & $11\% $\\
					4 & $\beta=(1,0.5, 2)$ & $2 $ & $2$  & $\lambda_C=0.35 $ & $40\%$ & $7\% $\\
					& $\tau_0=18$, $\tau=20$ &  & $3$  & $\lambda_C=0.6 $ & $45\%$ & $5\% $\\
					\cline{2-7}
					&$\gamma=(-0.22,0.3,-0.4,0.5,-0.2) $&  & $1$  & $\lambda_C=0.2 $ & $55\%$ & $30\% $\\
					& $\beta=(0.4,-0.1,0.5)$ & $3 $ & $2$  & $\lambda_C=0.4 $ & $59\%$ & $20\% $\\
					& $\tau_0=6$, $\tau=8$ &  & $3$  & $\lambda_C=0.7 $ & $65\%$ & $12\% $\\
				\end{tabular}
			}
		\end{table}
		
		\begin{table}
			\caption{\label{tab:results1_2}Bias, variance and MSE of $\hat\gamma$ and $\hat\beta$ for \texttt{smcure} (second rows) and our approach (first rows) in Model 1 and 2.}
			\centering
			\scalebox{0.85}{
				\fbox{
					\begin{tabular}{ccccrrrrrrrrr}
						&	& & & & & & & & & && \\[-7pt]
						&	&&&\multicolumn{3}{c}{Cens. level 1}&\multicolumn{3}{c}{Cens. level 2}&\multicolumn{3}{c}{Cens. level 3}\\
						Mod.&	n&  scen. & Par. &  Bias & Var. & MSE & Bias & Var. & MSE & Bias & Var. & MSE\\[2pt]
						\hline
						&	& & & & & & & & & && \\[-8pt]
						1&	$200$ & $1 $ & $\gamma_1 $ & $0.001 $  & $0.060 $ & $0.060  $ & $0.020  $ & $0.065 $ & $0.065 $ & $0.005  $ & $0.078 $ & $0.078 $\\
						&	& &  & $0.021 $  & $0.063 $ & $ 0.063 $ & $0.050  $ & $0.068 $ & $0.071 $ & $  0.044$ & $0.084 $ & $0.086 $\\
						&	& & $\gamma_2 $ & $-0.034 $  & $0.164 $ & $0.165  $ & $ -0.014 $ & $0.202 $ & $0.202 $ & $-0.051  $ & $0.209 $ & $0.212 $\\
						&	& & & $ 0.026$  & $0.173 $ & $0.173  $ & $ 0.067 $ & $0.222 $ & $0.226 $ & $0.044  $ & $0.229 $ & $0.230 $\\
						&	& & $\beta $ & $0.008 $  & $0.028 $ & $0.028  $ & $0.015  $ & $ 0.029$ & $ 0.029$ & $0.013  $ & $0.034 $ & $0.035 $\\
						&	& &  & $0.007 $  & $0.028 $ & $ 0.028 $ & $ 0.012 $ & $0.029 $ & $0.029 $ & $0.009  $ & $0.035 $ & $0.035 $\\
						\cline{3-13}
						&	& & & & & & & & & && \\[-8pt]
						&	& $3 $ & $\gamma_1 $ & $-0.001 $  & $0.059 $ & $0.059  $ & $0.009  $ & $ 0.065$ & $ 0.065$ & $ -0.014 $ & $0.091 $ & $ 0.092$\\
						&	& &  & $0.010 $  & $0.064 $ & $0.064  $ & $0.029  $ & $0.074 $ & $0.075 $ & $0.037  $ & $0.113 $ & $0.115 $\\
						&	& & $\gamma_2 $ & $-0.034 $  & $0.536 $ & $0.537  $ & $-0.111  $ & $ 0.595$ & $0.608 $ & $ -0.085 $ & $ 0.809$ & $0.816 $\\
						&	& & & $0.201 $  & $ 0.649$ & $ 0.689 $ & $0.218  $ & $0.768 $ & $0.816 $ & $0.400  $ & $1.146 $ & $1.306 $\\
						&	& & $\beta $ & $ 0.011$  & $0.090 $ & $ 0.090 $ & $ 0.024 $ & $ 0.109$ & $0.110 $ & $0.014  $ & $0.128 $ & $0.128 $\\
						&	& &  & $ 0.007$  & $0.091 $ & $ 0.091 $ & $0.014  $ & $0.110 $ & $0.110 $ & $-0.001  $ & $0.129 $ & $0.129 $\\
						\cline{2-13}
						&	& & & & & & & & & && \\[-8pt]
						&	$400$ & $1 $ & $\gamma_1 $ & $0.001 $  & $0.028 $ & $0.028  $ & $ 0.007 $ & $0.032 $ & $0.032 $ & $0.001  $ & $0.037 $ & $0.037 $\\
						&	& &  & $0.015 $  & $0.029 $ & $0.030  $ & $0.027  $ & $0.033 $ & $ 0.034$ & $0.024  $ & $0.039 $ & $0.039 $\\
						&	& & $\gamma_2 $ & $-0.024 $  & $0.083 $ & $ 0.084 $ & $-0.004  $ & $ 0.088$ & $ 0.088$ & $ -0.018 $ & $0.107 $ & $0.107 $\\
						&	& & & $ 0.021$  & $0.087 $ & $0.087  $ & $ 0.049 $ & $0.093 $ & $0.095 $ & $0.041  $ & $0.111 $ & $0.113 $\\
						&	& & $\beta $ & $0.003 $  & $0.013 $  & $0.013  $ & $0.007 $ & $0.015 $ & $0.015  $ & $0.002 $ & $ 0.016$ & $0.016$\\
						&	& &  & $0.002 $  & $0.013 $ & $0.013  $ & $0.005  $ & $0.015 $ & $0.015 $ & $ 0.000 $ & $ 0.016$ & $0.016 $\\
						\cline{3-13}
						&	& & & & & & & & & && \\[-8pt]
						&	& $3 $ & $\gamma_1 $ & $-0.004 $  & $0.029 $ & $0.029  $ & $-0.004  $ & $0.030 $ & $ 0.030$ & $-0.007  $ & $0.048 $ & $ 0.048$\\
						&	& &  & $ 0.002$  & $0.030 $ & $0.030  $ & $ 0.009 $ & $0.033 $ & $0.033 $ & $0.015 $ & $0.053 $ & $0.053 $\\
						&	& & $\gamma_2 $ & $-0.050 $  & $ 0.237$ & $ 0.239 $ & $-0.080  $ & $0.312 $ & $0.318 $ & $-0.134  $ & $0.432 $ & $ 0.450$\\
						&	& & & $0.111 $  & $0.260 $ & $0.273  $ & $ 0.142 $ & $0.361 $ & $ 0.381$ & $ 0.167 $ & $0.491 $ & $0.519 $\\
						&	& & $\beta $ & $-0.003 $  & $0.039 $ & $0.039  $ & $0.024  $ & $0.051 $ & $0.052 $ & $ 0.013 $ & $0.071 $ & $0.071 $\\
						&	& &  & $-0.007 $  & $0.039 $ & $ 0.039 $ & $ 0.017 $ & $0.051 $ & $0.052 $ & $0.000  $ & $0.071 $ & $0.071 $\\
						\hline
						&	& & & & & & & & & && \\[-8pt]
						2&	$200$ & $1 $ & $\gamma_1 $ & $0.004 $  & $0.040 $ & $0.040  $ & $0.020  $ & $0.045 $ & $0.045 $ & $-0.016  $ & $0.060 $ & $0.060 $\\
						&	& &  & $0.017 $  & $ 0.040$ & $ 0.040 $ & $ 0.058 $ & $0.047 $ & $0.050 $ & $ 0.083 $ & $0.079 $ & $0.086 $\\
						&	& & $\gamma_2 $ & $ 0.001$  & $0.039 $ & $0.039  $ & $- 0.022 $ & $0.042 $ & $0.043 $ & $-0.027  $ & $ 0.055$ & $0.056 $\\
						&	& & & $0.016 $  & $0.040 $ & $0.040  $ & $ 0.008 $ & $0.047 $ & $0.047 $ & $0.029 $ & $0.072 $ & $0.073 $\\
						&	& & $\beta $ & $0.006 $  & $ 0.011$ & $0.011  $ & $ 0.000 $ & $ 0.014$ & $ 0.014$ & $0.011  $ & $0.015 $ & $0.015 $\\
						&	& &  & $0.005 $  & $ 0.011$ & $0.011  $ & $ -0.002 $ & $0.014 $ & $ 0.014$ & $0.004  $ & $0.016 $ & $0.016 $\\
						\cline{3-13}
						&	& & & & & & & & & && \\[-8pt]
						&	& $3 $ & $\gamma_1 $ & $-0.016 $  & $0.071 $ & $ 0.071 $ & $ -0.057 $ & $0.065 $ & $0.068 $ & $ -0.139 $ & $ 0.083$ & $ 0.102$\\
						&	& &  & $ 0.029$  & $0.092 $ & $0.092  $ & $ 0.051 $ & $0.119 $ & $0.121 $ & $0.024  $ & $0.175 $ & $0.176 $\\
						&	& & $\gamma_2 $ & $-0.468 $  & $0.723 $ & $0.942  $ & $ -0.943 $ & $0.823 $ & $ 1.713$ & $-1.348  $ & $0.829 $ & $ 2.646$\\
						&	& & & $0.364 $  & $ 0.926$ & $ 1.058 $ & $ 0.495 $ & $ 1.453$ & $ 1.698$ & $0.596  $ & $2.128 $ & $2.482 $\\
						&	& & $\beta $ & $0.017 $  & $0.035 $ & $ 0.035 $ & $0.022  $ & $0.039 $ & $0.039 $ & $ 0.036 $ & $0.052 $ & $0.054 $\\
						&	& &  & $ 0.014$  & $ 0.035$ & $0.035  $ & $ 0.017 $ & $0.040 $ & $0.040 $ & $0.025  $ & $0.053 $ & $0.054 $\\
						\cline{2-13}
						&	& & & & & & & & & && \\[-8pt]
						&	$400$ & $1 $ & $\gamma_1 $ & $0.011 $  & $0.019 $ & $0.019  $ & $0.019  $ & $0.023 $ & $ 0.023$ & $0.002  $ & $0.032 $ & $0.032 $\\
						&	& &  & $0.018 $  & $0.019 $ & $0.019  $ & $ 0.037 $ & $0.023 $ & $0.025 $ & $0.047  $ & $0.034 $ & $0.036 $\\
						&	& & $\gamma_2 $ & $-0.002 $  & $0.018 $ & $ 0.018 $ & $-0.010  $ & $0.023 $ & $ 0.023$ & $ -0.019 $ & $0.027 $ & $0.028 $\\
						&	& & & $0.009 $  & $0.018 $ & $ 0.018 $ & $ 0.007 $ & $0.025 $ & $ 0.025$ & $ 0.008 $ & $ 0.032$ & $0.032 $\\
						&	& & $\beta $ & $ 0.000$  & $0.006 $ & $ 0.006 $ & $0.004  $ & $0.006 $ & $ 0.006$ & $  0.003$ & $ 0.008$ & $0.008 $\\
						&	& &  & $0.000 $  & $ 0.006$ & $0.006  $ & $ 0.002 $ & $0.006 $ & $0.006 $ & $ 0.000 $ & $0.008 $ & $0.008 $\\
						\cline{3-13}
						&	& & & & & & & & & && \\[-8pt]
						&	& $3 $ & $\gamma_1 $ & $-0.015 $  & $0.031 $ & $0.031  $ & $-0.071  $ & $0.034 $ & $0.039 $ & $-0.086  $ & $ 0.041$ & $0.048 $\\
						&	& &  & $0.014 $  & $0.037 $ & $0.038  $ & $ 0.001 $ & $0.050 $ & $0.050 $ & $ 0.047 $ & $0.072 $ & $0.074 $\\
						&	& & $\gamma_2 $ & $-0.444 $  & $ 0.330$ & $ 0.527 $ & $-0.802  $ & $0.410 $ & $ 1.053$ & $ -1.191 $ & $0.463 $ & $1.881 $\\
						&	& & & $0.149 $  & $ 0.364$ & $0.386  $  & $ 0.244$ & $0.557 $ & $0.616  $ & $0.325 $ & $ 0.739$& $ 0.845 $\\
						&	& & $\beta $ & $0.007 $  & $0.016 $ & $0.016  $ & $0.015  $ & $0.019 $ & $0.020 $ & $0.017  $ & $ 0.024$ & $0.024 $\\
						&	& &  & $0.004 $  & $ 0.016$ & $0.016  $ & $ 0.010 $ & $0.019 $ & $0.020 $ & $0.010  $ & $ 0.024$ & $0.024 $\\
					\end{tabular}
			}}
		\end{table}
		For each setting we consider samples of size $n=200, 400, 1000$. This leads to a total of $108$ settings ($4$ models, $3$ scenarios for the cure rate, $3$ censoring levels and $3$ sample sizes). In this way, we hope to address a number  of issues such as the effect of the cure proportion, the sample size, amount and type of censoring, covariates (number, relation between $X$ and $Z$ and their distribution).
		For each configuration $1020$  datasets were generated and the estimators of $\beta_0$ and $\gamma_0$ were computed through \texttt{smcure} and our method. We report the bias, variance and mean squared error (MSE) of the estimators, computed after omitting the lowest and the highest $1\%$ of the estimators (for stability of the reported results) and rounded to three decimals. Tables~\ref{tab:results1_2}-\ref{tab:results4} show some of the results, while the rest can be found in the online Supplementary Material. {We aim to provide a ready-to-use method that works well in practice without needing to think about how to choose the kernel function or the bandwidth. Hence, we illustrate the performance of the method for some standard and commonly used choices. } The kernel function $k$ is taken to be the Epanechnikov kernel $k(u)=(3/4)(1-u^2)\1_{\{|u|\leq1\}}
		$. 
		We use the  cross-validation bandwidth  (implemented in the R package \texttt{np}) for kernel estimators of conditional distribution functions, in our case for estimation of $H=H_0+H_1$ given the continuous covariates (affecting the incidence). In addition, we restrict to the interval  $[0,Y_{(m)}]$, where $Y_{(m)}$ is the last observed event time since the estimator of the cure probability $\hat\pi$ in \eqref{def:hat_pi} is essentially a product over values of $t$ that are equal to the observed event times. This means that we use the {cross-validation} bandwidth for estimation of the conditional distribution $H(t|x)$ for $t\leq Y_{(m)}$. This choice of bandwidth improves significantly the performance of the estimators, compared to the cross-validation bandwidth on the whole interval $[0,\tau]$, in situations with a large percentage  of observations in the plateau, while it leads to little difference otherwise. 
		\begin{table}
			\caption{	\label{tab:results3}Bias, variance and MSE of $\hat\gamma$ for \texttt{smcure} (second rows) and our approach (first rows) in Model 3.}
			\centering
			\scalebox{0.85}{
				\fbox{
					\begin{tabular}{ccccrrrrrrrrr}
						&	& & & & & & & & & && \\[-8pt]
						&	&&&\multicolumn{3}{c}{Cens. level 1}&\multicolumn{3}{c}{Cens. level 2}&\multicolumn{3}{c}{Cens. level 3}\\
						Mod.&	n&  scen. & Par. &  Bias & Var. & MSE & Bias & Var. & MSE & Bias & Var. & MSE\\[2pt]
						\hline
						&	& & & & & & & & & && \\[-8pt]
						3&	$200$ & $1 $ & $\gamma_1 $ & $ 0.025$  & $0.147 $ & $0.147  $ & $0.010  $ & $0.192 $ & $0.192 $ & $ -0.008 $ & $0.243 $ & $0.243 $\\
						&	& &  & $ 0.034$  & $0.147 $ & $0.148  $ & $0.034  $ & $0.191 $ & $ 0.192$ & $0.062  $ & $0.249 $ & $0.253 $\\
						&	& & $\gamma_2 $ & $-0.045 $  & $0.042 $ & $ 0.044 $ & $ -0.078 $ & $0.049 $ & $ 0.055$ & $-0.085  $ & $0.059 $ & $ 0.066$\\
						&	& & & $-0.077 $  & $ 0.050$ & $ 0.056 $ & $ -0.122 $ & $0.065 $ & $0.080 $ & $-0.148  $ & $0.092 $ & $0.144 $\\
						&	& & $\gamma_3 $ & $0.081 $  & $0.366 $ & $ 0.373 $ & $ 0.074 $ & $ 0.485$ & $ 0.491$ & $0.029  $ & $0.536 $ & $ 0.537$\\
						&	& & & $0.174 $  & $ 0.397$ & $0.427  $ & $ 0.266 $ & $0.574 $ & $0.644 $ & $0.309  $ & $ 0.799$ & $0.895 $\\
						&	& & $\gamma_4 $ & $ -0.046$  & $0.326 $ & $ 0.373 $ & $-0.160  $ & $ 0.412$ & $0.437 $ & $ -0.289 $ & $ 0.453$ & $ 0.537$\\
						&	& & & $0.087 $  & $0.366 $ & $ 0.374 $ & $ 0.089 $ & $0.528 $ & $0.535 $ & $0.186  $ & $ 0.908$ & $0.943 $\\
						\cline{3-13}
						&	& & & & & & & & & && \\[-8pt]
						&	& $3 $ & $\gamma_1 $ & $-0.059 $  & $ 0.161$ & $ 0.164 $ & $  -0.091$ & $ 0.258$ & $0.266 $ & $ -0.223 $ & $ 0.419$ & $0.468 $\\
						&	& &  & $-0.053 $  & $ 0.163$ & $ 0.166 $ & $-0.071  $ & $0.261 $ & $0.266 $ & $-0.138  $ & $0.524 $ & $ 0.543$\\
						&	& & $\gamma_2 $ & $0.018 $  & $0.046 $ & $  0.046$ & $0.026  $ & $ 0.063$ & $0.064 $ & $0.086  $ & $0.088 $ & $ 0.096$\\
						&	& & & $0.080 $  & $ 0.052$ & $ 0.058 $ & $0.121  $ & $0.080 $ & $0.095 $ & $0.252  $ & $0.170 $ & $0.233 $\\
						&	& & $\gamma_3 $ & $0.060 $  & $0.235 $ & $0.238  $ & $0.076  $ & $ 0.366$ & $0.372 $ & $ 0.135 $ & $0.517 $ & $ 0.535$\\
						&	& & & $0.091 $  & $0.242 $ & $0.251  $ & $0.135  $ & $ 0.375$ & $ 0.393$ & $0.228  $ & $0.642 $ & $0.694 $\\
						&	& & $\gamma_4 $ & $-0.030 $  & $0.202 $ & $ 0.203 $ & $ -0.040 $ & $0.292 $ & $ 0.293$ & $ -0.081 $ & $0.479 $ & $0.486 $\\
						&	& & & $-0.027 $  & $ 0.205$ & $0.205  $ & $ -0.017 $ & $0.277 $ & $0.277 $ & $ -0.037 $ & $0.534 $ & $ 0.535$\\
						\cline{2-13}
						&	& & & & & & & & & && \\[-8pt]
						&	$400$ & $1 $ & $\gamma_1 $ & $ 0.016$  & $0.074 $ & $0.074  $ & $0.021  $ & $0.091 $ & $ 0.092$ & $ 0.003 $ & $ 0.128$ & $0.128 $\\
						&	& &  & $0.017 $  & $ 0.072$ & $0.073  $ & $ 0.022 $ & $0.082 $ & $0.082 $ & $ 0.023 $ & $0.108 $ & $0.108 $\\
						&	& & $\gamma_2 $ & $ -0.026$  & $0.019 $ & $0.019  $ & $ -0.039 $ & $0.023 $ & $ 0.025$ & $-0.070  $ & $0.032 $ & $ 0.037$\\
						&	& & & $-0.042 $  & $0.020 $ & $0.021  $ & $ -0.049 $ & $0.025 $ & $0.027 $ & $-0.081  $ & $0.035 $ & $ 0.041$\\
						&	& & $\gamma_3 $ & $0.039 $  & $0.194 $ & $ 0.195 $ & $0.028  $ & $0.219 $ & $ 0.220$ & $ 0.026 $ & $0.298 $ & $ 0.298$\\
						&	& & & $0.093 $  & $ 0.190$ & $ 0.198 $ & $ 0.097 $ & $0.206 $ & $0.215 $ & $0.158  $ & $0.297 $ & $ 0.322$\\
						&	& & $\gamma_4 $ & $ 0.010$  & $0.171 $ & $ 0.171 $ & $-0.091  $ & $0.193 $ & $0.201 $ & $ -0.178 $ & $0.276 $ & $0.307 $\\
						&	& & & $ 0.070$  & $0.177 $ & $ 0.182 $ & $0.038  $ & $0.198 $ & $0.200 $ & $0.088  $ & $0.289 $ & $0.297 $\\
						\cline{3-13}
						&	& & & & & & & & & && \\[-8pt]
						&	& $3 $ & $\gamma_1 $ & $-0.023 $  & $0.089 $ & $0.089  $ & $-0.051  $ & $0.118 $ & $0.121 $ & $-0.124  $ & $0.212 $ & $0.228 $\\
						&	& &  & $-0.029 $  & $0.092 $ & $ 0.093 $ & $-0.032  $ & $0.112 $ & $ 0.113$ & $-0.062  $ & $0.200 $ & $0.204 $\\
						&	& & $\gamma_2 $ & $0.003 $  & $0.023 $ & $  0.023$ & $0.006  $ & $0.033 $ & $0.033 $ & $0.042  $ & $0.048 $ & $0.050 $\\
						&	& & & $0.042 $  & $0.023 $ & $ 0.025 $ & $ 0.057 $ & $0.034 $ & $0.037 $ & $ 0.104 $ & $ 0.055$ & $0.066 $\\
						&	& & $\gamma_3 $ & $0.010 $  & $0.113 $ & $0.113  $ & $ 0.042 $ & $ 0.166$ & $ 0.168$ & $0.090  $ & $0.276 $ & $0.284 $\\
						&	& & & $0.039 $  & $0.111 $ & $0.111  $ & $ 0.060 $ & $0.152 $ & $0.156 $ & $ 0.108 $ & $0.250 $ & $ 0.262$\\
						&	& & $\gamma_4 $ & $0.012 $  & $0.110 $ & $ 0.110 $ & $ -0.021 $ & $ 0.131$ & $ 0.131$ & $-0.047  $ & $0.220 $ & $0.223 $\\
						&	& & & $0.014 $  & $ 0.111$ & $0.111  $ & $ -0.018 $ & $0.117 $ & $ 0.118$ & $-0.020  $ & $0.183 $ & $0.183 $\\
					\end{tabular}
			}	}
		\end{table}
		
		Simulations show that, for not large sample size, the new method performs better than \texttt{smcure} for estimation of $\gamma_0$, mostly because of a smaller variance.
		As  the sample size increases, they tend {to behave quite similarly.} 
		On  the other hand, both methods give almost the same estimates for $\beta_0$ and $\Lambda$. 
		The most favorable situation for our method is when there is little censoring among uncured observations and 
		the censored uncured observations are in the region of covariates that corresponds to higher cure rate. This comes from the fact that the nonparametric estimator in \eqref{def:hat_pi} takes larger values when the  product has more terms equal to one. This should not be a problem when we expect that subjects with high probability of being cured correspond to longer survival times, meaning that it is more probable for them to be censored compared to those with small cure probability and shorter survival times.
		
		\begin{table}
			\caption{	\label{tab:results4}Bias, variance and MSE of $\hat\gamma$ for \texttt{smcure} (second rows) and our approach (first rows) in Model 4.}
			\centering
			\scalebox{0.85}{
				\fbox{
					\begin{tabular}{ccccrrrrrrrrr}
						&	& & & & & & & & & && \\[-8pt]
						&	&&&\multicolumn{3}{c}{Cens. level 1}&\multicolumn{3}{c}{Cens. level 2}&\multicolumn{3}{c}{Cens. level 3}\\
						Mod.&	n&  scen. & Par. &  Bias & Var. & MSE & Bias & Var. & MSE & Bias & Var. & MSE\\[2pt]
						\hline
						&	& & & & & & & & & && \\[-8pt]
						4&	$200$ & $1 $ & $\gamma_1 $ & $ 0.041$  & $ 0.157$ & $0.159   $ & $ 0.016 $ & $ 0.187 $ & $0.188  $ & $-0.010 $ & $ 0.210 $ & $0.210 $\\
						&	& &  &   $0.077 $  & $0.178 $ & $ 0.184  $ & $ 0.096 $ & $0.228  $ & $ 0.238 $ & $ 0.127$ & $0.285  $ & $0.301 $\\
						&	& & $\gamma_2 $ &   $ -0.017$  & $0.039 $ & $ 0.039  $ & $ -0.019 $ & $ 0.042 $ & $ 0.042 $ & $ -0.015$ & $0.049  $ & $ 0.049$\\
						&	& &  &   $-0.090 $  & $0.052 $ & $ 0.060  $ & $ -0.125 $ & $  0.069$ & $ 0.085 $ & $ -0.164$ & $ 0.108 $ & $ 0.135$\\
						&	& & $\gamma_3 $   & $ -0.245$  & $0.159 $ & $0.219   $ & $ -0.281 $ & $  0.165$ & $ 0.244 $ & $ -0.355$ & $0.179  $ & $0.305 $\\
						& & &  & $0.064 $  & $0.244 $ & $ 0.249  $ & $ 0.084 $ & $0.304  $ & $ 0.311 $ & $ 0.114$ & $ 0.395 $ & $0.408 $\\
						&	& & $\gamma_4 $   & $-0.068 $  & $0.331 $ & $ 0.336  $ & $-0.162  $ & $ 0.385 $ & $ 0.411 $ & $ -0.285$ & $ 0.443 $ & $0.524 $\\
						& & &  & $0.171 $  & $0.401 $ & $ 0.430 $ & $ 0.241 $ & $0.561  $ & $ 0.619 $ & $ 0.314$ & $ 0.842 $ & $ 0.941$\\
						&	& & $\gamma_5 $   & $-0.095 $  & $0.301 $ & $ 0.310  $ & $ -0.234 $ & $0.349  $ & $ 0.404 $ & $ -0.366$ & $ 0.371 $ & $ 0.505$\\
						& & &  & $ 0.106$  & $0.363 $ & $ 0.375  $ & $ 0.143 $ & $0.509  $ & $0.529  $ & $ 0.177$ & $ 0.693 $ & $ 0.724$\\
						\cline{3-13}
						&	& & & & & & & & & && \\[-8pt]
						&	& $3 $ &$\gamma_1 $ & $ -0.044$  & $0.079 $ & $  0.081 $ & $-0.079  $ & $  0.095$ & $ 0.101 $ & $-0.148 $ & $  0.132$ & $0.154 $\\
						& &  &  & $0.000 $  & $ 0.079$ & $ 0.079  $ & $ 0.003 $ & $ 0.096 $ & $ 0.096 $ & $0.009 $ & $ 0.141 $ & $ 0.141$\\
						&	& & $\gamma_2 $ &   $ 0.018$  & $0.007 $ & $  0.008 $ & $ 0.024 $ & $  0.008$ & $0.009  $ & $0.041 $ & $ 0.010 $ & $0.012 $\\
						& & &  & $0.015 $  & $0.008 $ & $ 0.008  $ & $ 0.017 $ & $ 0.009 $ & $ 0.010 $ & $0.028 $ & $ 0.013 $ & $0.014 $\\
						&	& & $\gamma_3 $ &   $0.034 $  & $0.066 $ & $  0.067 $ & $ 0.046 $ & $ 0.073 $ & $ 0.075 $ & $0.067 $ & $0.087  $ & $0.091 $\\
						& & &  & $-0.025 $  & $0.080 $ & $   0.080$  & $ -0.033 $ & $ 0.091 $ & $ 0.092 $ & $ -0.041$ & $0.120  $ & $ 0.122$\\
						&	& & $\gamma_4 $ &   $0.041 $  & $0.102 $ & $  0.104 $ & $  0.054$ & $ 0.125 $ & $0.128  $ & $0.082 $ & $ 0.166 $ & $0.173 $\\
						& & &  & $0.022 $  & $0.100 $ & $0.101   $ & $ 0.023 $ & $ 0.126 $ & $ 0.126 $ & $0.026 $ & $0.179  $ & $ 0.180$\\
						&	& & $\gamma_5 $ &   $-0.031 $  & $ 0.103$ & $ 0.104  $ & $ -0.034 $ & $0.120  $ & $ 0.121 $ & $ -0.054$ & $0.159  $ & $0.162 $\\
						& & &  &$ -0.016$  & $0.099 $ & $  0.099 $ & $ -0.013 $ & $ 0.115 $ & $ 0.115 $ & $ -0.018$ & $ 0.150 $ & $ 0.151$\\
						\cline{2-13}
						&	& & & & & & & & & && \\[-8pt]
						&	$400$ & $1 $ & $\gamma_1 $ & $ 0.013$  & $ 0.067$ & $ 0.067  $ & $ 0.015 $ & $ 0.079 $ & $0.080  $ & $ 0.005$ & $ 0.097 $ & $ 0.097$\\
						& &  &  & $0.024 $  & $0.067 $ & $   0.068$ & $ 0.037 $ & $  0.080$ & $0.082  $ & $ 0.043$ & $ 0.101 $ & $0.103 $\\
						&	& & $\gamma_2 $ &   $ -0.001$  & $ 0.017$ & $  0.017 $ & $  -0.003$ & $ 0.020 $ & $ 0.020 $ & $-0.007 $ & $ 0.023 $ & $0.023 $\\
						& & &  & $-0.042 $  & $0.020 $ & $   0.021$ & $ -0.055 $ & $0.025  $ & $ 0.028 $ & $-0.079 $ & $  0.034$ & $0.041 $\\
						&	& & $\gamma_3 $ &   $-0.229 $  & $0.089 $ & $  0.141 $ & $ -0.207 $ & $ 0.090 $ & $ 0.133 $ & $ -0.275$ & $ 0.102 $ & $0.178 $\\
						& & &  & $0.046 $  & $0.107 $ & $ 0.109  $ & $  0.061$ & $ 0.137 $ & $ 0.141 $ & $ 0.066$ & $ 0.162 $ & $ 0.166$\\
						&	& & $\gamma_4 $ &   $ -0.063$  & $ 0.161$ & $ 0.165  $ & $ -0.143 $ & $ 0.175 $ & $  0.196$ & $-0.222 $ & $ 0.215 $ & $ 0.265$\\
						& & &  & $ 0.085$  & $ 0.176$ & $  0.183 $ & $ 0.107 $ & $ 0.222 $ & $ 0.234 $ & $ 0.145$ & $0.318  $ & $0.339 $\\
						&	& & $\gamma_5 $ &  $ -0.075$  & $0.146 $ & $ 0.151  $ & $-0.192  $ & $ 0.177 $ & $  0.214$ & $-0.299 $ & $ 0.199 $ & $ 0.289$\\
						& & &  & $0.043 $  & $0.157 $ & $ 0.159  $ & $ 0.049 $ & $ 0.194 $ & $0.196  $ & $ 0.060$ & $ 0.253 $ & $ 0.257$\\
						\cline{3-13}
						&	& & & & & & & & & && \\[-8pt]
						&	& $3 $ & $\gamma_1 $ & $-0.024 $  & $ 0.038$ & $ 0.039  $ & $ -0.038 $ & $ 0.047 $ & $0.048  $ & $-0.092 $ & $0.064  $ & $0.073 $\\
						& &  &  & $-0.003 $  & $0.036 $ & $ 0.036  $ & $ 0.002 $ & $ 0.044 $ & $ 0.044 $ & $0.004 $ & $ 0.060 $ & $ 0.060$\\
						&	& & $\gamma_2 $ &   $ 0.006$  & $ 0.003$ & $ 0.003  $ & $0.010  $ & $ 0.004 $ & $ 0.004 $ & $ 0.019$ & $  0.005$ & $0.006 $\\
						& & &  & $0.005 $  & $ 0.004$ & $ 0.004  $ & $ 0.005 $ & $0.004  $ & $0.004  $ & $0.006 $ & $  0.006$ & $0.006 $\\
						&	& & $\gamma_3 $ &   $ 0.028$  & $0.033 $ & $  0.034 $ & $0.045  $ & $0.038  $ & $ 0.040 $ & $0.065 $ & $ 0.045 $ & $0.049 $\\
						& & &  & $ -0.010$  & $ 0.036$ & $ 0.036  $ & $ -0.008 $ & $ 0.042 $ & $ 0.042 $ & $ -0.008$ & $ 0.052 $ & $0.052 $\\
						&	& & $\gamma_4 $ &  $ 0.021$  & $ 0.052$ & $  0.053 $ & $0.032  $ & $ 0.062 $ & $ 0.063 $ & $0.060 $ & $ 0.088 $ & $ 0.092$\\
						& & &  & $0.015 $  & $ 0.050$ & $  0.050 $ & $ 0.016 $ & $0.060  $ & $ 0.060 $ & $0.019 $ & $ 0.083 $ & $ 0.083$\\
						&	& & $\gamma_5 $ &   $ -0.024$  & $0.049 $ & $ 0.050  $ & $  -0.041$ & $ 0.059 $ & $ 0.061 $ & $ -0.048$ & $ 0.077 $ & $ 0.079$\\
						& & &  & $-0.017 $  & $0.049 $ & $ 0.049  $ & $ -0.022 $ & $  0.056$ & $  0.057$ & $-0.022 $ & $ 0.069 $ & $0.069 $\\
					\end{tabular}
			}	}
		\end{table}
		
		This is indeed the case in Model 1 and we observe that our approach outperforms \texttt{smcure} in all the scenarios. The difference between the two is more marked when $n$ is small and the absolute value of the $\gamma$ coefficient is larger. In Model 2, the situation is more difficult because censoring depends on the covariate in such a way that, the non-cured subjects have the same probability of being censored independently of their cure probability. However, for the first two scenarios the new method is still superior. The third scenario is more problematic because  
		the cure probability drops very fast from almost one to almost zero, resulting in a large fraction of uncured observations with almost zero cure probability. The presence of censoring in this region leads to overestimation of the cure rate. If we would take $\beta_C=0.1$ (meaning larger probability of being censored for higher cure rate), then the new approach is significantly superior (see  Table \ref{tab:2_4} for $n=400$ and scenario 3). In Model 3,  complications arise because of the presence of different covariates for the incidence and latency. Hence, subjects with higher cure rate might correspond to shorter survival times. As a result, the previous problem might still happen and its effects are more visible for large sample size and large censoring rate. Finally, Model 4 suggests that, even though the assumptions in Section \ref{sec:asymptotics} were shown to be satisfied only for one continuous covariate, the method  {could be applied in more general cases.}
		We noticed that, when a continuous covariate affects only the incidence and not the latency, the bandwidth selected by the \texttt{np} package is often very large, meaning that it fails to capture the effect of this covariate on the conditional distribution function. In those cases, we truncate the selected bandwidth from above at $2$. Note that the bandwidth is chosen for standardized covariates so the truncation level can be fixed regardless of the distribution of the covariate. We decided to truncate at $2$ since it seems to be a kind of boundary for a `reasonable' bandwidth {with standardized covariates} (we do not want to externally affect chosen bandwidths smaller than $2$ but we only replace extremely large values by $2$). However, even when reasonable, the \texttt{np} bandwidth for $X_2$ seems to be larger than it should, resulting in more bias in the estimator of $\gamma_3$. Nevertheless, in terms of mean squared error, the method performs well  for not large sample size. If $X_2$ would affect also the latency, the selected bandwidth would be more adequate and there would be no bias problems.
		
		\begin{table}
			\caption{\label{tab:2_4}Bias, variance and MSE of $\hat\gamma$ and $\hat\beta$ for \texttt{smcure}  and our approach in Model 2, scenario 3 when $\beta_C=0.1$ and $n=400$.}
			\centering
			\scalebox{0.85}{
				\fbox{
					\begin{tabular}{crrrrrr}
						& \multicolumn{3}{c}{}  & \multicolumn{3}{c}{} \\[-8pt]
						&	 \multicolumn{3}{c}{\texttt{smcure} package}   & \multicolumn{3}{c}{Our approach}\\
						Parameter	& Bias & Var & MSE & Bias &Var & MSE\\[2pt]
						\cline{1-7}
						& & & & & & \\[-8pt]
						$\gamma_1$ & $0.014 $ & $0.123 $ & $0.123$   & $-0.058 $ & $ 0.103$ & $0.106$\\
						$\gamma_2$ & $0.418 $ & $1.243 $ & $1.418 $   & $-0.535 $ & $0.652 $ & $0.937 $\\
						$\beta$ & $0.001$ & $0.025$ & $0.025 $  & $0.001 $ & $0.027 $ & $ 0.027$\\
			\end{tabular}}}
		\end{table}
		
		To conclude,  the new approach seems to perform significantly better than \texttt{smcure} when  the sample size is not large and the fraction of censored observations is not much higher than the expected cure proportion. In other situations, both methods are comparable. However, one has to be more careful when there is no reason to expect that the censored subjects correspond to higher cure probabilities.
		
		{In the previous settings, we truncated the event times at $\tau_0$ in such a way that condition \eqref{eqn:jump_cond} is satisfied  but in practice
				it is unlikely to observe event times at $\tau_0$. Next, we consider one additional model for which condition \eqref{eqn:jump_cond} is not satisfied. The covariates and the parameters are as in Model 3 described above, but the event times are  generated from a Weibull distribution on $[0,\tau_0]$ with $\tau_0=15$, i.e.
		\[
		S_u(t|z)=\frac{\exp\left\{-\mu t^\rho\exp(\beta'z)\right\}-\exp\left\{-\mu\tau_0^\rho\exp(\beta'z)\right\}}{1-\exp\left\{-\mu\tau_0^\rho\exp(\beta'z)\right\}}
		\]
		The censoring times are exponentially distributed  as in Model 3 and truncated at $\tau=20$. Results for sample size $n=200$ and three censoring levels are shown in Table~\ref{tab:nojump}. Compared to Model 3 above, we observe that, when condition \eqref{eqn:jump_cond} is not satisfied, presmoothing is even more superior than the smcure estimator. 
				}
			\begin{table}[h!]
			\caption{	\label{tab:nojump}Bias, variance and MSE of $\hat\gamma$ for \texttt{smcure} (second rows) and our approach (first rows) in Model 3 without condition \eqref{eqn:jump_cond}.}
			\centering
			\scalebox{0.85}{
				\fbox{
					\begin{tabular}{crrrrrrrrr}
					 & & & & & & & && \\[-8pt]
					&\multicolumn{3}{c}{Cens. level 1}&\multicolumn{3}{c}{Cens. level 2}&\multicolumn{3}{c}{Cens. level 3}\\
						 Par. &  Bias & Var. & MSE & Bias & Var. & MSE & Bias & Var. & MSE\\[2pt]
						\hline
					 & & & & & & & && \\[-8pt]
					  $\gamma_1 $ & $ 0.015$  & $0.152 $ & $0.152  $ & $0.000  $ & $0.196 $ & $0.196$ & $ -0.032 $ & $0.246 $ & $0.247 $\\
						 & $ 0.017$  & $0.150 $ & $0.151  $ & $0.027  $ & $0.193 $ & $ 0.194$ & $0.035  $ & $0.260 $ & $0.262 $\\
					 $\gamma_2 $ & $-0.054 $  & $0.044 $ & $ 0.047 $ & $ -0.077 $ & $0.052 $ & $ 0.058$ & $-0.109  $ & $0.064 $ & $ 0.076$\\
					 & $-0.085$  & $ 0.050$ & $ 0.057 $ & $ -0.119 $ & $0.069 $ & $0.083 $ & $-0.171  $ & $0.101 $ & $0.130$\\
					 $\gamma_3 $ & $0.087 $  & $0.379 $ & $ 0.386 $ & $ 0.073 $ & $ 0.450$ & $ 0.456$ & $0.045  $ & $0.578 $ & $ 0.580$\\
						& $0.197 $  & $ 0.423$ & $0.462  $ & $ 0.249 $ & $0.561 $ & $0.623 $ & $0.343  $ & $ 0.885$ & $1.002 $\\
					 $\gamma_4 $ & $ -0.010$  & $0.339 $ & $ 0.339 $ & $-0.106 $ & $ 0.373$ & $0.385 $ & $ -0.228 $ & $ 0.498$ & $ 0.550$\\
						& $0.125 $  & $0.364 $ & $ 0.380 $ & $ 0.156 $ & $0.523 $ & $0.548 $ & $0.260  $ & $ 1.513$ & $1.581 $\\
							\end{tabular}
			}	}
		\end{table}
		
	Finally we conclude with a remark about the computational aspect.	{The proposed approach is computationally more intensive than the MLE mainly because of the bandwidth selection through a cross-validation procedure. For example,  for one iteration in Model 3 with sample size $200$ and $400$, \texttt{smcure} computes the estimates in $0.7$ and $0.8$ seconds respectively, while the new approach requires $4.1$ and $23.5$ seconds (with a Core i7-8665U CPU desktop).  However, this seems still reasonable since the method is not meant for much larger sample sizes.}

		\section{Application: Melanoma study}
		\label{sec:application}
		To illustrate the practical performance, we apply the proposed estimation procedure to two medical datasets for patients with melanoma and compare the results with \texttt{smcure}. Melanoma is the third most common skin cancer type with overall incidence rate 21.8 out of 100,000 people in the US (Cancer statistics from the Center for Disease Control and Prevention) and according to the American Cancer Society, $6850$ people are expected to die of melanoma in $2020$. However, in the recent years,  the chances of survival for melanoma patients have increased due to earlier diagnosis and improvement of treatment and surgical techniques. The 5-year survival rates based on the stage of the cancer when it was first diagnosed are $92\%$ for localized, $65\%$ for regional and $25\%$ for distant stage. It is also known that this disease is more common among white people and the death rate is higher for men  than women.  Even though most melanoma patients are cured by their initial treatment, it is not possible to distinguish them from the uncured patients. Hence, accurately estimating the probability of being cured is important in order to plan further treatment and prevent recurrence of uncured patients. 
		
		\subsection{Eastern Cooperative Oncology Group (ECOG) Data}
		We use the melanoma data (ECOG phase III clinical trial e1684) from the \texttt{smcure} package \cite{cai_smcure} in order to compare our results with those of \texttt{smcure}.
		The purpose of this study was to evaluate the effect of treatment (high dose interferon alpha-2b regimen) as the postoperative adjuvant therapy. The event time is the time from initial treatment to recurrence of melanoma and three covariates have been considered: age (continuous variable centered to the mean), gender (0=male and 1=female) and treatment (0=control and 1=treatment).
		The data consists of $284$ observations (after deleting missing data) out of which $196$ had recurrence of the melanoma cancer (around  $30\%$ censoring).  The Kaplan-Meier curve is shown in Figure~\ref{fig:KM_melanoma_1}. 
		The parameter estimates, standard errors and corresponding p-values for the Wald test using our method and the \texttt{smcure} package are given in Table~\ref{tab:melanoma1}. Standard errors are computed through  $500$ {naive} bootstrap samples.
		
		\begin{figure}
			\centering
			\makebox{
				\includegraphics[width=0.48\linewidth]{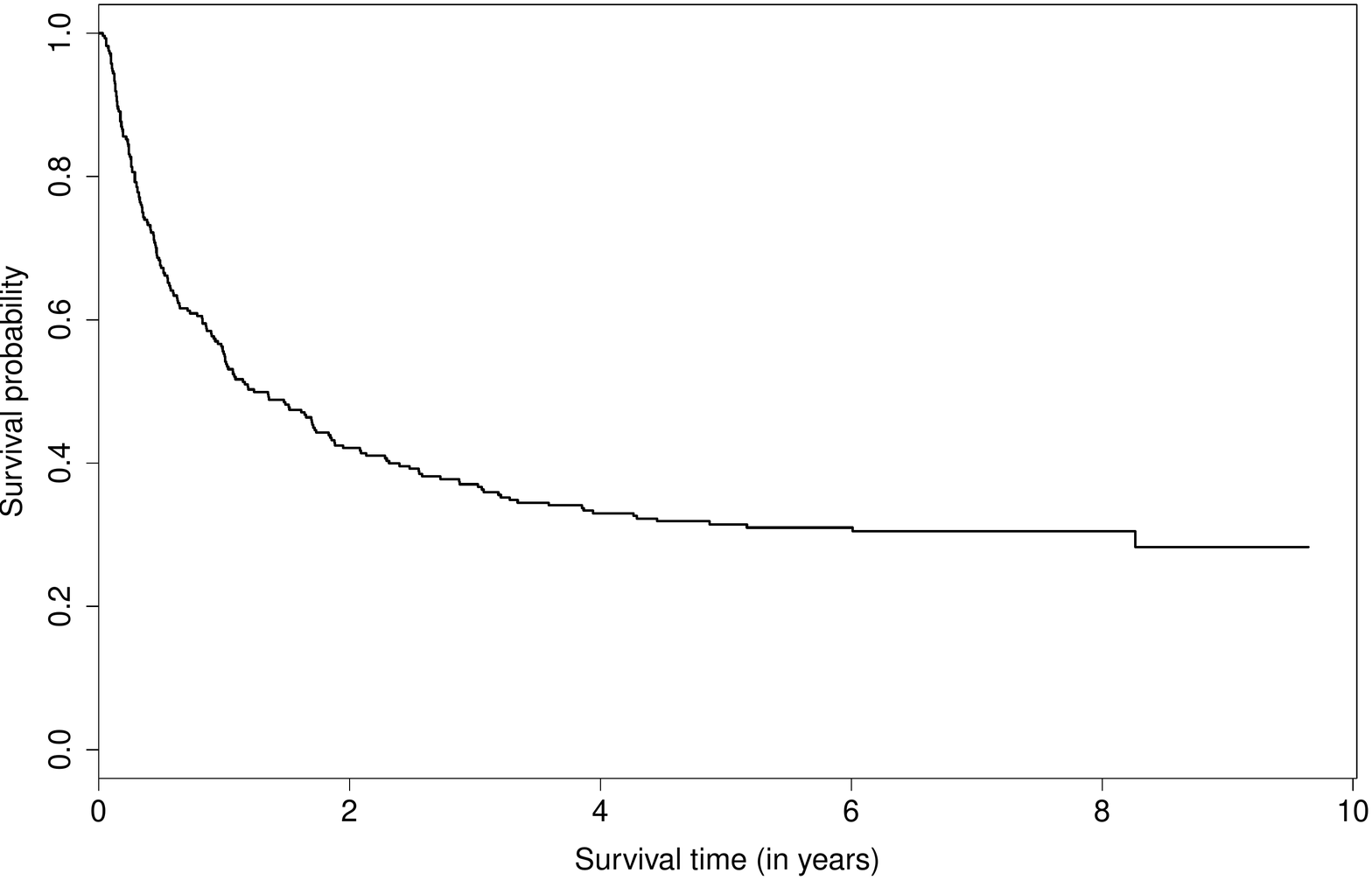}\qquad\includegraphics[width=0.48\linewidth]{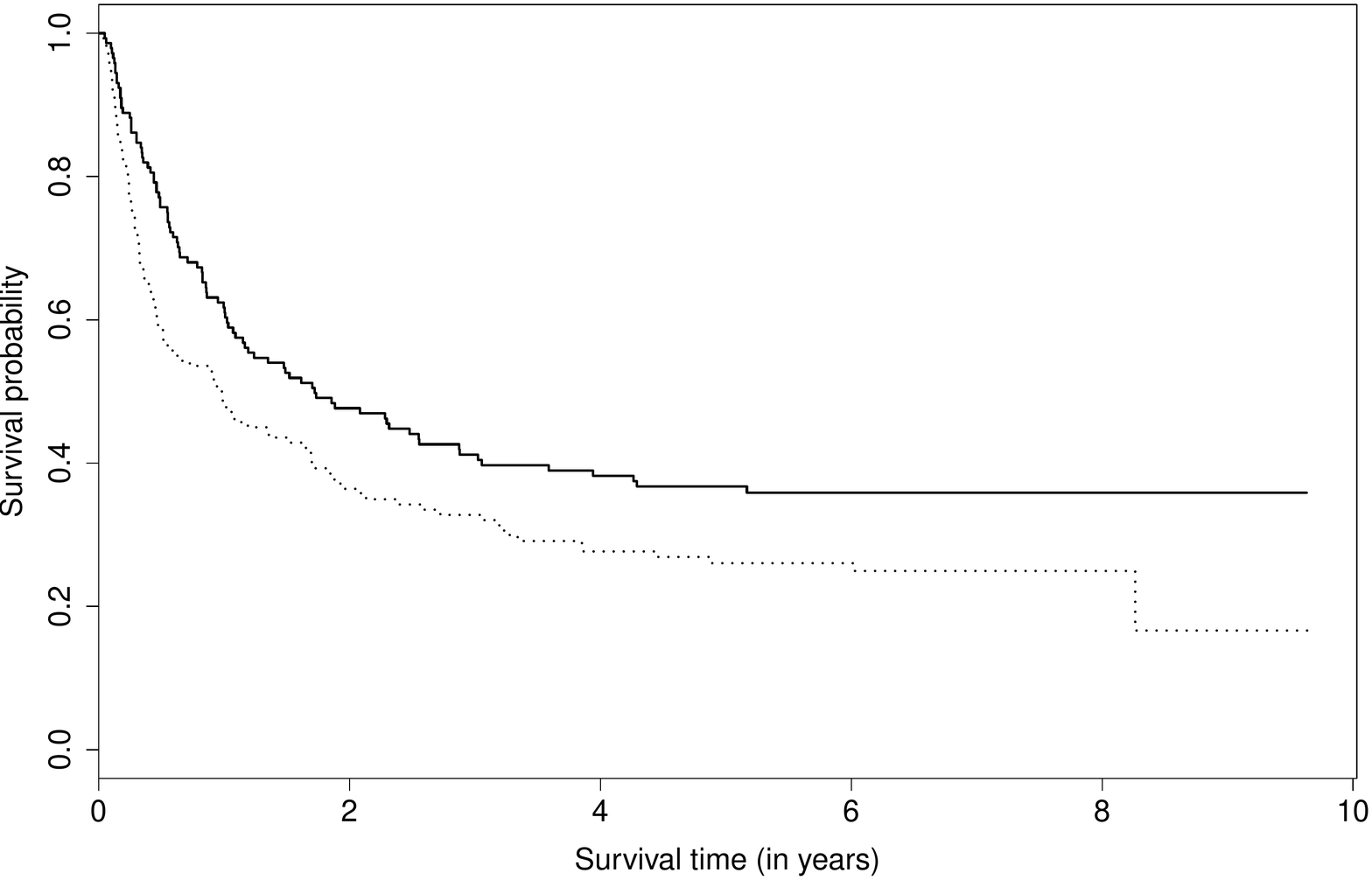}}
			\caption{\label{fig:KM_melanoma_1}Left panel: Kaplan-Meier survival curve for ECOG data. Right panel: Kaplan-Meier survival curves for the treatment group (solid) and control group (dotted) in the ECOG data.}
		\end{figure}
		
		\begin{table}
			\caption{\label{tab:melanoma1}Results for the incidence (logistic component) and the latency (Cox PH component) from the ECOG data.}
			\centering
			\scalebox{0.85}{
				\fbox{
					\begin{tabular}{c|crrrrrr}
						&& \multicolumn{3}{c}{}  & \multicolumn{3}{c}{} \\[-8pt]
						&	& \multicolumn{3}{c}{\texttt{smcure} package}   & \multicolumn{3}{c}{Our approach}\\
						&	Covariates	& Estimates & SE & p-value & Estimates & SE & p-value\\[2pt]
						\cline{2-8}
						&	& & & & & & \\[-8pt]
						\multirow{4}{*}{\STAB{\rotatebox[origin=c]{90}{incidence}}}	&Intercept & $1.3649 $ & $0.3457 $ & $8\cdot10^{-5} $   & $1.6697 $ & $ 0.3415$ & $10^{-6} $\\
						&	Age & $0.0203 $ & $0.0159 $ & $0.2029 $   & $0.0220 $ & $0.0104 $ & $0.0344 $\\
						&	Gender & $-0.0869 $ & $0.3347 $ & $0.7949 $  & $-0.3039 $ & $0.3448 $ & $ 0.3493$\\
						&	Treatment & $-0.5884 $ & $ 0.3706$ & $0.1123 $   & $-0.9345 $ & $0.3603 $ & $0.0095 $\\
						\hline
						&	& & & & & &   \\[-8pt]
						\multirow{3}{*}{\STAB{\rotatebox[origin=c]{90}{latency}}}	&	Age & $-0.0077 $ & $0.0069 $ & $0.2663 $  &  $-0.0079 $ & $0.0060 $ & $0.1861 $\\
						&	Gender & $0.0994 $ & $0.1932 $ & $0.6067 $   & $0.1240 $ & $0.1653 $ & $0.4534 $\\
						&	Treatment & $-0.1535 $ & $0.1715 $ & $0.3707 $  & $-0.0947 $ & $ 0.1692$ & $0.5756 $\\
			\end{tabular}}}
		\end{table}
		
		We observe that, for both methods, the effects of the covariates
		have the same direction. Only the intercept was found significant for the incidence with \texttt{smcure}, while our method concludes that  also age and treatment are  significant. In particular, the probability of recurring melanoma is higher for the control group compared to the treatment group. This seems to be indeed the case if we look at the Kaplan Meier survival curves for the two groups in Figure~\ref{fig:KM_melanoma_1}.  On the other hand, both methods agree that none of the covariates is significant for the latency. 
		
		{To illustrate another advantage of the new approach, we also compute the maximum likelihood estimator with the  \texttt{smcure} package for different choices of the latency model. We see in Table~\ref{tab:melanoma3} that the estimators of the incidence component (and their significance) change depending on which variables are included in the latency. On the other hand, the new method does not suffer from this problem because it estimates the incidence independently of the latency.
			\begin{table}
			\caption{\label{tab:melanoma3}Results for the incidence (logistic component) and the latency (Cox PH component) from the ECOG data.}
			\centering
			\scalebox{0.85}{
				\fbox{
					\begin{tabular}{c|c|rrr|rrr|rrr}
						&& \multicolumn{3}{c}{}  & \multicolumn{3}{c}{}  & \multicolumn{3}{c}{}\\[-8pt]
						&	& \multicolumn{3}{c}{Model 1}    & \multicolumn{3}{c}{Model 2 }& \multicolumn{3}{c}{Model 3 }\\
						&	Covariates	& Estimates & SE & p-value & Estimates & SE & p-value	& Estimates & SE & p-value\\[2pt]
						\cline{2-11}
						&	& & & & & & &&&\\[-8pt]
						\multirow{4}{*}{\STAB{\rotatebox[origin=c]{90}{incidence}}}	&Intercept & $1.3507 $ & $0.3001 $ & $7\cdot10^{-6} $   & $1.4148 $ & $ 0.3213$ & $10^{-5} $& $1.4181$ & $ 0.3073$ & $4\cdot10^{-6} $\\
						&	Age & $0.0164$ & $0.0125 $ & $0.1905 $   & $0.0205 $ & $0.0154 $ & $0.1803 $& $0.0209 $ & $ 0.0146$ & $0.1528 $\\
						&	Gender & $-0.0265 $ & $0.3113 $ & $0.9320 $  & $-0.0673 $ & $0.3352$ & $ 0.8407$& $-0.0222$ & $ 0.3130$ & $0.9432 $\\
						&	Treatment & $-0.6060$ & $ 0.3509$ & $0.0842 $   & $-0.6773 $ & $0.3223 $ & $0.0415 $& $-0.6913$ & $ 0.3439$ & $0.0444 $\\
						\hline
						&	& & & & & &  &&& \\[-8pt]
						\multirow{3}{*}{\STAB{\rotatebox[origin=c]{90}{latency}}}	&	Age &  &  &  &  $-0.0074 $ & $0.0066 $ & $0.2568 $& $-0.0073 $ & $ 0.0064$ & $0.2579 $\\
						&	Gender &&  &    & $0.0789 $ & $0.1863 $ & $0.6719$&  & & \\
						&	Treatment & $-0.1324 $ & $0.1561 $ & $0.3963$  &  &  & &  &  &\\
			\end{tabular}}}
		\end{table}
	}
		\subsection{Surveillance, Epidemiology and End Results (SEER) database}
		
		The SEER database collects cancer incidence data from population-based cancer registries in US. These data consist of patient demographic characteristics, primary tumor site, tumor morphology, stage at diagnosis, length of follow up and vital status. We select the database `Incidence - SEER 18 Regs Research Data' and extract the melanoma cancer data for the county of San Francisco in California during the period $2004-2015$. We consider only patients with stage at diagnosis: localized, regional and distant and exclude those with unknown or zero follow-up time  and restrict the study to white people because of the very small number of cases from other races. The event time is death because of melanoma.  This cohort consists of $1445$ melanoma cases out of which $596$ are female and $849$ male. The age ranges from $11 $ to $101$ years old, the follow-up from $1$ to $155$ months. For most of the patients the cancer has been diagnosed at early stage (localized), while for $101$ of them the stage at diagnosis is `regional' and only for $42$ it is `distant'. We aim at evaluating how age, gender and stage at diagnosis affect the survival  of melanoma patients in this cohort. The use of cure models is justified from the presence of a long plateau containing around $20\%$ of the observations (see the Kaplan-Meier curve in Figure \ref{fig:KM_melanoma_2}). Moreover, the Kaplan-Meier curves depending on gender  and stage at diagnosis  in Figure~\ref{fig:KM_melanoma_2} confirm that  gender and stage affect the cure rate. 
		
		\begin{figure}
			\begin{center}	\makebox{
					{\includegraphics[width=0.48\linewidth]{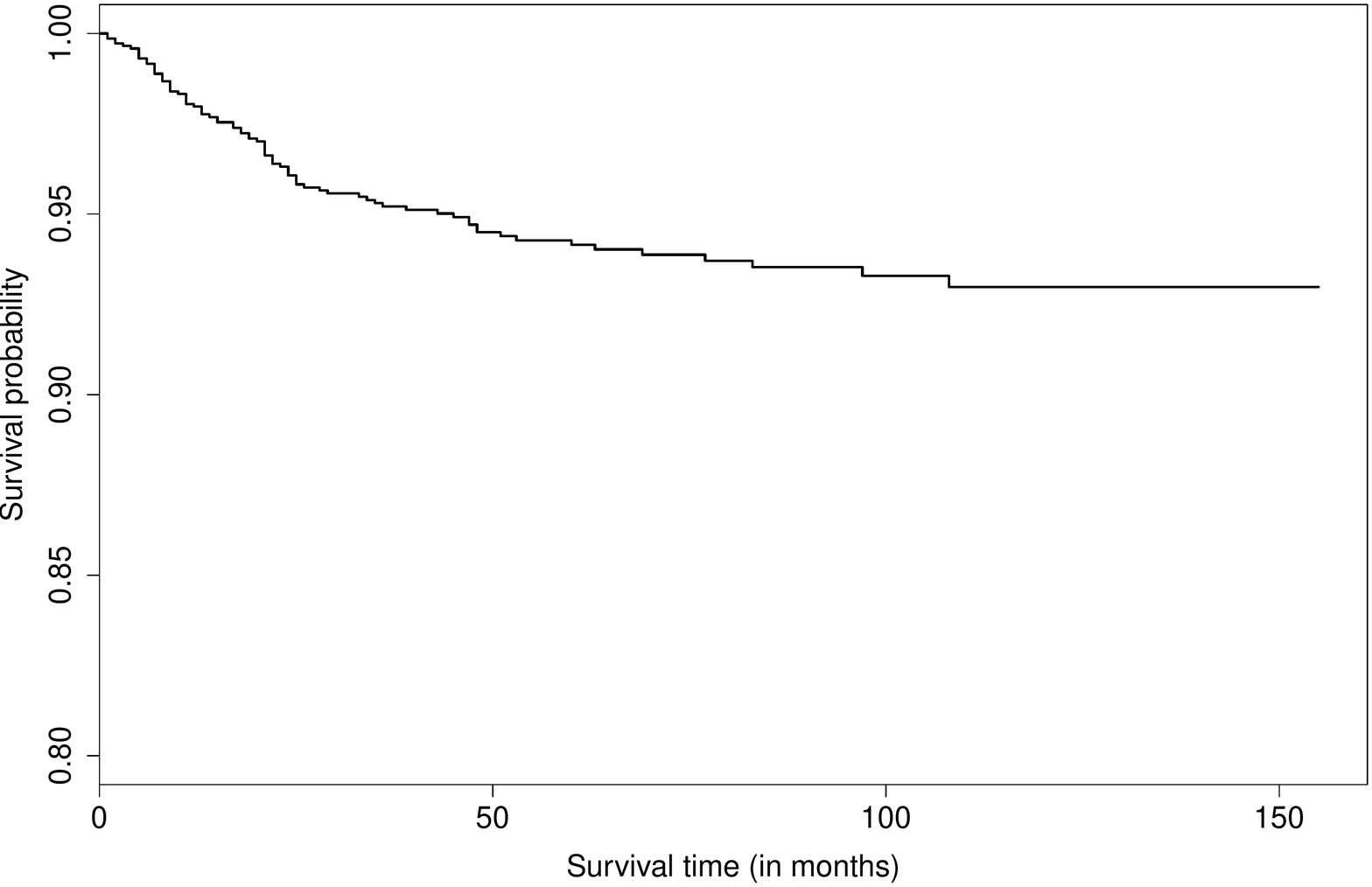}}
			}	\end{center}
			\makebox{
				{\includegraphics[width=0.48\linewidth]{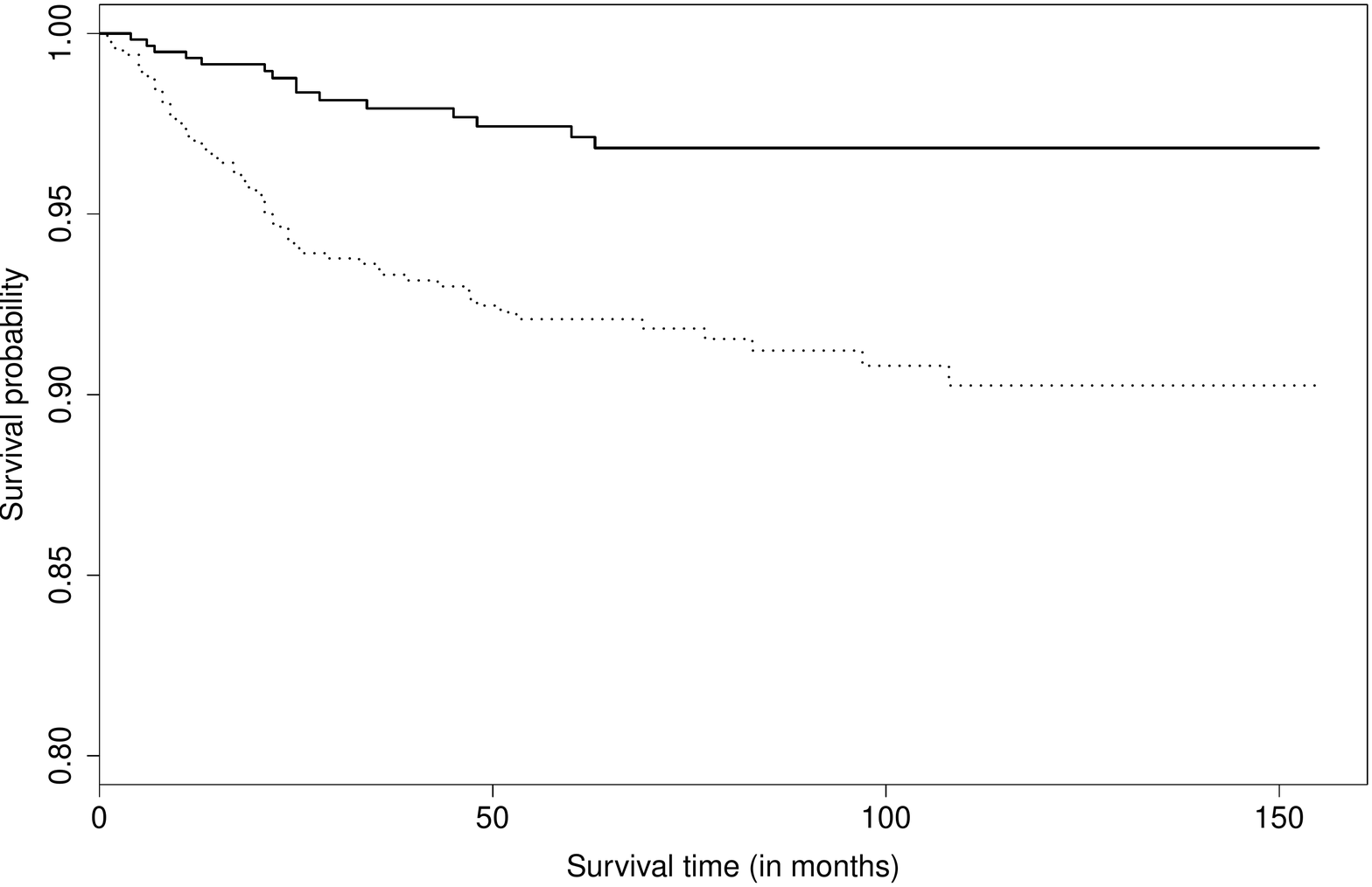}\qquad 	\includegraphics[width=0.48\linewidth]{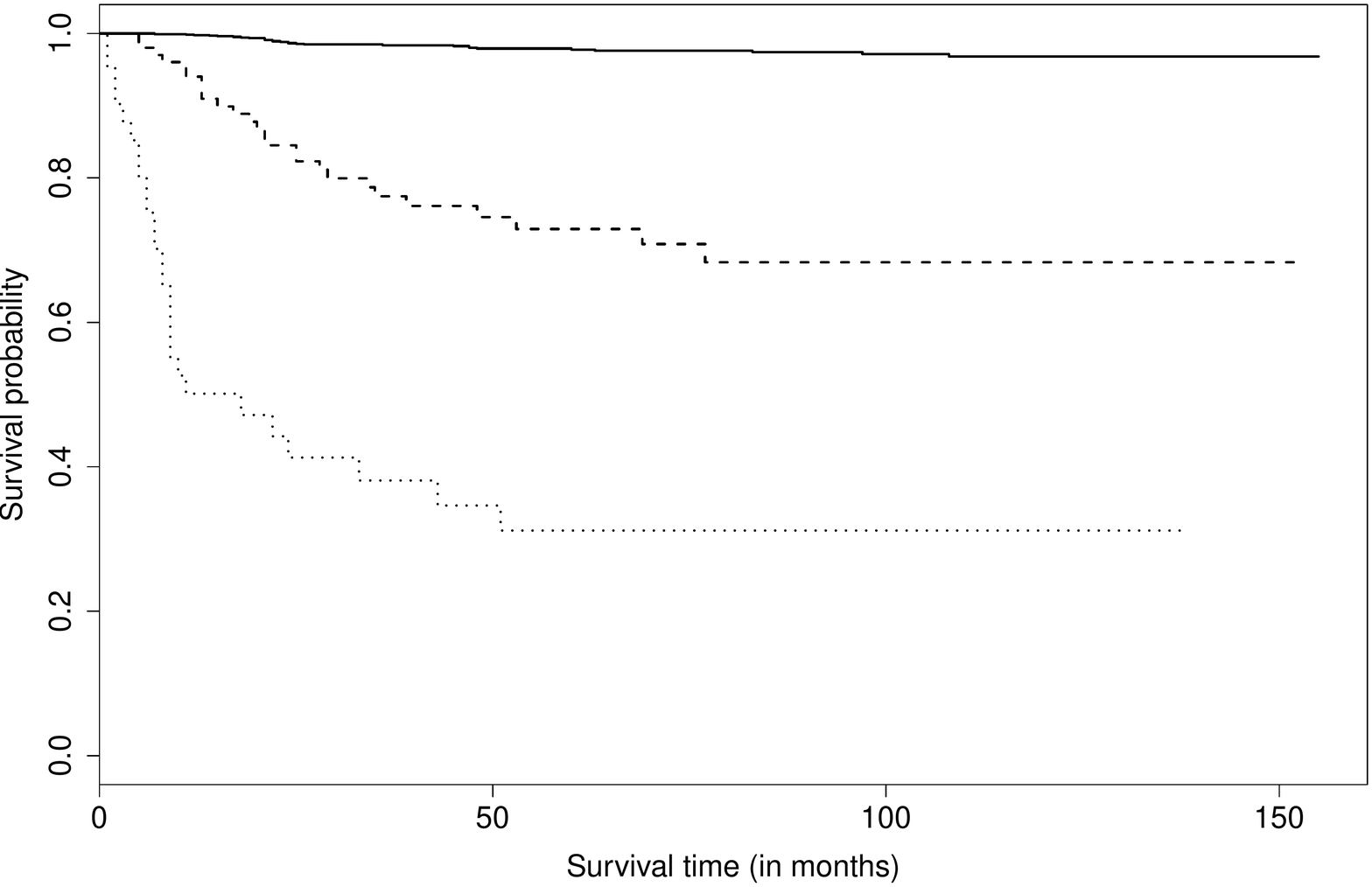}}}
			\caption{\label{fig:KM_melanoma_2}Upper panel: Kaplan-Meier survival curves for SEER data. Left panel: group division based on gender, females (solid) and males (dotted). Right panel: group division based on cancer stage at diagnosis, localized (solid), regional (dashed) and distant (dotted).}
		\end{figure}
		
		We checked the fit of the logistic model by comparing it with the single-index mixture cure model proposed in \cite{AKL19} through the prediction error of the incidence. More precisely, as in \cite{AKL19}, we divide the data into a training set and a test set of size $964$ and $481$ respectively. Using the training set, we estimate the logistic/Cox model and the single-index/Cox model. Afterwards, we compute the prediction error in the test set given by
		\[
		PE=-\sum_{j=1}^{481}\left\{\hat{W}_j\log [1-\hat\pi(X_j^{\text{test}})]+(1-\hat{W}_j)\log \hat\pi(X_j^{\text{test}})\right\}
		\]
		where $\hat\pi(X_j^{\text{test}})$ and $\hat{W_j}$ are the predicted cure probability and the predicted weight for the $j$th observation in the test set, computed based on the parameter estimates (and the link function for the single-index model) in the training set. {More precisely, for the logistic/Cox model we have $\hat\pi(X_j^{\text{test}})=\phi(\hat\gamma_n,X_j^{\text{test}})$ and 
			\[
			\hat{W}_j=\Delta_j^{\text{test}}+(1-\Delta_j^{\text{test}})\frac{\hat\pi(X_j^{\text{test}})\exp\left(-\hat\Lambda_n(Y_j^{\text{test}})e^{\hat\beta'_nZ_j^{\text{test}}}\right)}{1-\hat\pi(X_j^{\text{test}})+\hat\pi(X_j^{\text{test}})\exp\left(-\hat\Lambda_n(Y_j^{\text{test}})e^{\hat\beta'_nZ_j^{\text{test}}}\right)}
			\]
			where $\hat\gamma_n$, $\hat\beta_n$ and $\hat\Lambda_n$ are  the estimated parameters and the estimated hazard function in the training set. For the single-index/Cox model, the only difference is that $\hat\pi(X_j^{\text{test}})=\hat{g}_n(\hat\gamma_n,X_j^{\text{test}})$ where $\hat{g}_n$ is the estimated link function as in \cite{AKL19}. The weights $\hat W_j$ correspond to the conditional expectation of the cure status $B$ given the observations. } 
		We find that the prediction error for the logistic model is $98.53$, whereas for the single-index model it is $156.55$. This means that the logistic model performs better. 
		
		\begin{table}
			\caption{\label{tab:melanoma2}
				Results for the incidence (logistic component) and the latency (Cox PH component) from the  SEER data.}
			\centering
			\scalebox{0.85}{
				\fbox{
					\begin{tabular}{c|crrrrrr}
						&& \multicolumn{3}{c}{}  & \multicolumn{3}{c}{} \\[-8pt]
						&	& \multicolumn{3}{c}{\texttt{smcure} package}   & \multicolumn{3}{c}{Our approach}\\
						&	Covariates	& Estimates & SE & p-value & Estimates & SE & p-value\\[2pt]
						\cline{2-8}
						&	& & & & & & \\[-8pt]
						\multirow{5}{*}{\STAB{\rotatebox[origin=c]{90}{incidence}}}&	Intercept & $-4.2071 $ & $0.3817 $ & $0 $  & $-4.2436 $ & $0.3980 $ & $0 $ \\
						&	Age & $0.0304 $ & $0.0122 $ & $0.0124 $  & $0.0328 $ & $0.0172 $ & $ 0.0565$ \\
						&	Gender & $ 1.1318$ & $0.4211 $ & $0.0072 $  & $1.2341 $ & $0.4792 $ & $0.010 $ \\
						&	$S_1$& $2.6738 $ & $0.3702 $ & $5\cdot 10^{-13} $  & $ 2.4474$ & $0.4247 $ & $ 8\cdot 10^{-9}$ \\
						&	$S_2$ & $4.0763 $ & $0.5067 $ & $8\cdot 10^{-16} $  & $3.9426 $ & $0.4536 $ & $ 0$ \\
						\hline
						&	& & & & & &  \\[-8pt]
						\multirow{4}{*}{\STAB{\rotatebox[origin=c]{90}{latency}}}&	Age & $-0.0139 $ & $ 0.0098$ & $ 0.1577$  & $-0.0143 $ & $0.0106 $ & $0.1756 $ \\
						&	Gender & $-0.0549 $ & $0.4065 $ & $ 0.8925$  & $-0.0871 $ & $0.3687 $ & $0.8131 $ \\
						&	$S_1$& $0.5176 $ & $0.3993 $ & $0.1949$  & $0.6130 $ & $0.3971 $ & $0.1226 $ \\
						&	$S_2$ & $1.8039 $ & $0.4529 $ & $7\cdot 10^{-5} $  & $1.8623 $ & $0.5072 $ & $0.0002 $ \\
					\end{tabular}
			}}
		\end{table}
		
		The parameter estimates, standard errors and corresponding p-values for the Wald test using our method and the \texttt{smcure} package are given in Table \ref{tab:melanoma2}. Standard errors  are computed  with $500$ {naive} bootstrap samples. The covariate stage is classified using two dummy Bernoulli variables $S_1$ and $S_2$, where $S_1=1$ indicates the regional stage and
		$S_2=1$ indicates the distant stage.
		The gender variable is equal to zero for females and one for males. We observe that both methods agree that all the considered covariates are significant for the incidence (with age being a borderline case for our approach).  
		For the latency, only being in the distant stage is found significant with both methods. 
		Moreover, again the effects of all the covariates on the latency and incidence have the same direction for both methods.

		\section{Discussion} \label{sec:disc}
		In this paper we proposed a new estimation procedure for the mixture cure model with a parametric form of the incidence (for example logistic) and {any semiparametric model for the latency. We investigated more in detail the logistic/Cox model given its practical relevance}. Instead of using an iterative algorithm for dealing with the unknown cure status, this method relies on a preliminary nonparametric estimator of the cure probabilities.  We showed through simulations that the new approach improves upon the classical maximum likelihood estimator implemented in the package \texttt{smcure}, mainly for smaller sample sizes.  For the latency, both methods behave similarly. Hence, it is of particular interest in situations in which the focus is on the estimation of cure probabilities.  {The real data application on the ECOG clinical trial also showed that the improvement in estimation can be meaningful in practice and help detecting significant effects.}
		
		{ The proposed method has the advantage of direct estimation of the incidence component, without relying on the latency, which makes it robust to latency model misspecification. On the contrary, the \texttt{smcure} estimator  strongly depends on the choice of the variables for the latency and could be biased for a misspecified Cox  model. Hence, for practical reason, confronting the estimators obtained with the two methods is valuable for confirming the results or obtaining new insights.  }
		{From the theoretical point of view, unlike the standard maximum likelihood estimation,  presmoothing allows us  to obtain consistency and asymptotic normality without requiring the `unrealistic' assumption that the distribution of uncured subjects  has a positive mass at the end point of the support.  } 
		
		{It might be argued that since the proposed method relies on smoothing, it is more complex and the results can be affected by the choice of the kernel function or the bandwidth. Our purpose was to show that the user doesn't have to think about this  because the standard choices proposed in this paper perform well in practice. In addition, since the final estimator is a parametric one and the kernel estimator is only a preliminary step of the procedure, the results would anyway be more stable with respect to these choices than in a nonparametric setting. The main challenge this method faces is extension to many continuous covariates for the incidence. 		}
		We did not deeply investigate such situations since, in that case, multiple bandwidths have to be chosen, which can be more problematic and computationally intensive. However, our approach based on presmoothing allows to efficiently handle these situations if the estimator $\hat\pi$ is constructed in a more adequate way. {One possibility would be to construct the estimator assuming a single-index model for the latency, which is reasonable since the final goal is a parametric estimator. With this approach one can avoid the choice of multiple bandwidths and perform the estimation as in the one dimensional case.} However, this problem will be addressed by future research.   {In this regard, even though considering only one continuous covariate might seem restrictive in practice, the proposed procedure constitutes the basis for further developments of new estimators for general dimension scenarios that do not require multidimensional smoothing.}

		\section{Appendix}
		\label{sec:appendix}
		\subsection{Proof of Theorem \ref{theo:asymptotic_normality2}}
		\label{sec:appendix_Cox}
		We obtain the asymptotic normality of $\hat\Lambda_n$, $\hat\beta_n$ following the proof of Theorem 3 in \cite{Lu2008}. In order to work with a one-dimensional submodel, {for $d$ in a neighbourhood of the origin,} let   $\Lambda_d(t)=\int_0^t\{1+dh_1(s)\}\dd \hat\Lambda_n(s)$ and $\beta_d=dh_2+\hat\beta_n$, where $h_1$ is a function of bounded variation on $[0,{\tau_0}]$ and $h_2$ is a $q$-dimensional real vector. Let $\hat{S}_n(\hat\Lambda_n,\hat\beta_n)(h_1,h_2)$ denote the derivative of ${\hat{l}_n}
		(\Lambda_d,\beta_d)$ (defined in \eqref{def:hat_l_cox}) with respect to $d$ and evaluated at $d=0$. We have 
		\[
		\begin{split}
			\hat{S}_n(\hat\Lambda_n,\hat\beta_n)(h_1,h_2)
			&=\frac{1}{n}\sum_{i=1}^n \Delta_i\1_{\{Y_i<\tau_0\}} \left[h_1(Y_i)+h'_2Z_i\right]\\
			&\quad-\frac{1}{n}\sum_{i=1}^n \left\{\Delta_i+(1-\Delta_i){\1_{\{Y_i\leq\tau_0\}}}g_i(Y_i,\hat\Lambda_n,\hat\beta_n,\hat\gamma_n)\right\}\\
			&\qquad\qquad\qquad\quad\times\left\{e^{\hat\beta'_nZ_i}\int_0^{Y_i}h_1(s)\dd\hat\Lambda_n(s)+e^{\hat\beta'_nZ_i}\hat\Lambda_n(Y_i)h'_2Z_i
			\right\},
		\end{split}
		\]
		where $g_j$ is defined in \eqref{def:g_j} and $\hat\gamma_n$ is the maximizer of \eqref{def:hat_L_gamma}.  Let  $\Upsilon_n=(\hat\Lambda_n,\hat\beta_n)$ and $\Upsilon_0=(\Lambda_0,\beta_0)$. Furthermore, denote  by $S$ the asymptotic version of $\hat{S}_n$:
		\[
		\begin{split}
			S(\Lambda,\beta)(h_1,h_2)&=\E\bigg[ \Delta\1_{\{Y<\tau_0\}} \{h_1(Y)+h'_2Z\}-\left\{\Delta+(1-\Delta){\1_{\{Y\leq\tau_0\}}}g(Y,\Lambda,\beta,\gamma_0)\right\}\\
			&\qquad\qquad\qquad\qquad\left. \times\left\{e^{\beta'Z}\int_0^{Y}h_1(s)\dd\Lambda(s)+e^{\beta'Z}\Lambda(Y)h'_2Z
			\right\}\right].
		\end{split}
		\]
		We have $\hat{S}_n(\Upsilon_n)=0$ and $S(\Upsilon_0)=0$.  
		The score function $S_n$ and $S$ are respectively  a random  and a deterministic map from $\Xi$ to $l^{\infty}(\mathcal{H}_\mathfrak{m})$ (the space of bounded real-valued functions on $\mathcal{H}_\mathfrak{m}$), where
		\[
		\Xi=\left\{(\Lambda,\beta)\,:\,\sup_{h\in\mathcal{H}_\mathfrak{m}}\left|\int_0^{{\tau_0}}h_1(s)\dd\Lambda(s)+h'_2\beta\right|<\infty\right\}
		\] 
		and  $\mathcal{H}_\mathfrak{m}=\{h\in\mathcal{H}\,:\, \Vert h\Vert_{H}\leq \mathfrak{m}\}$. Here $\Vert h\Vert_{H}=\Vert h_1\Vert_v+\Vert h_2\Vert_{L_1}$, $\Vert h_2\Vert_{L_1}=\sum_{j=1}^q|h_{2,j}|$, $\Vert h_1\Vert_v=|h_1(0)|+V_0^{{\tau_0}}(h_1)$ and $V_0^{{\tau_0}}(h_1)$ denotes the total variation of $h_1$ on $[0,{\tau_0}]$. 
		This means that $S_n$ is a random variable defined in the abstract probability space $(\Omega,\mathcal{F},\p)$  (where the random vector $(B,T_0,C,X,Z)$ is defined) with values in the space of bounded functions $\Xi\mapsto l^\infty(\mathcal{H}_\mathfrak{m})$ with respect to the supremum norm. The latter one is a Banach space equipped with the Borel $\sigma$-field.
		
		We need to show that conditions 1-4 of Theorem 4 in \cite{Lu2008} (or Theorem 3.3.1 in \cite{VW96}) are satisfied. 
		The main difference of the function $S$ from the one in \cite{Lu2008} is that here $\gamma=\gamma_0$ fixed. We are only considering variation with respect to $\beta$ and not $\gamma$, so the components of $h$ that correspond to $\gamma$ are set to zero.
		However, conditions 2 and 3 of Theorem 4 in \cite{Lu2008}  for $S$ can be shown in the same way as in \cite{Lu2008}.  
		Details about conditions 1 and 4 can be found in the online Supplementary Material.
		\qed

		\subsection{Proof of Theorem \ref{theo:pi_hat}}
		The logistic model for the cure probability obviously satisfies assumptions (AN1) and  (AN3).  Let $\Pi$ be the space of  continuously differentiable functions $f$ from $\X$ to $[0,1]$ such that $\sup_{x\in\X}|f'(x)|\leq M$ and
		\[
		\sup_{x_1, x_2\in\X}\frac{|f'(x_1)-f'(x_2)|}{|x_1-x_2|^\xi}\leq M
		\]
		for some $M>0$ and $\xi\in(0,1]$.  If such space is equipped with the supremum norm, the covering numbers  satisfy 
		\[
		\log N\left(\epsilon,\Pi,\Vert\cdot\Vert_\infty\right)\leq K\frac{1}{\epsilon^{1/(1+\xi)}}
		\]
		for some constant $K>0$ independent of $\epsilon$ (see Theorem 2.7.1 in \cite{VW96}). Obviously, for $\epsilon>1$, $\log N(\epsilon,\Pi,\Vert\cdot\Vert_\infty)=0$. Hence, assumption (AN2) is satisfied.   
		It remains to check (AN4). Recall that the estimator of the cure probability $\hat\pi(x)$ is the value at time $\tau_0$ of the Beran estimator $\hat{S}(t|x)$, while $\pi_0(x)=S(\tau_0|x)$. Moreover, by assumption  \eqref{eqn:CI2}, we have $\inf_x H((\tau_0,\infty)|x)>0$. From Proposition 4.1 and 4.2 in \cite{KA99} it follows that
		\[
		\begin{aligned}
			\sup_x\left|\hat{\pi}(x)-\pi_0(x)\right|&=O\left((nb)^{-1/2}(\log b^{-1})^{1/2}\right)\quad a.s., \\
			\sup_x\left|\hat{\pi}'(x)-\pi_0'(x)\right|&=O\left((nb^3)^{-1/2}(\log b^{-1})^{1/2}\right)\quad a.s.
		\end{aligned}
		\]
		and
		\[
		\sup_{x_1, x_2\in\X}\frac{|\hat\pi'(x_1)-\pi'_0(x_1)-\hat\pi'(x_2)+\pi'_0(x_2)|}{|x_1-x_2|^{\xi/2}}=O\left(\left(nb^{3+\xi}\right)^{-1/2}(\log b^{-1})^{1/2}\right)\quad a.s.,
		\]
		where $\xi$ is as in assumption (C1). Since $\pi_0$ is twice continuously differentiable, from  assumption (C1) it follows that $\hat\pi$ satisfies (i,ii) of (AN4).
		From Theorem 3.2 of \cite{DA2002} (with $T=\tau_0$) we have $\hat{\pi}(x)-\pi_0(x)=\frac{1}{n}\sum_{i=1}^n A_i(x)+R_n(x)$, where
		\begin{equation}
			\label{eqn:iid_pi}
			\begin{split}
				A_i(x)=-\frac{1-\phi(\gamma_0,x)}{f_X(x)}\frac{1}{b}k\left(\frac{x-X_i}{b}\right)
				\left\{\frac{\Delta_i\1_{\{Y_i\leq\tau_0\}}}{H([Y_i,\infty)|x)}-\int_0^{Y_i\wedge \tau_0}\frac{H_1(ds|x)}{H^2([s,\infty)|x)}\right\}
			\end{split}
		\end{equation}
		and $\sup_x|R_n(x)|=O\left((nb)^{-3/4}(\log n)^{3/4}\right)$ a.s..  
		Hence
		\[
		\begin{split}
			&\E^*\left[\left(\hat\pi(X)-\pi_0(X)\right)\left(\frac{1}{\phi(\gamma_0,X)}+\frac{1}{1-\phi(\gamma_0,X)}\right)\nabla_\gamma\phi(\gamma_0,X)\right]\\
			&=\frac{1}{n}\sum_{i=1}^n\E^*\left[A_i(x)\left(\frac{1}{\phi(\gamma_0,X)}+\frac{1}{1-\phi(\gamma_0,X)}\right)\nabla_\gamma\phi(\gamma_0,X)\right] \\
			&\quad+\E^*\left[R_n(X)\left(\frac{1}{\phi(\gamma_0,X)}+\frac{1}{1-\phi(\gamma_0,X)}\right)\nabla_\gamma\phi(\gamma_0,X)\right].
		\end{split}
		\] 
		The second term on the right hand side of the previous display is bounded by $c\sup_x|R_n(x)|=o(n^{-1/2})$ for some $c>0$ because of assumptions (C1) and (AN1). Furthermore, from (AN1) and (AC4) and a Taylor expansion, it follows that the generic element of the sum in the first term is equal to 
		\[
		\begin{split}
			&-\int_{\X}\frac{1}{b}k\left(\frac{x-X_i}{b}\right)
			\left\{\frac{\Delta_i\1_{\{Y_i\leq\tau_0\}}}{H([Y_i,\infty)|x)}-\int_0^{Y_i\wedge \tau_0}\frac{H_1(ds|x)}{H^2([s,\infty)|x)}\right\}\frac{1}{\phi(\gamma_0,x)}\nabla_\gamma\phi(\gamma_0,x)\,\dd x\\
			&=-
			\left\{\frac{\Delta_i\1_{\{Y_i\leq\tau_0\}}}{H([Y_i,\infty)|X_i)}-\int_0^{Y_i\wedge \tau_0}\frac{H_1(ds|X_i)}{H^2([s,\infty)|X_i)}\right\}
			\frac{1}{\phi(\gamma_0,X_i)}\nabla_\gamma\phi(\gamma_0,X_i)+O(b^2).
		\end{split}
		\]
		Since because of (C1) we have $O(b^2)=o(n^{-1/2})$, (AN4-iii) holds with 
		\[
		\Psi(Y,\Delta,X)=-\left\{\frac{\Delta\1_{\{Y\leq\tau_0\}}}{H([Y,\infty)|X)}-\int_0^{Y\wedge \tau_0}\frac{H_1(ds|X)}{H^2([s,\infty)|X)}\right\}\frac{1}{\phi(\gamma_0,X)}\nabla_\gamma\phi(\gamma_0,X).
		\]
		\qed
		\section*{Acknowledgements}
		I. Van Keilegom and E. Musta acknowledge  financial support from the European Research Council (2016-2021, Horizon 2020 and grant agreement 694409). For the simulations we used the infrastructure of the Flemish
		Supercomputer Center (VSC).
		
		\section*{Supplement}
		Supporting information may be found in the online appendix.
		{It contains the proofs of Theorems 1, 2 and 3} in Section~\ref{sec:asymptotics} and  additional simulation results.

\bibliographystyle{imsart-number} 
\bibliography{cure_models_all}       


\clearpage
\setcounter{equation}{0}
\renewcommand{\theequation}{A\arabic{equation}} \setcounter{table}{0}
\renewcommand{\thetable}{A\arabic{table}}  
\setcounter{page}{1}

\newpage

\medskip

\newpage

\thispagestyle{plain}
\centerline{\LARGE\bf A presmoothing approach}
\smallskip
\centerline{\LARGE\bf for estimation in semiparametric mixture cure models}
\bigskip
\centerline{\LARGE Supplementary Material}
\bigskip
\centerline{\Large Eni Musta$^*$, Valentin Patilea$^\dagger$ and Ingrid Van Keilegom$^*$}
\medskip
\centerline{\it $^*$KU Leuven, $^\dagger$Ensai }
\bigskip

\appendix

This supplement is organized as follows.
Appendix A contains technical lemmas and proofs.
 Appendix B collects additional simulation results, that were omitted from the main paper due to page limits.

\section{Technical lemmas and proofs}
\label{sec:proofs1}
\begin{lemma}
	\label{lemma:BP}
Let $B$ be a Bernoulli random variable, $T_0$ a nonnegative random variable and let $T=T_0$ if $B=1$ and $T=\infty$ if $B=0$. Let $X$ and $Z$ be two real-valued random vectors. Then
\[
T_0\perp (C,X)\mid Z\quad \text{ and } B\perp (C,T_0,Z)\mid X\quad\Longrightarrow\quad T\perp C\mid (X,Z)
\]
\end{lemma}
\begin{proof}
This lemma is similar to Lemma 8.1 in \cite{BP}. We provide the proof for completeness. By elementary properties of conditional independence we have
\[
B\perp (C,T_0,Z)\mid X \quad\Longleftrightarrow \quad B\perp C\mid (X,Z,T_0)\quad\text{and}\quad B\perp T_0\mid (X,Z)\quad\text{and}\quad B\perp Z\mid X
\]
and
\[
T_0\perp (C,X)\mid Z \quad\Longleftrightarrow \quad T_0\perp C\mid (X,Z)\quad\text{and}\quad T_0\perp X\mid Z.
\]
Then,
\[
C\perp B\mid (X,Z,T_0) \quad\text{and}\quad C\perp T_0\mid (X,Z)\quad\Longleftrightarrow \quad (B,T_0)\perp C\mid (X,Z)
\]
The result follows from the fact that $T$ is completely determined by $B$ and $T_0$.
\end{proof}


\subsection{Identifiability with restricted survival times}

For any $0<\tau^*\leq \tau_0$, let
$$
T_0^*= \min (T_0, \tau^*), \qquad  T^* = B T_0^*+ (1-B) \infty \qquad \text{and} \qquad  
C^* =  \min (C, \tau^*).
$$
Moreover, let 
$$
Y^* = \min (T^*, C^*)\quad \text{ and }
\Delta^* = \1_{\{T^*\leq C^*\}}.
$$

A first aspect to study is the identifiability of the true values of the parameter when $(Y,\Delta)$ is replaced by $(Y^*,\Delta^*)$. Here, identifiability means that the true values $\beta_0$ and $\Lambda_0$ of the parameters maximize the expectation of the criterion maximized to obtain the estimators. This issue is addressed in Lemma \ref{theo:identif}. Let us introduce some additional notation: for any $0<\tau^* \leq  \tau_0$ and $\Lambda\in\mathcal H$, $\Lambda_{|\tau^*}$ is defined as 
\begin{equation}\label{def:lambda_trunc}
\Lambda_{|\tau^*}(t) = \Lambda (t), \;\;\forall t\in [0,\tau^*) \quad \text{ and } \quad \Delta\Lambda_{|\tau^*}(\tau^*)=\Lambda_{|\tau^*}(\{\tau^*\})=1.
\end{equation} 
The dominating measure for the model of $T_0$ changes with such a stopped cumulative hazard measure to allow for a positive mass at $\tau^*$. Then, $\ell$ defined in \eqref{def:hat_l_cox_g} becomes  
\begin{multline}\label{el_star}
\ell  (y, d, x, z ; \beta, \Lambda_{|\tau^*}, \gamma) 
=   \1_{\{y< \tau^*\}} \left[d \log f_u(y|z;\beta,\Lambda) \right. \\
\left. + (1-d)\log\left\{1-\phi(\gamma,x)+\phi(\gamma,x)S_u(  y |z;\beta,\Lambda)\right\} \right]\\
+   \1_{\{y \geq  \tau^*\}} \left[ d \log S_u(\tau^*|z;\beta,\Lambda) + (1-d)\log\left\{1-\phi(\gamma,x)\right\}\right] .
\end{multline}

\begin{lemma}
\label{theo:identif}
Let  $0<\tau^* \leq \tau_0$.  
Assume that for any $\tilde \beta\in B$ and $\tilde \Lambda\in\mathcal H$,
\begin{equation}\label{eq:ident_T0}
S_u ( t | z; \tilde \beta , \tilde \Lambda_{|\tau^*}) = S_u ( t | z;\beta_0,\Lambda_{0|\tau^*}), \; \forall t \in [0, \tau^*) \quad \Longrightarrow \quad \tilde \beta=\beta_0 \;\text{ and } \; \tilde \Lambda_{|\tau^*} = \Lambda_{0|\tau^*}.
\end{equation}
 Then $(\beta_0,\Lambda_{0|\tau^*})$ is  the unique solution of 
\begin{equation}
\label{eq:ident_T0_b}
\max_{\beta\in B,\Lambda\in\mathcal H }\mathbb E \left[ \ell ( Y^*   ,\Delta^*,X,Z;\beta,\Lambda_{|\tau^*},\gamma_0) \right] .
\end{equation}
\end{lemma}
{Condition \eqref{eq:ident_T0} is a minimal requirement of identification of the true value of the parameters  in the model for the uncured subjects if the variable $T_0\wedge C$ was observed and only the events in a subset of the support of $T_0$ are considered.}  In the Cox PH model 
\eqref{eq:ident_T0} is guaranteed by the requirement that $Var(Z)$ has full rank.

\begin{proof}[Proof of Lemma \ref{theo:identif}]
First, let 
$$
H_k([0,t] |x,z) = \mathbb P (Y\leq t, { \Delta = k } | X=x,Z=z), \quad k\in\{0,1\}, \quad t\in[0,\infty),
$$
and let $H_k(dt |x,z) $ be the associated conditional measures. These conditional measures characterize the distribution of $(Y,\Delta)$ given $X=x$ and $Z=z$. By the model and independence assumptions, for any $t\geq 0$,
\begin{equation}\label{eq:inv1_0}
H_1(dt | x, z) = \phi(\gamma_0,x) F_C([t,\infty)|x,z) f_u(t|z;\beta_0,\Lambda_{0}) dt,
\end{equation}
and
\begin{equation}\label{eq:inv2_0}
 H_0(dt | x, z ) =  \{1-\phi(\gamma_0,x)+ \phi(\gamma_0,x) S_u(t|z;\beta_0,\Lambda_{0}) \}F_C(dt|x,z) .
\end{equation}
Following an usual notation abuse, herein we
treat $dt$ not just as the length of a small interval but also as the name of the interval itself. Note that up to additive terms which do not depend on the parameters $\beta,\Lambda$,
$$
(y,d)\mapsto d   \log f_u(y |z ;\beta_0,\Lambda_0) 
+ 
(1-d )\log\left\{1-\phi(\gamma_0,x)+\phi(\gamma_0,x)S_u(y|z;\beta_0,\Lambda_0)\right\},
$$
is the conditional log-density of $(Y,\Delta)$ given $X=x$ and $Z=z$. From this and Kullback information inequality one can deduce that the expectation of $\ell$ defined in \eqref{def:hat_l_cox_g} is maximized by $\beta_0,\Lambda_0$ and $\gamma_0$. 

Let $0<\tau^*\leq \tau_0$. Note that 
$$
 H_1([\tau^*,\tau_0] | x, z) = H_1([\tau^*,\infty) | x, z) = \phi(\gamma_0,x) \int_{[\tau^*,\tau_0]}  F_C([t,\infty)|x,z)f_u(t|z;\beta_0,\Lambda_{0}) dt,
$$
and
\begin{multline*}
H_0([\tau^*,\infty) | x, z)  = \phi(\gamma_0,x) \int_{[\tau^*,\tau_0]} S_u(t|z;\beta_0,\Lambda_{0})  F_C(dt|x,z) \\+\{1- \phi(\gamma_0,x)\}  F_C([\tau^*,\infty)|x,z) .
\end{multline*}
Moreover, 
\begin{multline*}
d(x,z;\tau^*) :=  \phi(\gamma_0,x) \int_{[\tau^*,\tau_0]}  F_C([t,\infty)|x,z)f_u(t|z;\beta_0,\Lambda_{0}) dt \\ +\phi(\gamma_0,x) \int_{[\tau^*,\tau_0]} S_u(t|z;\beta_0,\Lambda_{0})  F_C(dt|x,z)  \\ = \phi(\gamma_0,x) F_C([\tau^*,\infty) |x,z) S_u(\tau^*|z;\beta_0,\Lambda_{0})\\ = \mathbb P (T_0 \wedge C \geq \tau^*, B=1).
\end{multline*}
In the limit case of no cure, $d(x,z;\tau^*) =  H_1([\tau^*,\infty) | x, z)+ H_0([\tau^*,\infty) | x, z)$. 
By construction we have 
$
Y^* =   \min (Y, \tau^*),
$
and
$$
\mathbb P (Y^*=\tau^*, \Delta^* =1 | X=x, Z=z) = 
d(x,z;\tau^*).
$$

Next, let 
$$
H_k^*([0,t] |x,z) = \mathbb P (Y^*\leq t, { \Delta^* = k } | X=x,Z=z), \quad k\in\{0,1\}, \quad t\in[0,\infty),
$$
and let $H_k^*(dt |x,z) $ be the associated conditional measures. 
This means
for any $t\in[0,\tau^*)  $, 
\begin{equation*}
H^*_1(dt | x, z) = H_1(dt | x, z) \quad \text{and} \quad H^*_0(dt | x, z) = H_0(dt | x, z) .
\end{equation*}
Moreover, 
$$
H_1^*(\{\tau^*\} | x, z) = H_1^*([\tau^*,\infty) | x, z) = d(x,z;\tau^*),
$$
and
$$
H_0^*( \{\tau^*\} | x, z) =H_0^*([\tau^*,\tau_0] | x, z)  =  \{1- \phi(\gamma_0,x)\}  F_C([\tau^*,\infty)|x,z).
$$

Now, according to  the inversion formulae of \cite{PK}, without any reference to a model, one can solve the set of equations 
\begin{eqnarray}\label{sq2}
H^*_1(dt | x, z) &=& \phi^*(x,z) F^*_C([t,\infty)|x,z) F^*_u(dt|x,z),\notag \\
\\
H_0^*(dt | x, z ) &=& \{1-\phi^*(x,z)+ \phi^*(x,z) S^*_u(t|x,z) \}F^*_C(dt|x,z) ,\notag
\end{eqnarray}
where $F_u^* = 1-S_u^*$. Solving \eqref{sq2} for $F^*_C$, $S^*_u$ and $\phi^*$, the functional $S^*_u$ is a proper survival function which puts mass only on sets where  $H_1^*$ does. Note that solving the similar system with  $H_1,H_0$ instead of $H_1^*, H_0^*$, one gets the true $F_C$, $S_u$ and $\phi$. 
If $\Lambda_C^*$ denotes the cumulative hazard function associated to the solution $F^*_C$, then 
$$
\Lambda_C^*(dt | x,z) = \frac{H_0^*(dt | x, z )}{H^*_1((t,\infty)| x, z)+ H_0^*([t,\infty) | x, z )},\quad t\geq 0,
$$
and thus, by construction, we have 
$ F_C(dt|x,z) = F^*_C(dt|x,z)$  on  $[0,\tau^*)$, for any $x,z$. Then, by \eqref{eq:inv2_0} and the second equation in \eqref{sq2} we deduce 
$$
\phi^*(x,z) F^*_u(t|x,z) = \phi(\gamma_0,x)F_u(t|z;\beta_0,\Lambda_{0}),  \quad \forall t\in [0,\tau^*), \forall x,z.
$$ 
Next, taking into account that $S^*_u(t|x,z)= 0$, $\forall t\geq \tau^*$, $\forall x,z$, and integrating the second equation \eqref{sq2} on $[\tau^*,\infty)$, 
we obtain
\begin{multline*}
\{1-\phi^*(x,z)\}F^*_C([\tau^*,\infty)|x,z) = H_0^*( \{\tau^*\} | x, z ) = \{1-\phi(x,z)\}F_C([\tau^*,\infty)|x,z) .
\end{multline*}
Since $F^*_C([0,\tau^*)|x,z) = F_C([0,\tau^*)|x,z)$, we deduce that $\phi^*(x,z) =  \phi(\gamma_0,x)$
and thus 
\begin{equation}\label{za1}
F^*_u(t|x,z) = F_u(t|z;\beta_0,\Lambda_{0})= F_u(t|z;\beta_0,\Lambda_{0|\tau^*}),  \quad \forall t\in [0,\tau^*), \forall x,z.
\end{equation}
The second equality in the last display is by the construction of the survival function from the cumulative hazard function: only the values of $\Lambda_0$ on $[0,t]$ contribute to obtain $F_u(t|z;\beta_0,\Lambda_{0})$.
Since the inversion formula necessarily yields $F^*_u([0,\tau^*]|x,z) \equiv 1$, we deduce 
\begin{equation}\label{za2}
F^*_u(\{\tau^* \}|x,z)=S_u(\tau^*|z;\beta_0,\Lambda_{0})=S_u(\tau^*|z;\beta_0,\Lambda_{0|\tau^*}).
\end{equation} 

Finally, we can write 
\begin{multline*}
E \left[ \ell ( Y^*  ,\Delta^*,X,Z;\beta,\Lambda_{|\tau^*},\gamma_0) \right] \\
= \iiint 
\log \ell ( t  ,1,x,z;\beta,\Lambda_{|\tau^*},\gamma_0) 
H^*_1(dt | x, z) G(dx, dz)
\\ + \iiint  \log \ell ( t  ,0,x,z;\beta,\Lambda_{|\tau^*},\gamma_0)  H_0^*(dt | x, z)  G(dx, dz).
\end{multline*}
To obtain the identifiability result it remains to apply Kullback information inequality. More precisely, it suffices to {notice} that here, up to additive terms which do not depend on the parameters, {$\ell$ defined} in \eqref{el_star} considered with $\beta_0,\Lambda_{0|\tau^*}$ corresponds to the {log-density of  the conditional law of $(Y^*,\Delta^*)$ given $X=x$ and $Z=z$}. (Note that the dominated measure changed as we introduce jumps at $\tau^*$.) This follows from \eqref{za1} and \eqref{za2}. Thus $\beta_0,\Lambda_{0|\tau^*}$ is solution of the problem \eqref{eq:ident_T0_b}. The unicity of the solution is guaranteed by \eqref{eq:ident_T0}.
\end{proof}

\subsection{Consistency}

\begin{proof}[\textsc{Proof of Theorem~\ref{theo:consistency}.}]
	We follow the idea of \cite{scharfstein}. Since we are interested in almost sure convergence, we work with fixed realizations of the data, $\omega$ that will lie in a set of probability one. Let $\Omega$ be the abstract probability space where the random vector $(B,T_0,C,X,Z)$ is defined (for example we can take $\Omega=\{0,1\}\times[0,\tau_0]\times[0,\tau]\times\X\times\mathcal{Z}$ and $(B,T_0,C,X,Z)(\omega)=\omega$. Let $N\subset\Omega$ be a set of probability one $\p(N)=1$ and fix $\omega\in N$. We will show that each subsequence $\hat\gamma_{n_k}$ has a subsequence that converges to $\gamma_0$. As a  bounded sequence in $\R^p$, $\hat\gamma_{n_k}$ has a convergent subsequence $\hat\gamma_{m_k}\to\gamma^*$. It suffices to show that $\gamma^*=\gamma_0$. Since $\hat\gamma_{m_k}$ maximizes $\log \hat{L}_{m_k,1}$, we have
	\begin{equation}
	\label{eqn:consistency_1}
	\begin{split}
	0&\leq \frac{1}{m_k}\log \hat{L}_{m_k,1}(\hat\gamma_{m_k})-\frac{1}{m_k}\log \hat{L}_{m_k,1}(\gamma_0)\\
	&=\frac{1}{m_k}\sum_{i=1}^{m_k}\left[\left\{1-\hat\pi(X_i)\right\}\log\frac{\phi(\hat\gamma_{m_k},X_i)}{\phi(\gamma_0,X_i)}+\hat\pi(X_i)\log\frac{1-\phi(\hat\gamma_{m_k},X_i)}{1-\phi(\gamma_0,X_i)}\right]\\
	&=\frac{1}{m_k}\sum_{i=1}^{m_k}\left[\left\{1-\pi_0(X_i)\right\}\log\frac{\phi(\gamma^*,X_i)}{\phi(\gamma_0,X_i)}+\pi_0(X_i)\log\frac{1-\phi(\gamma^*,X_i)}{1-\phi(\gamma_0,X_i)}\right]+o(1)
	\end{split}
	\end{equation}
	if $N\subset\{\omega: \sup_x\left|\hat\pi(x)-\pi_0(x)\right|\to 0\}$. Note that the remainder term $o(1)$ in the previous display depends on $\omega$ and converges to zero as $\hat\pi$ converges to $\pi_0$. Next we will show that, for an appropriate choice of $N$,  the first term converges to 
	\begin{equation}
	\label{eqn:expectation_consistency}
	\E\left[\left\{1-\pi_0(X)\right\}\log\frac{\phi(\gamma^*,X)}{\phi(\gamma_0,X)}+\pi_0(x)\log\frac{1-\phi(\gamma^*,X)}{1-\phi(\gamma_0,X)}\right]
	\end{equation}
	where the expectation is taken with respect to $X$ and $\gamma^*\in\R^p$ (for a fixed $\omega$). Since here we are dealing with a simple parametric model, this convergence follows easily from the uniform law of large numbers. However, we follow a longer argument to explain the idea that will be used also in the proof of Theorem~\ref{theo:consistency2} (where the model is semiparametric). It is obvious, by the law of large numbers, that 
	\[
	\begin{split}
	&\frac{1}{m_k}\sum_{i=1}^{m_k}\left[\left\{1-\pi_0(X_i)\right\}\log\phi(\gamma_0,X_i)+\pi_0(X_i)\log\left(1-\phi(\gamma_0,X_i)\right)\right]\\
	&\to \E\left[\left\{1-\pi_0(X)\right\}\log \phi(\gamma_0,X)+\pi_0(x)\log\left(1-\phi(\gamma_0,X)\right)\right]\text{ a.s. }
	\end{split}
	\]
	and, at first sight it seems that the same holds when $\gamma_0$ is replaced by $\gamma^*$. However, the proof is more delicate because $\gamma^*$ depends on $\omega$ and thus also the event of probability one where  the strong law of large numbers holds for this average. To avoid this we consider a countable dense subset of $G$, $\{\tilde{\gamma}_l\}_{l\geq 1}$ (for example the subset for which all components of $\gamma$ are rational numbers). Now, consider the countable collection  of the probability one sets $\{N_l\}_{l\geq 1}$ where
	\[
	\begin{split}
	&\frac{1}{m_k}\sum_{i=1}^{m_k}\left[\left\{1-\pi_0(X_i)\right\}\log\phi(\tilde\gamma_l,X_i)+\pi_0(X_i)\log\left(1-\phi(\tilde\gamma_l,X_i)\right)\right]\\
	&\to \E\left[\left\{1-\pi_0(X)\right\}\log \phi(\tilde\gamma_l,X)+\pi_0(x)\log\left(1-\phi(\tilde\gamma_l,X)\right)\right]\quad\forall l\geq 1.
	\end{split}
	\]
	If $N\subseteq\left(\cap_{l\geq 1} N_l\right)$, we can write 
	\[
	\begin{split}
	&\left|\frac{1}{m_k}\sum_{i=1}^{m_k}\left[\left\{1-\pi_0(X_i)\right\}\log\phi(\gamma^*,X_i)+\pi_0(X_i)\log\left(1-\phi(\gamma^*,X_i)\right)\right]\right.\\
	&\quad-\E\left[\left\{1-\pi_0(X)\right\}\log \phi(\gamma^*,X)+\pi_0(x)\log\left(1-\phi(\gamma^*,X)\right)\right]\Bigg|\\
	&\leq \left|\frac{1}{m_k}\sum_{i=1}^{m_k}\left[\left\{1-\pi_0(X_i)\right\}\log\frac{\phi(\gamma^*,X_i)}{\phi(\tilde{\gamma}_l,X_i)}+\pi_0(X_i)\log\frac{\left(1-\phi(\gamma^*,X_i)\right)}{\left(1-\phi(\tilde\gamma_l,X_i)\right)}\right]\right|\\
	&\quad+\left|\frac{1}{m_k}\sum_{i=1}^{m_k}\left[\left\{1-\pi_0(X_i)\right\}\log\phi(\tilde\gamma_l,X_i)+\pi_0(X_i)\log\left(1-\phi(\tilde\gamma_l,X_i)\right)\right]\right.\\ 
	&\qquad-\E\left[\left\{1-\pi_0(X)\right\}\log \phi(\tilde\gamma_l,X)+\pi_0(x)\log\left(1-\phi(\tilde\gamma_l,X)\right)\right]\Bigg|\\
	&\quad+\left|\E\left[\left\{1-\pi_0(X)\right\}\log\frac{ \phi(\gamma^*,X)}{\phi(\tilde\gamma_l,X)}+\pi_0(x)\log\frac{1-\phi(\gamma^*,X)}{1-\phi(\tilde\gamma_l,X)}\right]\right|.
	\end{split}
	\]
	Since $\tilde\gamma_l$ can be taken arbitrarily close to $\gamma^*$, by properties of $\phi$ in assumptions (AC3)-(AC4), it can be easily derived that, for an appropriate choice of $\tilde\gamma_l$, the first and the third term on the right hand side in the previous equation converge to zero. Moreover, the second term also converges to zero in the set of probability one that we are considering. As a result, we can conclude that 
	\[
	\begin{split}
	0&\leq \frac{\log \hat{L}_{m_k,1}(\hat\gamma_{m_k})-\log \hat{L}_{m_k,1}(\gamma_0)}{m_k}\\
	&=\E\left[\left\{1-\pi_0(X)\right\}\log\frac{\phi(\gamma^*,X)}{\phi(\gamma_0,X)}+\pi_0(x)\log\frac{1-\phi(\gamma^*,X)}{1-\phi(\gamma_0,X)}\right]+o(1)
	\end{split}
	\]	
	For each $x\in\X$, consider the function
	\[
	g_x(z)=\phi(\gamma_0,x)\log\frac{z}{\phi(\gamma_0,x)}+\left\{1-\phi(\gamma_0,x)\right\}\log\frac{1-z}{1-\phi(\gamma_0,x)},\quad z\in(0,1).
	\]
	It is easy to check that $g_x(z)\leq 0$ and the equality holds only if $z=\phi(\gamma_0,x)$. Hence, the expectation in \eqref{eqn:expectation_consistency} is smaller or equal to zero. Due to the inequality in \eqref{eqn:consistency_1}, it must be equal to zero, which means that $\phi(\gamma^*,X)=\phi(\gamma_0,X)$. By the identifiability assumption \eqref{eqn:CI3}, this is possible only if $\gamma^*=\gamma_0$.
	\end{proof}

\begin{lemma}
	\label{lemma:boundedness_hat_Lambda}
	Assume (AC2),(AC5)
	hold {and $\tau^*$ is such that \eqref{eqn:no_jump_cond} is satisfied.}  Then $\sup_{n}\hat\Lambda_n({\tau^*})<\infty$ almost surely.
\end{lemma}
\begin{proof}
	
		By definition
		\[
		\hat\Lambda_n({\tau^*})=\frac{1}{n}\sum_{i=1}^{n}\frac{\Delta_i\1_{\{Y_i< {\tau^*}\}}}{\frac{1}{n}\sum_{j=1}^n\1_{\{Y_i\leq Y_j\leq \tau_0\}}\exp(\hat\beta'_nZ_j)\left\{{\Delta_j}+(1-{\Delta_j})g_j(Y_j,\hat\Lambda_n,\hat\beta_n,\hat\gamma_n)\right\}}.
		\]
		From assumptions (AC2) and (AC5) we have
	\[
	\begin{split}
	&\frac{1}{n}\sum_{j=1}^n\1_{\{\tau^*\leq Y_j\leq \tau_0\}}\exp(\hat\beta'_nZ_j)\left\{\Delta_j+(1-\Delta_j)g_j(Y_j,\hat\Lambda_n,\hat\beta_n,\hat\gamma_n)\right\}\\
	&\geq \frac{1}{n}\sum_{j=1}^n\Delta_j\1_{\{\tau^*\leq Y_j\leq \tau_0\}}\exp(\hat\beta'_nZ_j)\\
	&\geq c\frac{1}{n}\sum_{j=1}^n\Delta_j\1_{\{\tau^*\leq Y_j\leq \tau_0\}},
	\end{split}
	\]	
	for some $c>0$. Since $\frac{1}{n}\sum_{j=1}^n\Delta_j\1_{\{\tau^*\leq Y_j\leq \tau_0\}}\xrightarrow{a.s.} \p\left(Y\geq \tau^*,\Delta=1\right)>0$, it follows that $\frac{1}{n}\sum_{j=1}^n\Delta_j\1_{\{\tau^*\leq Y_j\leq \tau_0\}}$ is bounded from below away from zero almost everywhere. As a result
	\[
	\begin{split}
	\sup_n\hat\Lambda_n(\tau^*)&\leq\sup_n \frac{1}{n}\sum_{i=1}^{n}\frac{\Delta_i\1_{\{Y_i< \tau^*\}}}{\frac{1}{n}\sum_{j=1}^n\1_{\{\tau^*\leq Y_j\leq \tau_0\}}\exp(\hat\beta'_nZ_j)\left\{\Delta_j+(1-\Delta_j)g_j(Y_j,\hat\Lambda_n,\hat\beta_n,\hat\gamma_n)\right\}}\\
	&\leq\sup_n \frac{1}{n}\sum_{i=1}^{n}\frac{\Delta_i\1_{\{Y_i< \tau_0\}}}{c\frac{1}{n}\sum_{j=1}^n\Delta_j\1_{\{\tau^*\leq Y_j\leq \tau_0\}}}\\
	&\leq \frac{1}{c}\left(\inf_n{\frac{1}{n}\sum_{j=1}^n\Delta_j\1_{\{\tau^*\leq Y_j\leq \tau_0\}}}\right)^{-1}
	\end{split}
	\]
	is bounded almost surely.
	Note that, if \eqref{eqn:jump_cond} is satisfied, then we can take $\tau^*=\tau_0$.
\end{proof}

\begin{proof}[\textsc{Proof of Theorem~\ref{theo:consistency2}.}]

Let  $0<\tau^* \leq \tau_0$ and 
\begin{equation*}
\hat{l}_n^*( \beta, \Lambda_{|\tau^*}, \hat\gamma_n ) = \frac{1}{n}\sum_{i=1}^n \ell (Y_i^*,\Delta_i^*,X_i,Z_i;\beta, \Lambda_{|\tau^*}, \hat\gamma_n),
\end{equation*}
with $\ell$ defined in \eqref{el_star}. 
If we consider the Cox PH model for the conditional law of $T_0$, then 
\begin{multline*}
\hat{l}^*_n(\beta,\Lambda_{|\tau^*}, \hat\gamma_n)  
=\frac{1}{n}\sum_{i=1}^n \Delta_i \left[ \1_{\{Y_i < \tau^*\}} \left\{ \log \Delta\Lambda(Y_i)+\beta'Z_i-\Lambda(Y_i)e^{\beta'Z_i} \right\} 
\right] \\
+\frac{1}{n}\sum_{i=1}^n(1-\Delta_i ) \1_{\{Y_i < \tau^*\}}\log\left\{1-\phi(\hat\gamma_n,X_i)+\phi(\hat\gamma_n,X_i)\exp\left(-\Lambda(Y_i)e^{\beta'Z_i}\right)\right\}\\
- \frac{\Lambda(\tau^*-)}{n}\sum_{i=1}^n \1_{\{Y_i \geq \tau^*\}} \1_{\{B_i =1\}} e^{\beta'Z_i}
+  \frac{1}{n}\sum_{i=1}^n \1_{\{Y_i \geq \tau^*\}} \1_{\{B_i =0\}}  \log\left\{1-\phi(\hat\gamma_n,X_i)\right\},
\end{multline*}
and has to be maximized with respect to $\beta$ and $\Lambda$ in
the class of step functions $\Lambda$
with jumps of size $\Delta \Lambda$ at the event times in $[0,\tau^*) $. 
As in \cite{Lu}, it can be shown that the maximizer $(\hat\Lambda^*_n,\hat\beta^*_n)$ of $\hat{l}^*_n$ exists and it is finite. Moreover, for $t\in[0,\tau_0] $, $\hat{\Lambda}_n^* = \Lambda^*_{n,\hat\beta^*_n, \hat\gamma_n}$ where 
\begin{equation*}
\Lambda^*_{n,\beta, \gamma} (t)=\frac{1}{n}\sum_{i=1}^{n}\frac{\Delta_i\1_{\{Y_i{\leq t, Y_i<\tau^* }\}}}{\frac{1}{n}\sum_{j=1}^n\1_{\{{Y_j\geq Y_i }\}}\exp(\beta' Z_j)\left\{{\Delta_j^*}+(1-{\Delta_j})\1_{\{Y_j < \tau^*\}}g_j(Y_j, \Lambda^*_{n,\beta,\gamma},\beta,\gamma)\right\}} ,
\end{equation*}
$\Delta^* = \1_{\{T_0^* \leq C^*\}} = \Delta \1_{\{Y_j < \tau^*\}} + \1_{\{Y_i \geq \tau^*\}} \1_{\{B_i =1\}}$ and $g_j(t,\Lambda,\beta,\gamma)$ defined in \eqref{def:g_j}. 
%
%

Let
\begin{equation}
\label{def:tilde_Lambda_0}
\tilde\Lambda_{0,n}(t)=\frac{1}{n}\sum_{i=1}^{n}\frac{\Delta_i\1_{\{Y_i\leq t{, Y_i<\tau_0}\}}}{\frac{1}{n}\sum_{j=1}^n\1_{\{Y_j\geq Y_i{,Y_j\leq\tau_0}\}}\exp(\beta'_0Z_j)\left\{{\Delta_j}+(1-{\Delta_j})g_j(Y_j,\Lambda_{0},\beta_0,\gamma_0))\right\}}. 
\end{equation}

{We want to prove that $\hat\beta_n{\xrightarrow{a.s.}}\beta_0$, and $\sup_{t\in[0,\bar\tau]}|\hat\Lambda_n(t)-\Lambda_0(t)|{\xrightarrow{a.s.}}0$ for any $\bar\tau<\tau_0$. We suppose  that the previous statement is false, i.e $\hat\beta_n$ does not converge {almost surely} to $\beta_0$ or there exists $\bar\tau$ such that $\sup_{t\in[0,\bar\tau]}|\hat\Lambda_n(t)-\Lambda_0(t)|$ does not converge to zero {almost surely}. This means that, 
	there exist  $\epsilon>0$  and $\bar\tau<\tau_0$ such that
	\[
	\p[A_{1} (\bar\tau,\epsilon)]>0, \;\; \text{ with } \;\; A_{1} (\bar\tau,\epsilon)=\left\{\limsup_{n\rightarrow \infty }\left[ \left\|\hat\beta_n-\beta_0\right\|+\sup_{t\in[0,\bar\tau]}\left|\hat\Lambda_n(t)-\Lambda_0(t)\right|\right] >\epsilon \right\}.
	\]	
	{On the other hand,} since $(\hat\Lambda_n,\hat\beta_n)$ maximizes $\hat\ell_n(\Lambda,\beta,\hat\gamma_n)$, for any realization {$\omega$} of the data we have 
	\begin{equation}
	\label{eqn:inequality_1}
	\hat{l}_n(\hat \beta_n,  \hat \Lambda_{n}, \hat \gamma_n) - \hat{l}_n( \beta_0, \tilde \Lambda_{0,n}, \hat \gamma_n)  \geq 0.
	\end{equation}
{Then the idea for creating the contradiction  is to } show that  the previous inequality is not satisfied for any $\omega$ in some event of positive probability. 
	We argue for a fixed realization $\omega$ of the data. 
	As a  bounded sequence in $\R^q$, $\hat\beta_{n}$ has a convergent subsequence $\hat{\beta}_{n_k}\to\bar\beta$. Let $(\tau_i)_{i\geq 1}$ be an increasing sequence such that $\lim_{i\to \infty}\tau_i=\tau_0$.  Since for all $\tau<\tau_0$, $\hat\Lambda_n(\tau)<\infty$ almost surely (see Lemma~\ref{lemma:boundedness_hat_Lambda}), by Helly's  selection theorem (\cite{ash}), there exists a 
	subsequence $\hat{\Lambda}_{m_k}$ of $\hat{\Lambda}_{n_k}$, converging pointwise to a function $\bar\Lambda$ on $[0,\tau_1]$. Repeating the same argument, we can extract a further subsequence converging pointwise to a function $\bar\Lambda$ on $[0,\tau_2]$ and so on. Hence, there exist a subsequence $\hat{\Lambda}_{r_k}$ converging pointwise to a function $\bar\Lambda$ on all compacts of  $[0,\tau_0]$ that do not include $\tau_0$. {This defines a monotone function $\bar\Lambda$ on $[0,\tau_0)$, which could be extended at $\tau_0$ by taking the limit.}  
	As in Lemma 2 of \cite{Lu}, it can be shown that $\bar\Lambda$ is absolutely continuous and pointwise  convergence of monotone functions to a continuous monotone function implies uniform convergence on  compacts.  Note that the chosen subsequence and the limits $\bar\beta$ and $\bar\Lambda$ depend on $\omega$. To keep the notation simple, in what follows we use the index $n$ instead of the chosen subsequence $r_k$. 	
	For any $\tau^*<\tau_0$,  we can write
	\begin{align}\label{lik_deco}
	&\!\!\! 0\leq \hat{l}_n(\hat \beta_n,  \hat \Lambda_{n}, \hat \gamma_n) - \hat{l}_n( \beta_0, \tilde \Lambda_{0,n}, \hat \gamma_n) \notag \\ 
	&= \hat{l}_n^*(\hat \beta_n,  \hat \Lambda_{n|\tau^*}, \hat \gamma_n) + D_{1n} - \hat{l}^*_n( \beta_0, \tilde \Lambda_{0,n|\tau^*},\hat \gamma_n)   - D_{2n} \notag \\
	&= \mathbb E [\ell (Y^*,\Delta^*,X,Z;\bar \beta, \bar\Lambda_{|\tau^*},\gamma_0)] + D_{1n}+ R_{1n} \notag\\
	&\quad- \mathbb E [\ell (Y^*,\Delta^*,X,Z;\beta_0,\Lambda_{0|\tau^*} ,\gamma_0)] - D_{2n}-R_{2n}
	,
	\end{align}
	where
	\begin{equation}\label{deco_d1}
	D_{1n} =  \hat{l}_n(\hat \beta_n,  \hat \Lambda_{n}, \hat \gamma_n) - \hat{l}_n^*(\hat \beta_n,  \hat \Lambda_{n|\tau^*}, \hat \gamma_n),
	\end{equation}
	\begin{equation}\label{deco_d2}
	D_{2n} = \hat{l}_n( \beta_0, \tilde \Lambda_{0,n}, \hat \gamma_n) - 
	\hat{l}^*_n( \beta_0, \tilde \Lambda_{0,n|\tau^*},\hat \gamma_n),
	\end{equation}
	\begin{equation}\label{deco_r1}
	R_{1n} = \hat{l}_n^*(\hat \beta_n,  \hat\Lambda_{n|\tau^*}, \hat \gamma_n) - \mathbb E [\ell (Y^*,\Delta^*,X,Z;\bar\beta, \bar\Lambda_{|\tau^*} ,   \gamma_0)] ,
	\end{equation}
	\begin{equation}\label{deco_r2}
	R_{2n} =\hat{l}^*_n( \beta_0, \tilde \Lambda_{0,n|\tau^*},\hat \gamma_n)  - \mathbb E [\ell (Y^*,\Delta^*,X,Z;\beta_0,\Lambda_{0|\tau^*} ,\gamma_0)] .
	\end{equation}
Note that the limit of $(\hat\beta_n,\hat\Lambda_n)$ depends on $\omega$, but here the expectation is taken with respect to $(Y^*,\Delta^*,X,Z)$ for fixed $(\bar\beta,\bar\Lambda)$.
	We now define the event 
	$
	A_{3}(\tau^*)=\left\{|R_{1n}-R_{2n}|\rightarrow 0\right\}.
	$
	By Lemma \ref{lem:consist_full_1}, for any $\tau^*<\tau_0$, we have  $\p[A_{1} (\bar\tau,\epsilon)\cap A_{3}(\tau^*)]=\p[A_{1} (\bar\tau,\epsilon)]$.
	Next, for $\bar\tau<\tau_0$ and $\epsilon>0$ such that $\p[A_{1} (\bar\tau,\epsilon)]>0$, by Lemma \ref{lem:consist_full_3} there exist $0<c_1<1$ and $\delta>0$ such that we have
	\[
	\begin{split}
	c&=\inf\Bigg\{\mathbb E [\ell (Y^*,\Delta^*,X,Z;\beta_0,\Lambda_{0|\tau^*} ,\gamma_0)]-	\mathbb E [\ell (Y^*,\Delta^*,X,Z; \beta, \Lambda_{|\tau^*},\gamma_0)]:\\
	&\left.\qquad\bar{\tau}+\delta \leq \tau^*<\tau_0,\quad \|\beta-\beta_0\|\geq    c_1\epsilon/2   \quad\text{ or}\quad \sup\limits_{t\in[0,\bar \tau]}|\Lambda(t)-\Lambda_0(t)|\geq    (1-c_1) \epsilon /2   \right\}>0.
	\end{split}
	\]
	Note that if $\omega \in  A_{1}(\bar\tau,\epsilon)$ and $\bar\beta$ and $\bar \Lambda$ are the limits for $ \hat \beta_n$ and $\hat \Lambda_n$, respectively, then necessarily, either $\|\bar \beta-\beta_0\|\geq    c_1\epsilon/2 ,$ or $\sup\limits_{t\in[0, \bar \tau]}|\bar \Lambda(t)-\Lambda_0(t)|\geq    (1-c_1) \epsilon /2$, and consequently 
	$$
	\mathbb E [\ell (Y^*,\Delta^*,X,Z;\bar \beta, \bar\Lambda_{|\tau^*},\gamma_0)]  - \mathbb E [\ell (Y^*,\Delta^*,X,Z;\beta_0,\Lambda_{0|\tau^*} ,\gamma_0)] \geq c,\quad \forall \bar{\tau}{+\delta}\leq \tau^*<\tau_0.
	$$
	Finally, we define $$A_{2}(\tau^*)=\left\{ \limsup_{n\rightarrow\infty} \left|D_{1n}-D_{2,n} \right| \leq c/2\right\},$$ with $D_{1n}$ and $D_{2n} $ defined in \eqref{deco_d1} and \eqref{deco_d2}, and choose $\tau^*\in [\bar\tau+\delta,\tau_0)$ such that 
	$$
	c_b \left\{ \p(T_0\geq \tau^*) \log \{1/\p (T_0\geq \tau^*)\} +  \p(C\in[ \tau^*,\tau_0])\right\}<c/2,
	$$
	with $c_b$ the constant from Lemma \ref{lem:consist_full_2}. Then we have $\p[A_{2}(\tau^*)]=1$.
	Gathering facts, we deduce that  by a suitable choice of $\tau^*\in[\bar\tau+\delta,\tau_0)$, we necessarily have $\p[ A_{1} (\bar\tau,\epsilon)\cap
	A_{2}(\tau^*)\cap A_{3}(\tau^*)]>0$. Moreover,  with such a suitable $\tau^*$, 
	for any  $\omega\in A_{1} (\bar\tau,\epsilon)\cap
	A_{2}(\tau^*)\cap A_{3}(\tau^*)$, 
	we have  
	$$\limsup_{n\rightarrow \infty}  \left[ \hat{l}_n(\hat \beta_n,  \hat \Lambda_{n}, \hat \gamma_n) - \hat{l}_n( \beta_0, \tilde \Lambda_{0,n}, \hat \gamma_n)\right] \leq -c/2<0.$$ 
	We deduce that  \eqref{eqn:inequality_1} is violated on an event of positive probability, which by definition is impossible. Thus
	$\hat\beta_n{\xrightarrow{a.s.}}\beta_0$, and $\sup_{t\in[0,\bar\tau]}|\hat\Lambda_n(t)-\Lambda_0(t)|{\xrightarrow{a.s.}}0$ for any $\bar\tau<\tau_0$.
}

{If condition \eqref{eqn:jump_cond} is satisfied, we want to show in addition that $|\hat\Lambda_n(\tau_0)-\Lambda_0(\tau_0)|{\xrightarrow{a.s.}}0$. In that case,  $\hat\Lambda_n(\tau_0)<\infty$ almost surely and as a result, for any realization $\omega$, there exists a subsequence $\hat\Lambda_{r_k}$ converging  to some absolutely continuous function  $\bar\Lambda$ uniformly on $[0,\tau_0]$. 
	Since we already showed that 
	$|\hat\Lambda_n(t)-\Lambda_0(t)|{\xrightarrow{a.s.}}0$ for any $t<\tau_0$ and $\Lambda_0(\tau_0)=\lim_{t\uparrow\tau_0}\Lambda_0(t)$, we necessarily have $\bar\Lambda=\Lambda_0$ on the whole interval $[0,\tau_0]$. This concludes the proof of the Theorem.
}
\end{proof}

\begin{lemma}
	\label{lem:consist_full_1}
	Consider a realization of the data $\omega$ and assume that $\hat\beta_n(\omega)\to\bar\beta$ and $\hat\Lambda_n(\omega)(t)\to \bar\Lambda(t)$ for any $t\in[0,\tau_0)$, for some absolutely continuous function $\bar\Lambda$. Let $0<\tau^*<\tau_0$ and let $R_{1n}$, $R_{2n}$ be defined as in \eqref{deco_r1} and \eqref{deco_r2}, respectively. There exists an event $A_3(\tau^*)$ of probability one such that, for any  $\omega\in A_3(\tau^*)$, 
	$$
	R_{1n}(\omega)-R_{2n}(\omega)\to 0.
	$$
\end{lemma}

\begin{proof} {Let us consider some $0<\tau^*<\tau_0$}. From Theorem \ref{theo:consistency} and  Lemma 2 in \cite{Lu} it follows that the event 
	\[
	A_3^1 {(\tau^*)}=\left\{\hat\gamma_n\to\gamma_0\quad\text{ and }\quad\sup_{t\in[0,\tau^*]}|\tilde{\Lambda}_{0,n}(t)-\Lambda_0(t)|\to 0\right\}
	\]
	has probability one. Next we argue for the given realization of the data $\omega{\in A_3^1(\tau^*)}$ and  will determine the event $A_3{(\tau^*)}$ appropriately.  
	By the triangular inequality we can write
	\begin{align}
	\label{eqn:R2}
	|R_{1n}\!-R_{2n}|&\leq\left|\left\{\hat{l}^*_n(\hat\beta_n,\hat{\Lambda}_{n|\tau*},\hat\gamma_n)\! -\hat{l}^*_n(\beta_0,\tilde{\Lambda}_{0,n|\tau*},\hat\gamma_n)\right\} \!-\! \left\{\hat{l}^*_n(\bar\beta,\bar{\Lambda}_{|\tau*},\!\gamma_0) \!- \! \hat{l}^*_n(\beta_0,{\Lambda}_{0|\tau*},\!\gamma_0)\right\}\right| \notag\\
	&\quad+\left|\hat{l}^*_n(\beta_0,{\Lambda}_{0|\tau*},\!\gamma_0)-\E\left[l(Y^*,\Delta^*,X,Z;\beta_0,\Lambda_{0|\tau^*},\!\gamma_0)\right]\right|\notag\\
	&\quad+\left|\hat{l}^*_n(\bar\beta,\bar{\Lambda}_{|\tau*},\gamma_0)-\E\left[l(Y^*,\Delta^*,X,Z;\bar\beta,\bar\Lambda_{|\tau^*},\gamma_0)\right]\right|.
	\end{align}
	Since $\bar\Lambda$ is absolutely continuous, it is differentiable almost everywhere. Let $\bar\lambda(t)=\dd\bar\Lambda(t)/\dd t$. 
	By definition we have
	\[
	\begin{split}
	&\hat{l}^*_n(\hat\beta_n,\hat{\Lambda}_{n|\tau*},\hat\gamma_n)-\hat{l}^*_n(\beta_0,\tilde{\Lambda}_{0,n|\tau*},\hat\gamma_n)\\
	&=\frac{1}{n}\sum_{i=1}^n \Delta_i \1_{\{Y_i < \tau^*\}} \left\{ \log \frac{\Delta\hat\Lambda_n(Y_i)}{\Delta\tilde\Lambda_{0,n}(Y_i)}+(\hat\beta_n-\beta_0)'Z_i-\hat\Lambda_n(Y_i)e^{\hat\beta'_nZ_i}+\tilde\Lambda_{0,n}(Y_i)e^{\beta'_0Z_i} \right\}  \\
	&\qquad+\frac{1}{n}\sum_{i=1}^n(1-\Delta_i ) \1_{\{Y_i < \tau^*\}}\log\frac{1-\phi(\hat\gamma_n,X_i)+\phi(\hat\gamma_n,X_i)\exp\left(-\hat\Lambda_n(Y_i)e^{\hat\beta'_nZ_i}\right)}{1-\phi(\hat\gamma_n,X_i)+\phi(\hat\gamma_n,X_i)\exp\left(-\tilde\Lambda_{0,n}(Y_i)e^{\beta'_0Z_i}\right)}\\
	&\qquad - \frac{1}{n}\sum_{i=1}^n \1_{\{Y_i \geq \tau^*\}}\1_{\{B_i =1\}}\left\{
	\hat\Lambda_n(\tau^*)e^{\hat\beta'_nZ_i} - \tilde\Lambda_{0,n}(\tau^*)e^{\beta'_0Z_i}\right\}.  
	\end{split}
	\]
	If $\omega\in A^1_3{(\tau^*)}$, we obtain
	\[
	\begin{split}
	&\hat{l}^*_n(\hat\beta_n,\hat{\Lambda}_{n|\tau*},\hat\gamma_n)-\hat{l}^*_n(\beta_0,\tilde{\Lambda}_{0,n|\tau*},\hat\gamma_n)\\
	&=\frac{1}{n}\sum_{i=1}^n \Delta_i \1_{\{Y_i < \tau^*\}} \left\{ \log \frac{\bar\lambda(Y_i)}{\lambda_{0}(Y_i)}+(\bar\beta-\beta_0)'Z_i-\bar\Lambda(Y_i)e^{\bar\beta'Z_i}+\Lambda_{0}(Y_i)e^{\beta'_0Z_i} \right\}  \\
	&\qquad+\frac{1}{n}\sum_{i=1}^n(1-\Delta_i ) \1_{\{Y_i < \tau^*\}}\log\frac{1-\phi(\gamma_0,X_i)+\phi(\gamma_0,X_i)\exp\left(-\bar\Lambda(Y_i)e^{\bar\beta'Z_i}\right)}{1-\phi(\gamma_0,X_i)+\phi(\gamma_0,X_i)\exp\left(-\Lambda_{0}(Y_i)e^{\beta'_0Z_i}\right)}\\
	&\qquad  - \frac{1}{n}\sum_{i=1}^n \1_{\{Y_i \geq \tau^*\}}\1_{\{B_i =1\}}\left\{
	\bar\Lambda(\tau^*)e^{\bar\beta'Z_i}- \Lambda_{0}(\tau^*)e^{\beta'_0Z_i}\right\} +o(1)\\
	&=\hat{l}^*_n(\bar\beta,\bar{\Lambda}_{|\tau*},\gamma_0)-\hat{l}^*_n(\beta_0,{\Lambda}_{0|\tau*},\gamma_0)+o(1),
	\end{split}
	\]
	where the remainder term  depends on $\omega$ and converges to zero. Hence, {the first term on the right hand side of \eqref{eqn:R2}} converges to zero. Let $A_n^2{(\tau^*)}$ be the event where 
	\[
	\hat{l}^*_n(\beta_0,{\Lambda}_{0|\tau*},\gamma_0)\to\E\left[l(Y^*,\Delta^*,X,Z;\beta_0,\Lambda_{0|\tau^*},\gamma_0)\right] \quad\text{ as }n\to\infty.
	\]
	By the law of large numbers $\p[A^2_n{(\tau^*)}]=1$, implying that also {the second term on the right hand side of \eqref{eqn:R2}} converges to zero if $\omega\in A^2_n{(\tau^*)}$.  It remains to deal with the third term. Note that here $(\bar\beta,\bar\Lambda)$ depend on $\omega$ and the expectation is taken with respect to $(Y,\Delta,X,Z)$ for fixed $(\bar\beta,\bar\Lambda)$.
	We have the same issue as in the proof of Theorem \ref{theo:consistency} when dealing with the terms involving $\bar\beta$ and $\bar\Lambda$, so we need to consider approximations by elements of a countable dense subset of $\B$ and of the space of bounded, absolutely continuous, increasing functions in $[0,\tau^*]$ (is separable, so such subset exists). The same reasoning is used also in \cite{murphy,Lu,scharfstein}. Hence, there exists a countable collection of probability one sets $\{N_l\}_{l\geq 1}$ where
	\[
	\hat{l}^*_n(\beta_l,{\Lambda}_{l},\gamma_0)\to\E\left[l(Y^*,\Delta^*,X,Z;\beta_l,\Lambda_{l},\gamma_0)\right] \quad\text{ as }n\to\infty
	\]
	and $(\beta_l,\Lambda_l)$ can be taken arbitrarily close to $(\bar\beta,\bar\Lambda)$. As a result, if $\omega\in A^3_n{(\tau^*)}=\bigcap_{l\geq 1}N_l$, then 
	\[
	\left|\hat{l}^*_n(\bar\beta,\bar{\Lambda}_{|\tau*},\gamma_0)-\E\left[l(Y^*,\Delta^*,X,Z;\bar\beta,\bar\Lambda_{|\tau^*},\gamma_0)\right]\right|\to 0.
	\]
	To conclude, we define $A_3{(\tau^*)}=A^1_3{(\tau^*)}\cap A^2_3{(\tau^*)}\cap A^3_3{(\tau^*)}$ and we have $\p[A_3{(\tau^*)}]=1$. 
\end{proof}


\begin{lemma}
	\label{lem:consist_full_2}
	Let $D_{1n}$ and $D_{2n}$ be defined as in \eqref{deco_d1} and \eqref{deco_d2}, respectively, for some $\tau^* <\tau_0$. Then there exists a constant $c_b$ independent of $\tau^*$ such that  
	$$
	\mathbb P \left[ \limsup_{n\rightarrow\infty} \left|D_{1n}-D_{2,n}\right|  > c_b \left\{ \p(T_0\geq \tau^*) \log \{1/\p (T_0\geq \tau^*)\} +  \p(C\in[ \tau^*,\tau_0])\right\}  \right] =0.
	$$
\end{lemma}

\begin{proof}
	By definition, for any $\gamma$, $\beta$ and cumulative hazard function $\Lambda$ piecewise constant with jumps at the observed events
	\begin{align*}
	& {l}_n(\beta,   \Lambda,  \gamma) - \hat{l}_n^*(\beta,  \Lambda_{|\tau^*}, \gamma)
	\\
	& = \frac{1}{n}\sum_{i=1}^n \1_{\{{\tau^*\leq Y_i< \tau_0 } \}} \Delta_i \log \Lambda (\{Y_i\})\\
	& \quad + \frac{1}{n}\sum_{i=1}^n \1_{\{{\tau^*\leq Y_i< \tau_0 }\}} \Delta_i \beta^\prime Z_i  
	\\ 
	& \quad - \frac{1}{n}\sum_{i=1}^n \1_{\{ Y_i\geq \tau^*  \}} e^{\beta^\prime Z_i} \{ \Delta_i \Lambda (Y_i) - \1_{\{B_i=1\}}\Lambda (\tau^*-) \} \\
	& \quad +\frac{1}{n}\!\sum_{i=1}^n \!\1_{\{\tau_0\geq Y_i \geq  \tau^*\!\}}\!\! \left[\! (1\!-\!\Delta_i )\log\!\left\{\!1\!-\!\phi_i(\gamma)\!+\!\phi_i(\gamma)\exp\left(\!-\Lambda(Y_i)e^{\beta'Z_i}\!\right)\!\!\right\} \right. \! \\
	& \quad \quad \left. - \!\1_{\{B_i=0\}}\! \log\left\{1\!-\!\phi_i(\gamma)\right\}\! \right]\\
	& =: r_{1n} (\Lambda;\tau^*) + r_{2n} (\beta;\tau^*) - r_{3n}  (\Lambda, \beta;\tau^*) + r_{4n} (\Lambda, \beta, \gamma;\tau^*) ,
	\end{align*}
	where $\phi_i (\gamma) $ is a short notation for $\phi(\gamma,X_i)$.  For proving the Lemma, we have to suitably bound $r_{1n},\ldots,r_{4n}$.  For this purpose, let us notice that, by definition, all the cumulative hazard functions we have to consider ($\hat\Lambda_n$, $\tilde \Lambda_{0,n}$,...) have bounded jumps at the event times. More precisely, because the parameter space $\mathcal B$ and $Z$ are supposed bounded, there exist constants 
	$0<c_l < c_u $ such that 
	$$
	c_l\leq \exp(\beta'Z)\leq c_u.
	$$
	Then the largest jump of any of  the cumulative hazard functions we need to consider is bounded by $1/c_l$ (which is located at the last uncensored observation), the second largest one (and is located at the before last uncensored observation) is bounded by $1/2c_l$,...

	To control $r_{1n} (\Lambda;\tau^*) $, one would  look for a suitable lower bound forthe jumps of $\Lambda$. However, no meaningful lower bound could be derived for these jumps. More precisely, such a bound is necessarily of order $1/n$, so that the sequence of the logarithm of the jumps is unbounded. Fortunately, for our purposes it suffices to find a bound for 
	$$
	\left| r_{1n} (\hat \Lambda_{n};\tau^*) - r_{1n} ( \tilde \Lambda_{0,n} ;\tau^*)\right| = \frac{1}{n}\sum_{i=1}^n \1_{\{ {\tau^*\leq Y_i< \tau_0 } \}} \Delta_i \left| \log  \frac{\hat \Lambda_{n} (\{Y_i\})} { \tilde \Lambda_{0,n}  (\{Y_i\}) } \right|,
	$$
	where 
	$$
	\frac{\hat \Lambda_{n} (\{Y_i\})} { \tilde \Lambda_{0,n}  (\{Y_i\}) } = \frac{\sum_{j=1}^n \1_{\{\tau_0 \geq Y_j\geq Y_i \}}\exp(\beta'_0Z_j)\left\{\Delta_j+(1-\Delta_j)g_j(Y_j,\Lambda_{0},\beta_0,\hat \gamma_n))\right\} }
	{\sum_{j=1}^n\1_{\{{\tau_0\geq Y_j\geq Y_i  }\}}\exp(\beta'Z_j)\left\{ \Delta_j +(1-  \Delta_j )g_j(Y_j,{\Lambda}_{n,\hat\beta,\hat\gamma_n},\beta,\hat \gamma_n) \right\} }.
	$$
	Since all $g_j$'s are between 0 and 1, it  is easy to see that  for any uncensored $Y_i\geq \tau^*$,
	$$
	\frac{1}{\rho_n(Y_i)} \frac{c_l } {c_u} \leq \frac{\hat \Lambda_{n} (\{Y_i\})} { \tilde \Lambda_{0,n}  (\{Y_i\}) }  \leq \frac{c_u}{c_l } \rho _n(Y_i) ,
	$$
	where
	$$
	\rho_n(t) =  \frac{\sum_{j=1}^n \1_{\{ \tau_0\geq Y_j \geq t\geq \tau^*\} } } {\sum_{j=1}^n  \Delta_j \1_{\{ \tau_0\geq Y_j \geq t \geq \tau^*\} }}, \qquad t\in[\tau^*, \tau_0].
	$$
	Thus,  since all $\rho_n(Y_i)$'s are larger than 1, it suffices to suitably bound 
	$$
	0\leq A_n = \frac{1}{n}\sum_{i=1}^n \1_{\{ Y_i\geq \tau^* \}} \Delta_i \log (\rho_n(Y_i)),
	$$
	which we decompose as
	$$
	A_n \! = \frac{1}{n}\sum_{i=1}^n \1_{\{ Y_i\in [\tau^*\! ,\tau_0 - a_n  ] \}} \Delta_i \log (\rho_n(Y_i)) +  \frac{1}{n}\sum_{i=1}^n \1_{\{ Y_i\in [\tau_0 - a_n ,\tau_0 ] \}} \Delta_i \log (\rho_n(Y_i)) =: A_{1n} \!+\!  A_{2n} ,
	$$
	for some sequence of real numbers $a_n$, $n\geq 1$, decreasing to zero. The rate of $a_n$ should be taken such that, on one hand,  for any constant $C>0$,
	\begin{equation}\label{eq:up_A2}
	\p (\limsup_{n\rightarrow\infty} A_{2n} > C) = 0,
	\end{equation}
	and, on the other hand, the $\limsup$ of $A_{1n}$ could be controlled by a function of $\tau^*$ almost surely. More precisely, since 
	$$
	A_{2n} \leq \frac{\log n }{n}\sum_{i=1}^n \1_{\{ Y_i\in [\tau_0 - a_n ,\tau_0 ] \}} \Delta_i ,
	$$
	we take $a_n$ such that $p_n\log n \rightarrow 0$ and $p_n \log^2 n \rightarrow \infty$, where 
	\begin{multline*}
	p_n = \p (Y\in [\tau_0-a_n,\tau_0] , \Delta = 1) = \mathbb E \left[ \phi(\gamma_0,X) \int_{[\tau_0-a_n,\tau_0]} F_C([t,\infty)|X,Z) F_u(dt|Z)\right].
	\end{multline*} 
	Then, by Theorem 1(i) from  \cite{Wellner}, we have 
	$$
	\lim_{n\rightarrow \infty }  \frac{1 }{p_n} \frac{1}{n}\sum_{i=1}^n \1_{\{ Y_i\in [\tau_0 - a_n ,\tau_0 ] \}} \Delta_i = 1,\quad a.s.,
	$$
	which implies \eqref{eq:up_A2}. On the other hand, we have
	$$
	A_{1n} \leq \log\left(\sup_{t\in [\tau^*,\tau_0 - a_n  ] } \rho_n(t) \right) \times  \frac{1}{n}\sum_{i=1}^n \1_{\{ Y_i\in [\tau^*,\tau_0   ]\}} \Delta_i .
	$$
	By the same  Theorem 1(i) from  \cite{Wellner},
	$$
	\lim_{n\rightarrow \infty } \sup_{t\in [\tau^*,\tau_0 - a_n  ] }\left[  \rho_n(t)  \frac{\p (Y\in [t,\tau_0 - a_n  ] , \Delta = 1) }{\p (Y\in [t,\tau_0 - a_n  ] ) } \right]  = 1,\quad a.s.
	$$
	By our assumptions, there exists a constant $C_r$, independent of $\tau^*$, $\beta$, $\gamma$ and $\Lambda$, such that 
	$$
	1< \inf_{t\in [\tau^*,\tau_0{-a_n} ] }\frac{\p (Y\in [t,\tau_0 - a_n  ] ) }{\p (Y\in [t,\tau_0 - a_n  ] , \Delta = 1) } < \sup_{t\in [\tau^*,\tau_0{-a_n} ] }\frac{\p (Y\in [t,\tau_0 - a_n  ] ) }{\p (Y\in [t,\tau_0 - a_n  ] , \Delta = 1) } \leq C_r.
	$$
	Gathering facts, deduce with probability 1, for sufficiently large $n$, 
	$$
	\left| r_{1n} (\hat \Lambda_{n};\tau^*) - r_{1n} ( \tilde \Lambda_{0,n} ;\tau^*)\right|   \leq  c  \frac{N^*}{n} ,
	$$
	where $N^*$ be the number of uncensored observations in $[\tau^*,\tau_0]$ and $c$ is some constant (independent of $\tau^*$, $\beta$, $\gamma$ and $\Lambda$). Here,  $N^*$ is a binomial random variable  with $n$ trials and  success probability 
	\begin{multline*}
	p^* = \mathbb P ( Y\geq\tau^* , \Delta = 1) = \mathbb E \left[ \phi(\gamma_0,X) \int_{[\tau^*,\tau_0]} F_C([t,\infty)|X,Z) F_u(dt|Z)\right]\\ \leq \sup_{x} \phi(\gamma_0,x) \mathbb P (Y\geq\tau^* ).
	\end{multline*}

	To bound $r_{3n}=r_{3n}(\Lambda, \beta;\tau^*)$, we note that $\1_{\{B_i=1\}} =  \Delta_i  + (1-\Delta_i )\1_{\{B_i=1\}}$ and rewrite 
	\begin{multline*}
	r_{3n} = \frac{1}{n}\sum_{i=1}^n \1_{\{ Y_i\geq \tau^*  \}} e^{\beta^\prime Z_i}\Delta_i \{  \Lambda (Y_i) - \Lambda (\tau^*-) \}\\- \frac{\Lambda (\tau^*-)}{n}\sum_{i=1}^n \1_{\{ Y_i\geq \tau^*  \}} e^{\beta^\prime Z_i}\1_{\{B_i=1\}} (1- \Delta_i)  = r_{3an}-r_{3bn}.
	\end{multline*}
	On one hand,  
	$$
	r_{3an} \leq  \frac{c_u}{c_l} \times \frac{N^*}{n}.
	$$
	The last inequality is obtained by bounding the jumps of $\Lambda$ and using the following identity: for any integer $M\geq 1$, 
	$$
	\sum_{k=1}^M \sum_{j=k}^M \frac{1}{j} = \sum_{j,k=1}^M  \frac{\1_{\{k\leq j\}} }{j}  = M. 
	$$
	To bound $r_{3bn}$, let us note that 
	$$
	\Lambda (\tau^*-) \leq \frac{1}{c_l} \sum_{j=N^*+1}^N \frac{1}{j} \leq c_1 \log \frac{N}{N^*},
	$$
	with $c_1$ some constant depending only on $c_l$ and the maximal value of the convergent sequence 
	$$\sum_{j=1}^m \frac{1}{j}  - \log m, \quad m\geq 1.$$ Here, $N =\sum_{i=1}^n \Delta_i$ is a binomial random variable  with $n$ trials and success probability
	$$
	p = \mathbb P ( \Delta = 1) = \mathbb E \left[ \phi(\gamma_0,X) \int_{[0,\tau_0]} F_C([t,\infty)|X,Z) F_u(dt|Z)\right]  .
	$$
	Thus
	$$
	r_{3bn} \leq c \log \frac{N}{N^*} \times \frac{1}{n}\sum_{i=1}^n \1_{\{ Y_i\geq \tau^*  \}} \1_{\{B_i=1\}} (1- \Delta_i) =  c \log \frac{N/n}{N^*/n} \times \frac{Q^*}{n},
	$$
	where $Q^*$ is a binomial variable 
	with $n$ trials and success probability 
	\begin{multline*}
	q^* =  \mathbb E \left[ \phi(\gamma_0,X) \int_{[\tau^*,\tau_0]} F_u([t,\tau_0]|X,Z) F_C(dt|X,Z)\right]\\\leq  \mathbb E \left[ \phi(\gamma_0,X)   F_C([\tau^*,\tau_0]|X,Z) F_u([\tau^*,\tau_0]|Z)\right] \\
	\leq \left[ \sup_{x} \phi(\gamma_0,x) \right]
	\left[\sup_{x,z} \frac{F_C([\tau^*,\tau_0]|X=x,Z=z)}{\tau_0 - \tau^*}\right] 
	\times (\tau^* -\tau_0)\times  \mathbb P (T_0\geq \tau^*) \\ \leq c  (\tau^* -\tau_0)\times  \mathbb P (T_0\geq \tau^*) ,
	\end{multline*}
	and $c$ is some constant. By the strong Law of Large Numbers, 
	$$
	\lim_{n\rightarrow \infty} \log \frac{N}{N^*}  = \log \frac{p}{p^*}, \, a.s.
	$$
	
	Next, to bound $ r_{2n}= r_{2n} (\beta;\tau^*) $, we write
	$$
	r_{2n} = \frac{1}{n}\sum_{i=1}^n \1_{\{ {\tau^*\leq T_0<\tau_0} \}} e^{\beta^\prime Z_i} \Delta_i 
	\leq c_u \frac{N^*}{n}.
	$$ 
	
	Finally, to control $ r_{4n}= r_{4n} (\Lambda, \beta, \gamma;\tau^*)$, 
	since $\1_{\{B_i=0\}} = (1-\Delta_i) \1_{\{B_i=0\}}$ and $\log(1+u)\leq u,$ $\forall u\geq 0$,
	we have 
	\begin{multline*}
	r_{4n} = \frac{1}{n}\sum_{i=1}^n \1_{\{\tau_0\geq Y_i \geq  \tau^*\}} \1_{\{B_i=0\}}  \log\left\{1+\frac{\phi_i(\gamma)\exp\left(-\Lambda(Y_i)e^{\beta'Z_i}\right)}{1-\phi_i(\gamma)}\right\}  \\
	+ \frac{1}{n}\sum_{i=1}^n \1_{\{Y_i \geq  \tau^*\!\}}  (1-\Delta_i ) \1_{\{B_i=1\}}\log\left\{1-\phi_i(\gamma)+\phi_i(\gamma)\exp\left(-\Lambda(Y_i)e^{\beta'Z_i}\right)\right\}.
	\end{multline*}
	Thus 
	\begin{multline*}
	|r_{4n}| \leq \sup_{\gamma,x}\left| \frac{\phi(\gamma,x)}{1-\phi(\gamma,x)}\right| \exp\left(-c_l \Lambda(\tau^*-)\right)  \frac{1}{n}\sum_{i=1}^n \1_{\{Y_i \geq  \tau^*\}}  \1_{\{B_i=0\}}
	\\ + 
	\sup_{\gamma,x}\left| \log\left\{1-\phi(\gamma,x)\right\}\right| \times  \frac{1}{n}\sum_{i=1}^n  \1_{\{Y_i \geq  \tau^*\!\}} (1-\Delta_i )\1_{\{B_i=1\}}\\
	= c_1\exp\left(-c_l \Lambda(\tau^*-)\right)   \frac{R^*}{n} + c_2  \frac{Q^*}{n},
	\end{multline*}
	where $R^*$ is a binomial variable 
	with $n$ trials and success probability 
	$$
	r^* =  \mathbb E \left[ \{1-\phi(\gamma_0,X)\}  F_C([\tau^*,\tau_0]|X,Z)\right],
	$$
	and $c_1$ and $c_2$ are some constants. 
	Deduce that there exists a constant $c_3$ such that
	$$
	|r_{4n}| \leq c_3 \left( 
	\frac{R^*}{n} +   \frac{Q^*}{n}\right).
	$$
	
	Gathering facts, there exists a constants $C^*$ and 
	$c^*$, independent of $\tau^*$, $\beta$, $\gamma$ and $\Lambda$, such that 
	$$
	\left| \hat{l}_n(\beta,   \Lambda,  \gamma) - \hat{l}_n^*(\beta,  \Lambda_{|\tau^*}, \gamma) \right| \leq C^* \left\{ \frac{N^*}{n} + \frac{Q^*}{n}\left[1 + \log  {\frac{n}{N^*}}\right]  +  \frac{R^*}{n}  \right\} + o_{a.s.} (1), 
	$$
	where $N^*$, $Q^*$ and $R^*$ are binomial  with $n$ trials and success probabilities $p^*$, $q^*$ and $r^*$, respectively, and
	$$
	p^* + q^* \leq  c^*\p (T_0\geq \tau^*)  \quad \text{and}\quad r^* \leq c^* \p (C \in [\tau^*, \tau_0]) .
	$$
	
\end{proof}


\begin{lemma}
	\label{lem:consist_full_3}
	Assume that for any $x$ and $z$, the conditional distribution of the censoring times given $X=x$ and $Z=z$
	is such that there exists a constant $C>0$ such that
	$$
	\inf_{[t_1,t_2]\subset [0,\tau_0]} \inf_{x,z}\{ F_C(t_2|x,z)- F_C(t_1|x,z) \}>C(t_2-t_1) ,\qquad \forall \delta>0.
	$$
	Let $0<\bar{\tau}<\tau_0$ and $\epsilon>0$. There exist  $c_1, c_2>0$, $\delta>0$ such that $c_1+c_2=1$ and 
	\[
	\begin{split}
	&\inf\Bigg\{\mathbb E [\ell (Y^*,\Delta^*,X,Z;\beta_0,\Lambda_{0|\tau^*} ,\gamma_0)]-	\mathbb E [\ell (Y^*,\Delta^*,X,Z; \beta, \Lambda_{|\tau^*},\gamma_0)]:\\
	&\left.\qquad\quad \bar\tau+\delta\leq \tau^*<\tau_0,\quad \|\beta-\beta_0\|\geq c_1\epsilon\quad\text{ or}\quad \sup\limits_{t\in[0,\bar\tau]}|\Lambda(t)-\Lambda_0(t)|\geq c_2\epsilon\right\}>0
	\end{split}
	\]
\end{lemma}
\begin{proof}
	Note that, for any $\tau^*\in(\bar\tau,\tau_0)$,
	\[
	\mathbb E [\ell (Y^*,\Delta^*,X,Z;\beta_0,\Lambda_{0|\tau^*} ,\gamma_0)]-	\mathbb E [\ell (Y^*,\Delta^*,X,Z; \beta, \Lambda_{|\tau^*},\gamma_0)]
	\]
	is the Kullback-Leibler divergence $KL(\p|Q)$, where $\p$ and $Q$ are the probability measures of $(Y^*,\Delta^*,X,Z)$ when the true parameters are $(\Lambda_0,\beta_0,\gamma_0)$ and $(\Lambda,\beta,\gamma_0)$ respectively. By Pinsker's inequality, we have
	\[
	KL(\p|Q)\geq 2\delta(\p,Q)^2,
	\]
	where $\delta(\p,Q)$ is the total variation distance between the two probability measures, defined as 
	\[
	\delta(\p,Q)=\sup_{A}|\p(A)-Q(A)|,
	\]
	where the supremum is taken over all measurable sets $A$. We want to find a positive lower bound for $\delta(\p,Q)$ independent of $\tau^*$ and $Q$, for all $Q$ such that $\|\beta-\beta_0\|\geq c_1\epsilon$  or $ \sup_{t\in[0,\bar\tau]}|\Lambda(t)-\Lambda_0(t)|\geq c_2\epsilon$. Hence, it is sufficient to find $k>0$ and for each such $Q$ an event $A$, {which could depend on $Q$}, for which $|\p(A)-Q(A)|>k.$  Without loss of generality we can assume  that the covariate vector $Z$ has mean zero. 
	
	\bigskip
	\textit{Case 1.} If $\sup_{t\in[0,\bar\tau]}|\Lambda(t)-\Lambda_0(t)|\geq c_2\epsilon$, there exists $\bar{t}\in[0,\bar\tau] $ such that either
	\begin{equation}\label{eq:cas1}
	\Lambda(\bar{t})\geq\Lambda_0(\bar{t})+c_2\epsilon
	\end{equation}
	or 
	\begin{equation}\label{eq:cas2}
	\Lambda(\bar{t})\leq\Lambda_0(\bar{t})-c_2\epsilon.
	\end{equation} 
We first  consider  \eqref{eq:cas1} and define 
	$$
	\delta=\min\left
	\{\frac{\tau_0-\bar\tau}{2} ,\;\frac{c_2\epsilon}{2\sup_{t\in[0, (\bar\tau+\tau_0)/2]}\lambda_0(t)}\right\}.
	$$   
	It follows that for all 
	$t\in [\bar{t},\bar{t}+\delta]\subset [0, (\bar\tau+\tau_0)/2] \subset[0,{\tau_0})  $ we have $\Lambda(t)\geq\Lambda_0(t)+\frac12c_2\epsilon$. Indeed, we can write 
	\begin{equation*}
	\Lambda(t)\geq \Lambda(\bar{t})\geq\Lambda_0(\bar{t})+c_2\epsilon\geq \Lambda_0(t)-\delta\sup_{u\in[0,(\bar\tau+\tau_0)/2]}\lambda_0(u)+c_2\epsilon\geq \Lambda_0(t)+\frac12c_2\epsilon,\quad \forall t\in [\bar{t},\bar{t}+\delta] .
	\end{equation*}
	Since  $Z$ has mean zero, $(\beta-\beta_0)^\prime Z$ also has zero mean. 
	Moreover, since $B$ is compact and $Z$ is bounded non degenerated variance, we have
	\begin{equation}\label{whynot}
	\inf_{\beta\in B} \p((\beta-\beta_0)^\prime Z\geq 0)>0\quad \text{ and } \quad \inf_{\beta\in B} \p((\beta-\beta_0)^\prime Z\leq 0)>0
	\end{equation}
	{(see proof below)}.
	Let $A_\beta$ be the event $\{\Delta^*=0, Y^*\in[\bar{t},\bar{t}+\delta],{(\beta-\beta_0)^\prime Z}\geq 0\}$, which depends on $\beta$ and thus on $Q$. However, by \eqref{whynot} and the construction of the model, the event $A_\beta$ has positive probability which stays bounded away from zero. Moreover, we have
	\[
	\begin{split}
{\p(A_\beta)-Q(A_\beta)}&=\iint\limits_{{(\beta-\beta_0)^\prime z\geq 0}}\int_{\bar{t}}^{\bar{t}+\delta}\phi(\gamma_0,x)\\
	&\qquad \quad\times \left\{\exp\left(-\Lambda_0(t)e^{\beta_0^\prime z}\right)-\exp\left(-\Lambda(t)e^{\beta ^\prime z}\right)\right\} F_C(\dd t|x,z) G(\dd x,\dd z).
	\end{split}
	\]
	Whenever $(\beta-\beta_0)^\prime z\geq 0$, by the mean value theorem, we obtain
	\[
	\begin{split}
	\exp\left(-\Lambda_0(t)e^{\beta_0^\prime z}\right)-\exp\left(-\Lambda(t)e^{\beta ^\prime  z}\right)
	&=\left\{\Lambda(t)e^{\beta^\prime z}-\Lambda_0(t)e^{\beta_0 ^\prime z}\right\}e^{-\xi}\\
	&=\left[\{\Lambda(t)-\Lambda_0(t)\}e^{\beta_0^\prime z}+\Lambda(t)\{e^{\beta^\prime z}-e^{\beta_0 ^\prime z}\}\right] e^{-\xi}\\
	& { \geq \{\Lambda(t)-\Lambda_0(t)\}e^{\beta_0^\prime z} e^{-\xi},}
	\end{split}
	\]
	for some $\xi>0$ such that $|\xi-\Lambda_0(t)e^{\beta_0^\prime z}|\leq |\Lambda(t)e^{\beta^\prime z}-\Lambda_0(t)e^{\beta_0^\prime z}|$, {$t \in[\bar{t}, \bar{t}+\delta]$. }
	Now, let 
	\[
	M(t)=\Lambda_0(t)\frac{\sup_{\beta,z}e^{\beta^\prime z}}{\inf_{\beta,z}e^{\beta ^\prime z}}+\frac{\log 2}{\inf_{\beta,z}e^{\beta^\prime  z}},\quad {t \in[\bar{t}, \bar{t}+\delta]. }
	\]
	Then, for $(\beta-\beta_0)^\prime z\geq 0$ and $t\in[\bar{t},\bar{t}+\delta],$ such that  $\Lambda(t)\leq M(t)$ we simply use \eqref{eq:cas1} and write 
	\[
	\exp\left(-\Lambda_0(t)e^{\beta_0^\prime z}\right)-\exp\left(-\Lambda(t)e^{\beta^\prime  z}\right)\geq \frac12 c_2\epsilon e^{\beta_0^\prime z}e^{-\xi}\geq k_1\epsilon ,
	\]
	for some constant $k_1>0$ independent of $\Lambda$, $\beta$ and the event $A_\beta$, because $M(t)$ is uniformly bounded on $[0,(\bar{\tau}+\tau_0)/2]$ and thus $e^{-\xi}$ is bounded away from zero.  On the other hand, for $t\in[\bar{t},\bar{t}+\delta]$ such that  $\Lambda(t)>M(t)$, we have 
	\[
	\exp\left(-\Lambda(t)e^{\beta^\prime z}\right)\leq \exp\left(-M(t)\inf_{\beta,z}e^{\beta ^\prime z}\right)\leq \frac{1}{2}\exp\left(-\Lambda_0(t)e^{\beta_0^\prime z}\right).
	\] 
	Consequently,
	\begin{multline*}
	\exp\left(-\Lambda_0(t)e^{\beta_0^\prime z}\right)-\exp\left(-\Lambda(t)e^{\beta ^\prime z}\right)\geq \frac{1}{2}\exp\left(-\Lambda_0(t)e^{\beta_0 ^\prime z}\right) \\ \geq \frac{1}{2}\exp\left(-\Lambda_0((\bar\tau+\tau_0)/2)\sup_{\beta,z}e^{\beta_0 ^\prime z}\right)=k_2(\bar\tau)>0.
	\end{multline*}
	We conclude that, for any $t\in [\bar{t},\bar{t}+\delta]$,
	\[
	\exp\left(-\Lambda_0(t)e^{\beta_0^\prime z}\right)-\exp\left(-\Lambda(t)e^{\beta ^\prime z}\right)\geq {\min\{k_1\epsilon , k_2(\bar\tau)\}}>0.
	\]
	It follows that 
	\begin{multline*}
	|\p(A_\beta)\! -Q(A_\beta)|\geq {\min\{k_1\epsilon , k_2(\bar\tau)\}}\inf_x\phi(\gamma_0,x)\\
	\times \iint\limits_{{(\beta-\beta_0)^\prime z\geq 0}} \left\{F_C(\bar{t}+\delta|x,z)\!-F_C(\bar{t}|x,z)\right\} G(\dd x,\dd z).
	\end{multline*}
	By assumption  we have  
	\[
	\inf_{x,z}F_C(\bar{t}+\delta|x,z)-F_C(\bar{t}|x,z)\geq C \delta,
	\] 
	yielding  that there exist another constant $k_3>0$  independent of $\Lambda$, $\beta$ and the event {$A_\beta$ (but depending on $\epsilon$ and $\bar \tau$)} such that
	\[
	\forall \beta\in B,\quad |\p(A_\beta)-Q(A_\beta)|\geq k_3 \inf_{\beta\in B} \p((\beta-\beta_0)^\prime Z\geq 0)
	>0.
	\]
	Note that {the uniform lower bound} holds for any choice of the constants $c_1$ and $c_2$ in the statement of the Lemma. 
	
 We next consider  \eqref{eq:cas2}. Let
	$$
	{\bar\delta}=\min\left\{\frac{\bar \tau}{2},\, \frac{c_2\epsilon}{2\sup_{t\in[0,\bar\tau]}\lambda_0(t)}  \right\}.
	$$   
	It follows that for all 
	$t\in [\, \bar{t}-{\bar\delta},\bar{t}\, ]  $ we have 
	$\Lambda(t)\leq\Lambda_0(t)-\frac12c_2\epsilon$. Indeed, we can write 
	\begin{multline*}
	\Lambda(t)\leq \Lambda(\bar{t}) \leq \{ \Lambda_0(\bar{t}) - \Lambda_0(t) \} + \Lambda_0(t) - c_2\epsilon\\
	\leq {\bar\delta}\sup_{u\in[0,\bar{\tau}]}\lambda_0(u)+ \Lambda_0(t) - c_2\epsilon \leq 
	\Lambda_0(t)-\frac12c_2\epsilon,\quad \forall t\in [\, \bar{t}-{\bar\delta},\bar{t}\, ] .
	\end{multline*}
	Next we redefine  $A_\beta$ as the event $\{\Delta^*=0, Y^*\in[\bar{t}-{\bar\delta},\bar{t}\,],{(\beta-\beta_0)^\prime Z}\leq 0\}$, which depends on $\beta$ and thus on $Q$. However, by \eqref{whynot} and the construction of the model, the event $A_\beta$ has positive probability which stays bounded away from zero. Moreover, we have
	\[
	\begin{split}
	\p(A_\beta)-Q(A_\beta)&=\iint\limits_{{(\beta-\beta_0)^\prime z\leq 0}}\int_{\bar{t}-{\bar\delta}}^{\bar{t}}\phi(\gamma_0,x)\\
	&\qquad \quad\times \left\{\exp\left(-\Lambda_0(t)e^{\beta_0^\prime z}\right)-\exp\left(-\Lambda(t)e^{\beta ^\prime z}\right)\right\} F_C(\dd t|x,z) G(\dd x,\dd z).
	\end{split}
	\]
	Whenever $(\beta-\beta_0)^\prime z\leq 0$, by the mean value theorem, we obtain
	\[
	\begin{split}
	\exp\left(-\Lambda_0(t)e^{\beta_0^\prime z}\right)-\exp\left(-\Lambda(t)e^{\beta ^\prime  z}\right)
	\leq \{\Lambda(t)-\Lambda_0(t)\}e^{\beta_0^\prime z} e^{-\xi} \leq - \frac12 c_2\epsilon e^{\beta_0^\prime z}e^{-\xi},
	\end{split}
	\]
	for some $\xi>0$ such that $|\xi-\Lambda_0(t)e^{\beta_0^\prime z}|\leq |\Lambda(t)e^{\beta^\prime z}-\Lambda_0(t)e^{\beta_0^\prime z}|\leq 2\Lambda_0(t)e^{\beta_0^\prime z} $, $t \in[\bar{t}-{\bar\delta}, \bar{t}\,]$. Thus necessarily $0<\xi \leq 2\Lambda_0(\bar\tau)e^{\beta_0^\prime z}$, and thus $e^{-\xi}$ stays away from zero. 
	Using arguments as we used for the case \eqref{eq:cas2}, we deduce that $\p(A_\beta)-Q(A_\beta)$ is negative and away from zero. Thus we obtain the result with 
	$\bar\tau \leq \tau^*<\tau_0$ instead of $\bar\tau+\delta \leq \tau^*<\tau_0$.
	Finally, it remains to recall that $\inf$ is a decreasing function of nested sets. Now the arguments for \emph{Case 1} are complete for any choice of the constants $c_1$ and $c_2$ in the statement of the Lemma.

	\bigskip
	\textit{Case 2.} If $\sup_{t\in[0,\bar\tau]}|\Lambda(t)-\Lambda_0(t)|\leq c_2\epsilon$, then necessarily $\|\beta-\beta_0\|\geq c_1\epsilon.$ In particular we also have that $\Lambda(\bar{\tau})\leq\Lambda_0(\bar{\tau})+c_2\epsilon$, so all such functions $\Lambda$ are uniformly bounded on $[0,\bar{\tau}]$. Without loss of generality we can also assume that 
	$\Lambda_0(\bar\tau/2)\geq 1$ (otherwise we can take a larger $\bar\tau$).
	Note that 
	$$
	Var ((\beta-\beta_0)^\prime Z) = (\beta-\beta_0)^\prime Var (Z) (\beta-\beta_0) \geq ( c_1\epsilon)^2 \lambda_{min}   ,
	$$
	with $\lambda_{min} $ the smallest eigenvalue of $Var(Z)$. From this lower bound for the variance of $(\beta-\beta_0)^\prime Z$, and since $Z$ is centered and has a bounded support, we have
	\begin{equation}\label{whynot2}
	\inf_{|\beta-\beta_0|\geq c_1\epsilon}\left[ \p((\beta-\beta_0)^\prime Z\geq z_0)+ \p((\beta-\beta_0)^\prime Z\leq - z_0)\right] >\frac{( c_1\epsilon)^2 \lambda_{min}}{2\sup \|Z\|^2},
	\end{equation}
	for $z_0= c_1\epsilon \lambda^{1/2}_{min}/2  $ {(see proof below)}.
	If $$\inf_{|\beta-\beta_0|\geq c_1\epsilon}\p\left((\beta-\beta_0)^\prime Z\geq z_0\right)> \frac{( c_1\epsilon)^2\lambda_{min}}{2\sup \|Z\|^2},$$ let $A_\beta$ be the event $\{\Delta^*=0,Y^*\in[\bar\tau/2,\bar\tau],(\beta-\beta_0)^\prime Z \geq z_0\}$. 
	By \eqref{whynot2} and the construction of the model, the event $A_\beta$ has positive probability which stays bounded away from zero.
	Next, as in \emph{Case 1},  we write
	\[
	\begin{split}
	\p(A_\beta) - Q(A_\beta)&= \iint\limits_{(\beta-\beta_0)^\prime  z \geq z_0}\int_{\frac12\bar{\tau}}^{\bar{\tau}}\phi(\gamma_0,x)\\
	&\qquad   \qquad \times \left\{\exp\left(-\Lambda_0(t)e^{\beta_0^\prime z}\right)-\exp\left(-\Lambda(t)e^{\beta^\prime z}\right)\right\} F_C(\dd t|x,z)\, G(\dd x,\dd z) ,
	\end{split}
	\]
	and
	\[
	\begin{split}
	\exp\left(-\Lambda_0(t)e^{\beta_0^\prime z}\right)-\exp\left(-\Lambda(t)e^{\beta^\prime z}\right)=\left[\{\Lambda(t)-\Lambda_0(t)\}e^{\beta_0^\prime z}+\Lambda(t)\{e^{\beta ^\prime z}-e^{\beta_0 ^\prime z}\}\right] e^{-\xi},
	\end{split}
	\]
	for some $\xi>0$ such that $|\xi-\Lambda_0(t)e^{\beta_0 ^\prime z}|\leq |\Lambda(t)e^{\beta^\prime z}-\Lambda_0(t)e^{\beta_0^\prime z}|$. From the boundedness of $\beta$, $z$,  $\Lambda$ and $\Lambda_0$ on $[0,\bar\tau]$, it follows that $e^{-\xi}\geq k_4>0$ for some $k_4$ independent of $\Lambda,$ $\beta$ and the event $A_\beta$ (but depending on $\bar\tau$).  Moreover, since for $t\in[\bar\tau/2,\bar\tau]$, $(\beta-\beta_0)^\prime  z\geq z_0$, 
	\[
	\left|\Lambda(t)-\Lambda_0(t)\right|e^{\beta_0^\prime z}\leq c_2\epsilon e^{\beta_0^\prime z},
	\]
	and 
	\[
	\begin{split}
	\Lambda(t)\{e^{\beta^\prime z}-e^{\beta_0^\prime z}\} 
	\geq \Lambda_0 (\bar\tau/2) e^{\beta_0^\prime z}\{   e^{(\beta-\beta_0)^\prime z} - 1\}\geq e^{\beta_0^\prime z}z_0 =  e^{\beta_0^\prime z}\lambda^{1/2}_{min}  \,c_1\epsilon /2,
	\end{split}
	\]
	we obtain 
	\[
	\exp\left(-\Lambda_0(t)e^{\beta_0^\prime z}\right)-\exp\left(-\Lambda(t)e^{\beta ^\prime z}\right)\geq \epsilon \left[ \lambda^{1/2}_{min}  \,c_1/2 - c_2\right] e^{\beta_0^\prime z}e^{-\xi} .
	\]
	Define
	\[
	c_1= \frac{4}{\lambda_{min}^{1/2} + 4}\qquad\text{and}\qquad c_2= \frac{\lambda_{min}^{1/2}}{\lambda_{min}^{1/2} + 4},
	\]
	such that $0<c_1,c_2<1$ and $c_1+c_2=1$, and deduce that
	\[
	\exp\left(-\Lambda_0(t)e^{\beta_0^\prime z}\right)-\exp\left(-\Lambda(t)e^{\beta ^\prime z}\right)\geq   \epsilon \lambda^{1/2}_{min}   e^{\beta_0^\prime z}k_4 \left\{ \lambda_{min}^{1/2} + 4 \right\}^{-1} .
	\]
	Deduce that, for any $\|\beta-\beta_0\|\geq c_1\epsilon$,
	\begin{multline*}
	\p(A_\beta)-Q(A_\beta)\geq \epsilon\left[ \lambda^{1/2}_{min}   e^{\beta_0^\prime z}k_4  \left\{ \lambda_{min}^{1/2} + 4\right\}^{-1} \right]\inf_{z}e^{\beta_0^\prime z}\inf_x\phi(\gamma_0,x) \\ \times \inf_{|\beta-\beta_0|\geq c_1\epsilon} \p(C\in[\bar\tau/2,\bar\tau],(\beta-\beta_0)^\prime Z \geq z_0)>0.
	\end{multline*}
	Note that this bound does not depend on $\delta$ in the statement of the Lemma. Finally it is easy to see that similar arguments apply when $\inf_{|\beta-\beta_0|\geq c_1\epsilon}\p((\beta-\beta_0)^\prime Z\leq - z_0)> ( c_1\epsilon)^2 \lambda_{min}/\{4\sup \|Z\|^2\} $, for the same expression of $z_0$. In this case, we define 
	$A_\beta= \{\Delta^*=0,Y^*\in[\bar\tau/2,\bar\tau],(\beta-\beta_0)^\prime Z \leq -z_0\}$ and follow the same steps as above to show that 
	\[
	\exp\left(-\Lambda_0(t)e^{\beta_0^\prime z}\right)-\exp\left(-\Lambda(t)e^{\beta ^\prime z}\right) <0,
	\]
	and the difference of exponentials stays away from zero.  Now, the proof of Lemma \ref{lem:consist_full_3} is complete. 
\end{proof}


\begin{proof}[\textsc{Proof of Equation \eqref{whynot}.}]
	For $\beta=\beta_0$ we have $\p((\beta-\beta_0)^\prime Z\geq 0)=1$ and thus we only have to study $\beta\neq \beta_0$. Since  $Var(Z)$ is non degenerated, $\p((\beta-\beta_0)^\prime Z= 0)<1$ for any $\beta\neq \beta_0$. Next, we could write 
	\begin{multline*}
	0 = \mathbb E \left( \frac{(\beta-\beta_0)^\prime Z}{\|\beta-\beta_0\|}  \right) = \mathbb E  \left( \frac{(\beta-\beta_0)^\prime Z}{\|\beta-\beta_0\|}  \1_{\{ (\beta-\beta_0)^\prime Z \geq 0\}  }\right) + 
	\mathbb E  \left( \frac{(\beta-\beta_0)^\prime Z}{\|\beta-\beta_0\|}  \1_{ \{ (\beta-\beta_0)^\prime Z < 0 \} }\right) \\ \leq \|Z\| \p((\beta-\beta_0)^\prime Z\geq 0) +\mathbb E  \left( \frac{(\beta-\beta_0)^\prime Z}{\|\beta-\beta_0\|} \1_{\{ (\beta-\beta_0)^\prime Z < 0 \} }\right) .
	\end{multline*}
	It remains to notice that 
	$$\sup_{\|\tilde \beta\|=1}  \mathbb E  \left( \tilde \beta^\prime Z\1_{\{ \tilde \beta^\prime Z < 0 \} }\right)<0.$$
	Indeed, if $\mathbb E  \left( \tilde \beta^\prime Z\1_{\{ \tilde \beta^\prime Z < 0 \} }\right)$, which is negative,  could be arbitrarily close to zero, and since $\mathbb E  \left( \tilde \beta^\prime Z\1_{\{ \tilde \beta^\prime Z < 0 \} }\right)= - \mathbb E  \left( \tilde \beta^\prime Z\1_{\{ \tilde \beta^\prime Z \geq 0 \} }\right)$, we deduce that $\mathbb E  (| \tilde \beta^\prime Z|)$ could be arbitrarily close to zero, for suitable $ \tilde \beta$ with unit norm. Since the support of $Z$ is bounded and 
	$$
	\lambda_{min} (Var(Z)) \leq \mathbb E  (| \tilde \beta^\prime Z|^2)\leq \mathbb E  (| \tilde \beta^\prime Z|)\sup \|Z\|, \quad \forall  \|\tilde \beta\|=1,
	$$
	we thus get a contradiction with the assumption that $\lambda_{min} $, the smallest eigenvalue of $Var(Z)$, is positive. 
	Deduce that \eqref{whynot} holds true. \end{proof}


\quad

\begin{proof}[\textsc{Proof of Equation \eqref{whynot2}.}]
	It suffices to prove the following property.  
	Let $U$ be a centered variable such that $|U|\leq M$ for some constant $M$ and $Var(U)$ is bounded from below by some constant $0<C<M^2$. Then there exists $z_0>0$ such that $\p(|U|\geq z_0)> C/2M^2$ with $z_0$ depending on $M$ and $C$ but independent of the law of $U$. 
	
	For any $0<z_0 <M$ we can write 
	$$
	C\leq \mathbb E (U^2) = \mathbb E \left( U^2 \1 _{ \{ |U| \geq z_0 \} } \right) +
	\mathbb E \left( U^2 \1 _{ \{ |U| < z_0 \} } \right)\\
	\leq M^2 \p (|U| \geq z_0 ) + z_0^2 \{ 1 - \p (|U| \geq   z_0 ) \}.
	$$
	Deduce that 
	$$
	\p (|U| \geq   z_0 ) \geq \frac{C-z_0^2}{M^2-z_0^2}.
	$$
	Finally, it suffices to take $z_0^2 \leq  CM^2/(2M^2 - C)$ in order to obtain 
	$$
	\p (|U| \geq   z_0 ) \geq \frac{C}{2M^2}.
	$$
	Since $C$ could be arbitrarily small,  we could take $z_0^2 =  C/4$. \end{proof}
\subsection{Asymptotic normality}
\begin{proof}[\textsc{Proof of Theorem \ref{theo:asymptotic_normality}.}]
Let $m$ be as in \eqref{def:m} and define $M_n(\gamma,\pi)=\frac{1}{n}\sum_{i=1}^{n}m(X_i;\gamma,\pi)$ and $M(\gamma,\pi)=\E[m(X;\gamma,\pi)]$. Note that $m(x;\gamma_0,\pi_0)=0$ and $M(\gamma_0,\pi_0)=0$. If we take partial derivatives with respect to  $\gamma$ of the vector $M(\gamma,\pi)$ evaluated at $(\gamma_0,\pi_0)$ we obtain a matrix $\nabla_\gamma M(\gamma_0,\pi_0)$ with elements
\[
\nabla_\gamma M(\gamma_0,\pi_0)_{kl}=\E\left[\frac{\partial}{\partial \gamma_l}\left\{\frac{1-\pi(X)}{\phi(\gamma,X)}-\frac{\pi(X)}{1-\phi(\gamma,X)} \right\}\frac{\partial}{\partial \gamma_k}\phi(\gamma,X)\right]\Bigg|_{\gamma_0,\pi_0}, 
\]
for $k,l\in\{1,\dots,p\}$, because 
\[
\frac{1-\pi_0(x)}{\phi(\gamma_0,x)}-\frac{\pi_0(x)}{1-\phi(\gamma_0,x)}=0\quad\forall x\in\X.
\]
Hence, 
\begin{equation}
\label{def:gradient_M}
\Gamma_1:=\nabla_\gamma M(\gamma_0,\pi_0)=-\E\left[W(X)\nabla_\gamma\phi(\gamma_0,X)\nabla_\gamma\phi(\gamma_0,X)'\right]
\end{equation}
where
\[
W(X)=\frac{1-\pi_0(X)}{\phi(\gamma_0,X)^2}+\frac{\pi_0(X)}{\left(1-\phi(\gamma_0,X)\right)^2}=\frac{1}{\phi(\gamma_0,X)}+\frac{1}{1-\phi(\gamma_0,X)}>0.
\]
We will also use the Gateaux derivative of $M(\gamma,\pi_0)$ in a direction $[\pi-\pi_0]$ given by
{\small\begin{equation}
\label{eqn:directional_derivative}
\begin{split}
\Gamma_2(\gamma,\pi_0)[\pi-\pi_0]&:=\nabla_\pi M(\gamma,\pi_0)[\pi-\pi_0]\\
&=\lim_{h\to 0}\frac{1}{h}\left[M(\gamma,\pi_0+h(\pi-\pi_0))-M(\gamma,\pi_0)\right]\\
&=-\lim_{h\to 0}\frac{1}{h}\E\left[\left\{\frac{h\left\{\pi(X)-\pi_0(X)\right\}}{\phi(\gamma,X)}+\frac{h\left\{\pi(X)-\pi_0(X)\right\}}{1-\phi(\gamma,X)} \right\}\nabla_\gamma\phi(\gamma,X)\right]\\
&=-\E\left[\left\{\pi(X)-\pi_0(X)\right\}\left\{\frac{1}{\phi(\gamma,X)}+\frac{1}{1-\phi(\gamma,X)} \right\}\nabla_\gamma\phi(\gamma,X)\right].
\end{split}
\end{equation} }               
	We apply Theorem 2 in \cite{chen} so we need to verify its conditions. Consistency of $\hat\gamma_n$ is shown in Theorem \ref{theo:consistency}, while condition (2.1) in \cite{chen} is satisfied by construction since 
	\[
	\Vert M_n(\hat\gamma_n,\hat\pi_n)\Vert=0=\inf_{\gamma\in G}\Vert M_n(\gamma,\hat\pi_n)\Vert.
	\] 
	Note that assumption (AC1) was needed in Theorem \ref{theo:consistency} in order to obtain almost sure convergence. However, here we only need convergence in probability for which (AN4)-(ii) suffices. 
	For condition (2.2) in \cite{chen}, the derivative of $M$ with respect to  $\gamma$ is computed in \eqref{def:gradient_M} and the matrix is negative definite (as a result also full rank) because of our assumption (AN3).  Moreover the directional derivative was computed in \eqref{eqn:directional_derivative} and for $(\gamma,\pi)\in G_{\delta_n}\times\Pi_{\delta_n}$ with $G_{\delta_n}=\{\gamma\in G: \Vert\gamma-\gamma_0\Vert\leq\delta_n\}$, $\Pi_{\delta_n}=\{\pi\in \Pi: \Vert\pi-\pi_0\Vert_{\infty}\leq\delta_n\}$, $\delta_n=o(1)$, we have
	\[
	\begin{split}
	&\left\Vert M(\gamma,\pi)-M(\gamma,\pi_0)-\Gamma_2(\gamma,\pi_0)[\pi-\pi_0]\right\Vert\\
	&=\left\Vert\E\left[\left\{\frac{\pi_0(X)-\pi(X)}{\phi(\gamma,X)} +\frac{\pi_0(X)-\pi(X)}{1-\phi(\gamma,X)} \right\}\nabla_\gamma\phi(\gamma,X)\right]\right.\\
	&\quad+\left.\E\left[\left\{\pi(X)-\pi_0(X)\right\}\left\{\frac{1}{\phi(\gamma,X)} +\frac{1}{1-\phi(\gamma,X)}\right\}\nabla_\gamma\phi(\gamma,X)\right]\right\Vert=0,
	\end{split}
	\]
	which means that condition (2.3i) is satisfied. For condition (2.3ii), we have 
	\[
	\begin{split}
	&\Gamma_2(\gamma,\pi_0)[\pi-\pi_0]-\Gamma_2(\gamma_0,\pi_0)[\pi-\pi_0]\\
	&=-\E\left[\left\{\pi(X)-\pi_0(X)\right\}\left\{\left( \frac{1}{\phi(\gamma,X)} +\frac{1}{1-\phi(\gamma,X)}\right)\nabla_\gamma\phi(\gamma,X)\right.\right.\\
	&\quad\qquad-\left.\left. \left( \frac{1}{\phi(\gamma_0,X)} +\frac{1}{1-\phi(\gamma_0,X)}\right)\nabla_\gamma\phi(\gamma_0,X)\right\}\right].
	\end{split}
	\]
	Then, from $\sup_x|\pi(x)-\pi_0(x)|\leq\delta_n$, $|\gamma-\gamma_0|\leq\delta_n\to 0$ and (AN1),
	it follows that 
	\[
	\Vert\Gamma_2(\gamma,\pi_0)[\pi-\pi_0]-\Gamma_2(\gamma_0,\pi_0)[\pi-\pi_0]\Vert\leq o(1)\delta_n.
	\]
	Conditions (2.4) and (2.6) in \cite{chen} are satisfied thanks to our assumption (AN4) because
	\begin{equation}
	\label{eqn:M_n-Gamma_2}
	\begin{split}
	&M_n(\gamma_0,\pi_0)+\Gamma_2(\gamma_0,\pi_0)[\hat\pi-\pi_0]\\
	&=\frac{1}{n}\sum_{i=1}^n\left\{\frac{1-\pi_0(X_i)}{\phi(\gamma_0,X_i)}-\frac{\pi_0(X_i)}{1-\phi(\gamma_0,X_i)}\right\}\nabla_\gamma\phi(\gamma_0,X_i)\\
	&\quad+ \E^*\left[\left\{\hat\pi(X)-\pi_0(X)\right\}\left(\frac{1}{\phi(\gamma_0,X)} +\frac{1}{1-\phi(\gamma_0,X)}\right)\nabla_\gamma\phi(\gamma_0,X)\right] \\
	&=\E^*\left[\left\{\hat\pi(X)-\pi_0(X)\right\}\left(\frac{1}{\phi(\gamma_0,X)} +\frac{1}{1-\phi(\gamma_0,X)}\right)\nabla_\gamma\phi(\gamma_0,X)\right].
	\end{split}
	\end{equation}
	Then we conclude by central limit theorem that 
	\[
	\sqrt{n}\left(M_n(\gamma_0,\pi_0)+\Gamma_2(\gamma_0,\pi_0)[\hat\pi-\pi_0]\right)\xrightarrow{d} N(0,V)
	\]
	where $V=Var(\Psi(Y,\Delta,X,Z))$.  It remains to deal with condition (2.5), which is a consequence of Theorem 3 in \cite{chen} and assumption (AN2) because from (AN1) we have
	\[
	\begin{split}
	\Vert m(x;\gamma_1,\pi_1)-m(x;\gamma_2,\pi_2)\Vert
	&\leq \left\Vert\left(\frac{1-\pi_1(x)}{\phi(\gamma_1,x)} +\frac{\pi_1(x)}{1-\phi(\gamma_1,X)}\right)\nabla_\gamma\phi(\gamma_1,X)\right.\\
	&\left.\qquad\qquad-\left(\frac{1-\pi_2(x)}{\phi(\gamma_2,x)} +\frac{\pi_2(x)}{1-\phi(\gamma_2,X)}\right)\nabla_\gamma\phi(\gamma_2,X)\right\Vert\\
	&\leq C_1\Vert\gamma_1-\gamma_2\Vert+C_2\Vert\pi_1-\pi_2\Vert_\infty.
	\end{split}
	\]
	Finally, the asymptotic normality  follows from Theorem 2 in \cite{chen} and the asymptotic covariance matrix is given by 
	\begin{equation}
	\label{def:Sigma_gamma}
	\Sigma_\gamma=(\Gamma'_1\Gamma_1)^{-1}\Gamma'_1V\Gamma_1(\Gamma'_1\Gamma_1)^{-1}=\Gamma_1^{-1}V\Gamma_1^{-1}
	\end{equation}
	\end{proof}
\begin{proof}[\textsc{Proof of Theorem \ref{theo:asymptotic_normality2}.}]
	We  show that conditions 1 and 4 of Theorem 4 in \cite{Lu} are satisfied. Define $S_n$ as the version of $\hat{S}_n$ where $\hat\gamma_n$ is replaced by $\gamma_0$ 
\[
\begin{split}
{S}_n(\hat\Lambda_n,\hat\beta_n)(h_1,h_2)&=\frac{1}{n}\sum_{i=1}^n \Delta_i\1_{\{Y_i<\tau_0\}} \left[h_1(Y_i)+h'_2Z_i\right]\\
&\quad-\frac{1}{n}\sum_{i=1}^n \left\{ \Delta_i+(1-\Delta_i)\1_{\{Y_i\leq\tau_0\}}g_i(Y_i,\hat\Lambda_n,\hat\beta_n,\gamma_0) \right\}\\
&\qquad\qquad\quad\times\left\{e^{\hat\beta'_nZ_i}\int_0^{Y_i}h_1(s)\dd\hat\Lambda_n(s)+e^{\hat\beta'_nZ_i}\hat\Lambda_n(Y_i)h'_2Z_i
\right\}
\end{split}
\]
\textit{Condition 1.} We start by writing 
\begin{equation}
\label{eqn:condition1}
\begin{split}
\hat{S}_n(\Lambda_0,\beta_0)(h_1,h_2)-{S}(\Lambda_0,\beta_0)(h_1,h_2)&=\left[\hat{S}_n(\Lambda_0,\beta_0)(h_1,h_2)-{S}_n(\Lambda_0,\beta_0)(h_1,h_2)\right]\\
&\quad+\left[{S}_n(\Lambda_0,\beta_0)(h_1,h_2)-{S}(\Lambda_0,\beta_0)(h_1,h_2)\right].
\end{split}
\end{equation}
For the second term on the right hand side we have
\begin{equation}
\label{eqn:condition1_2}
{S}_n(\Lambda_0,\beta_0)(h_1,h_2)-{S}(\Lambda_0,\beta_0)(h_1,h_2)=\int f_h(y,\delta,x,z)\,\dd(\p_n-\p)(y,\delta,x,z)
\end{equation}
where
\[
\begin{split}
f_h(y,\delta,x,z)&=h_2'z\left\{\delta\1_{\{y<\tau_0\}}-\left[\delta-(1-\delta)\1_{\{y\leq\tau_0\}}g(y,\Lambda_0,\beta_0,\gamma_0)\right]e^{\beta'_0z}\Lambda_0(y) \right\}\\
&+\delta \1_{\{y<\tau_0\}}h_1(y)-\left[\delta -(1-\delta)\1_{\{y\leq\tau_0\}}g(y,\Lambda_0,\beta_0,\gamma_0)\right]e^{\beta'_0z}\int_0^{y}h_1(s)\dd\Lambda_0(s).
\end{split}
\]
The classes  $\{h_2\in\R^q,\Vert h_2\Vert\leq \mathfrak{m}\} $, $\{h_1\in BV[0,\tau_0],\, \Vert h_1\Vert_v\leq \mathfrak{m}\}$ and  $$\left\{\int_0^y h_1(t)\,\dd\Lambda_0(t), \,h_1\in BV[0,\tau_0], \Vert h_1\Vert_v\leq \mathfrak{m}\right\}$$ are Donsker classes (the last one because it consists of monotone bounded functions). As in \cite{Lu}, because of the boundedness of the covariates and $\Lambda_0$, it follows that $\{f_h(y,\delta,x,z),\,h\in\mathcal{H}_\mathfrak{m}\}$ is also a Donsker class since it is sum of products of Donsker classes with fixed uniformly bounded functions.  

On the other hand, for the first term on the right hand side of \eqref{eqn:condition1}, we have 
\begin{equation}
\label{eqn:condition1_4}
\begin{split}
&\left[\hat{S}_n(\Lambda_0,\beta_0)(h_1,h_2)-{S}_n(\Lambda_0,\beta_0)(h_1,h_2)\right]\\
&=-\frac1n\sum_{i=1}^n (1-\Delta_i)\1_{\{Y_i\leq\tau_0\}}\left\{e^{\beta'_0Z_i}\int_0^{Y_i}h_1(s)\dd\Lambda_0(s)+e^{\beta'_0Z_i}\Lambda_0(Y_i)h'_2Z_i
\right\}\\
&\qquad\qquad\qquad\qquad\qquad\qquad\qquad\qquad\times\left\{g_i(Y_i,\Lambda_0,\beta_0,\hat\gamma_n)-g_i(Y_i,\Lambda_0,\beta_0,\gamma_0)\right\}\\
&=-\frac1n\sum_{i=1}^n (1-\Delta_i)\1_{\{Y_i\leq\tau_0\}}e^{\beta'_0Z_i}\left\{\int_0^{Y_i}h_1(s)\dd\Lambda_0(s)+\Lambda_0(Y_i)h'_2Z_i
\right\}\\
&\quad\qquad\qquad\qquad\qquad\qquad\times\frac{\partial g_i}{\partial\phi}(Y_i,\Lambda_0,\beta_0,\gamma_0)\left\{\phi(\hat\gamma_n,X_i)-\phi(\gamma_0,X_i)\right\}+o_P(n^{-1/2}),
\end{split}
\end{equation}
where 
\[
\begin{split}
\frac{\partial g_i}{\partial\phi}(Y_i,\Lambda_0,\beta_0,\gamma_0)&=\frac{\exp\left(-\Lambda(Y_i) e^{\beta'Z_j}\right)}{1-\phi(\gamma,X_j)+\phi(\gamma,X_j)\exp\left(-\Lambda(Y_i) e^{\beta'Z_j}\right)}\\
&\quad+\frac{\phi(\gamma,X_j)\exp\left(-\Lambda(Y_i)e^{\beta'Z_j}\right)\left[\exp\left(-\Lambda(Y_i) e^{\beta'Z_j}\right)-1\right]}{\left[1-\phi(\gamma,X_j)+\phi(\gamma,X_j)\exp\left(-\Lambda(Y_i) e^{\beta'Z_j}\right)\right]^2}.
\end{split}
\]
In order to conclude that the remainder term is of order $o_P(n^{-1/2})$ we use 
\[
\sup_x\left|\phi(\hat\gamma_n,x)-\phi(\gamma_0,x)\right|\leq \sup_{\gamma\in G, x\in\X}\left\Vert\nabla_\gamma\phi(\gamma,x)\right\Vert |\hat\gamma_n-\gamma_0|=O_P(n^{-1/2})
\]
and  the fact that $\frac{\partial^2 g_i}{\partial\phi^2}(Y_i,\Lambda_0,\beta_0,\gamma)$ and  $$(1-\Delta_i)\1_{\{Y_i\leq\tau_0\}}e^{\beta'_0Z_i}\left\{\int_0^{Y_i}h_1(s)\dd\Lambda_0(s)+\Lambda_0(Y_i)h'_2Z_i
\right	\}$$  are uniformly bounded functions thanks to our assumptions on  $Z,$ $\Lambda$, $\Phi$ and $h$. 
From the same assumptions we also obtain 
\begin{equation}
\label{eqn:condition1_5}
\begin{split}
&\frac1n\sum_{i=1}^n (1-\Delta_i)\1_{\{Y_i\leq\tau_0\}}e^{\beta'_0Z_i}\left\{\int_0^{Y_i}h_1(s)\dd\Lambda_0(s)+\Lambda_0(Y_i)h'_2Z_i
\right\}\\
&\qquad\qquad\qquad\qquad\qquad\times\frac{\partial g_i}{\partial\phi}(Y_i,\Lambda_0,\beta_0,\gamma_0)\left\{\phi(\hat\gamma_n,X_i)-\phi(\gamma_0,X_i)\right\}\\
&=\frac1n\sum_{i=1}^n (1-\Delta_i)e^{\beta'_0Z_i}\left\{\int_0^{Y_i}h_1(s)\dd\Lambda_0(s)+\Lambda_0(Y_i)h'_2Z_i
\right\}\\
&\qquad\qquad\qquad\qquad\qquad\times\frac{\partial g_i}{\partial\phi}(Y_i,\Lambda_0,\beta_0,\gamma_0)\nabla_\gamma\phi(\gamma_0,X_i)'\left\{\hat\gamma_n-\gamma_0\right\}+o_P(n^{-1/2})\\
&=\E\left[(1-\Delta)\1_{\{Y\leq\tau_0\}}e^{\beta'_0Z}\left\{\int_0^{Y}h_1(s)\dd\Lambda_0(s)+\Lambda_0(Y)h'_2Z
\right\}\right.\\
&\left.\qquad\qquad\qquad\qquad\qquad\times\frac{\partial g}{\partial\phi}(Y,\Lambda_0,\beta_0,\gamma_0)\nabla_\gamma\phi(\gamma_0,X)'\right]\left(\hat\gamma_n-\gamma_0\right)+o_P(n^{-1/2}).
\end{split}
\end{equation}
and the expectation term is uniformly bounded. To prove the asymptotic normality of $\hat\gamma_n-\gamma_0$ in Theorem \ref{theo:asymptotic_normality} we used Theorem 2 in \cite{chen}. Going through the proof of Theorem 2 in \cite{chen}, we actually have 
\[
(\hat\gamma_n-\gamma_0)=-(\Gamma_1'\Gamma_1)^{-1}\Gamma'_1\left\{M_n(\gamma_0,\pi_0)+\Gamma_2(\gamma_0,\pi_0)[\hat\pi-\pi_0]\right\}+ o_P(n^{-1/2})
\]
where $\Gamma_1$ is defined in \eqref{def:gradient_M} and  $\Gamma_2$ in \eqref{eqn:directional_derivative}. From Assumption (AN4-iii) and \eqref{eqn:M_n-Gamma_2}, it follows that 
\begin{equation}
\label{eqn:condition_1_6}
(\hat\gamma_n-\gamma_0)=-(\Gamma_1'\Gamma_1)^{-1}\Gamma'_1	\int \Psi(y,\delta,x)\,\d(\p_n-\p)(y,\delta,x,z)+ o_P(n^{-1/2}). 
\end{equation}
Putting together \eqref{eqn:condition1}-\eqref{eqn:condition_1_6}, we have
\[
\begin{split}
&\hat{S}_n(\Lambda_0,\beta_0)(h_1,h_2)-{S}(\Lambda_0,\beta_0)(h_1,h_2)\\
&=\int \left\{ f_h(y,\delta,x,z)-Q_h\Gamma_1^{-1} \Psi(y,\delta,x)\right\}\,\dd (\p_n-\p)(y,\delta,x,z)+o_P(n^{-1/2})
\end{split}
\]
where
\[
Q_h=\E\left[(1-\Delta)\1_{\{Y\leq\tau_0\}}e^{\beta'_0Z}\left\{\int_0^{Y}h_1(s)\dd\Lambda_0(s)+\Lambda_0(Y)h'_2Z
\right\}\frac{\partial g}{\partial\phi}(Y,\Lambda_0,\beta_0,\gamma_0)\nabla_\gamma\phi(\gamma_0,X)'\right].
\]
In order to conclude the convergence of $\sqrt{n}(\hat{S}_n(\Upsilon_0)-S(\Upsilon_0))$ to a Gaussian process $G^*$, we need to have that $\{Q_h\Gamma_1^{-1} \Psi(y,\delta,x), h\in\mathcal{H}_\mathfrak{m}\}$ is a bounded Dosker class of functions (since sum of bounded Donsker classes is also Donsker). We can write 
\[
\begin{split}
&Q_h\Gamma_1^{-1}\Psi(y,\delta,x)\\
&=h'_2\E\left[(1-\Delta)\1_{\{Y\leq\tau_0\}}Ze^{\beta'_0Z}\Lambda_0(Y)
\frac{\partial g}{\partial\phi}(Y,\Lambda_0,\beta_0,\gamma_0)\nabla_\gamma\phi(\gamma_0,X)'\right]\Gamma_1^{-1}\Psi(y,\delta,x)\\
&\quad+\int_0^{\tau_0}\E\left[(1-\Delta)\1_{\{Y\leq\tau_0\}}e^{\beta'_0Z}\frac{\partial g}{\partial\phi}(Y,\Lambda_0,\beta_0,\gamma_0)\nabla_\gamma\phi(\gamma_0,X)'\right]h_1(s)\dd\Lambda_0(s) \Gamma_1^{-1} \Psi(y,\delta,x)
\end{split}
\]
By assumption (AN1) and $\inf_{x}H((\tau_0,\infty)|x)>0$, $\Lambda_0(\tau_0)<\infty$ and the boundedness of the covariates,  we have that $$\E\left[(1-\Delta)\1_{\{Y\leq\tau_0\}}Ze^{\beta'_0Z}\Lambda_0(Y)
\frac{\partial g}{\partial\phi}(Y,\Lambda_0,\beta_0,\gamma_0)\nabla_\gamma\phi(\gamma_0,X)'\right]\Gamma_1^{-1} \Psi(y,\delta,x)$$ is uniformly bounded. Hence 
\[
\begin{split}
&\left\{h'_2\E\left[(1-\Delta)\1_{\{Y\leq\tau_0\}}Ze^{\beta'_0Z}\Lambda_0(Y)
\frac{\partial g}{\partial\phi}(Y,\Lambda_0,\beta_0,\gamma_0)\nabla_\gamma\phi(\gamma_0,X)'\right]\Gamma_1^{-1} \Psi(y,\delta,x):\right.\\
&\,\qquad\qquad\qquad\qquad\qquad\qquad \qquad\qquad\qquad \qquad\qquad\qquad h_2\in\R^q, \Vert h_2\Vert_{L_1}\leq \mathfrak{m}\bigg \}
\end{split}
\]
is a Donsker class (see Example 2.10.10 in \cite{VW}). It can also be shown that, since $h_1$ belongs to the class of bounded functions with bounded variation and all the other terms are uniformly bounded,  that 
\[
\begin{split}
&\left\{\int_0^{\tau_0}\E\left[(1-\Delta)\1_{\{Y\leq\tau_0\}}e^{\beta'_0Z}\frac{\partial g}{\partial\phi}(Y,\Lambda_0,\beta_0,\gamma_0)\nabla_\gamma\phi(\gamma_0,X)'\right]h_1(s)\dd\Lambda_0(s) \Gamma_1^{-1} \Psi(y,\delta,x), \right.\\
&\qquad\qquad\qquad\qquad\qquad\qquad \qquad\qquad\qquad \qquad\qquad\qquad h_1\in BV[0,\tau_0], \,\Vert h_1\Vert_v\leq \mathfrak{m}\bigg\}
\end{split}
\]
is also a bounded Donsker class (covering numbers of order $\epsilon$ of $\{ h_1\in BV[0,\tau_0], \Vert h_1\Vert_v\leq \mathfrak{m}\}$ correspond to covering numbers of order $c\epsilon$ for some constant $c>0$).

The limit process $G^*$ has mean zero because
\[
\E[f_h(y,\delta,x,z)]=S(\Upsilon_0)(h)=0\qquad\text{and}\qquad\E[\Psi(Y,\Delta,X)]=0.
\]
The covariance process of $G^*$ is 
\begin{equation}
\label{eqn:cov_G*}
\begin{split}
&Cov\left(G^*(h),G^*(\tilde{h})\right)\\
&=\E\left[\left\{f_h(Y,\Delta,X,Z)-Q_h\Gamma_1^{-1} \Psi(Y,\Delta,X)\right\}\left\{f_{\tilde{h}}(Y,\Delta,X,Z)-Q_{\tilde{h}}\Gamma_1^{-1} \Psi(Y,\Delta,X)\right\}\right]\\
&=\E[f_h(Y,\Delta,X,Z)f_{\tilde{h}}(Y,\Delta,X,Z)]-Q_h\Gamma_1^{-1}\E[f_{\tilde{h}}(Y,\Delta,X,Z)\Psi(Y,\Delta,X)]\\
&\quad-Q_{\tilde{h}}\Gamma_1^{-1}\E[f_h(Y,\Delta,X,Z)\Psi(Y,\Delta,X)]+Q_{\tilde{h}}\Gamma_1^{-1}\E[\Psi(Y,\Delta,X)\Psi(Y,\Delta,X)']\Gamma_1^{-1}Q'_h.
\end{split}
\end{equation}
\textit{Condition 4 of Theorem 4 in \cite{Lu}. } As for condition 1, we write
\begin{equation}
\label{eqn:cond4}
\begin{split}
\sqrt{n}\left\{(\hat{S}_n-S)(\Upsilon_n)-(\hat{S}_n-S)(\Upsilon_0)\right\}&=\sqrt{n}\left\{({S}_n-S)(\Upsilon_n)-({S}_n-S)(\Upsilon_0)\right\}\\
&\quad+\sqrt{n}\left\{(\hat{S}_n-S_n)(\Upsilon_n)-(\hat{S}_n-S_n)(\Upsilon_0)\right\}
\end{split}
\end{equation}
For the second term in the right hand side of \eqref{eqn:cond4}, similarly to \eqref{eqn:condition1_4}-\eqref{eqn:condition1_5}, we have
\[
\begin{split}
&(\hat{S}_n-S_n)(\Upsilon_n)-(\hat{S}_n-S_n)(\Upsilon_0)\\
&=\frac1n\sum_{i=1}^n (1-\Delta_i)\1_{\{Y_i\leq\tau_0\}}e^{\hat\beta'_nZ_i}\left\{\int_0^{Y_i}h_1(s)\dd\hat\Lambda_n(s)+\hat\Lambda_n(Y_i)h'_2Z_i
\right\}\\
&\quad\qquad\qquad\qquad\qquad\qquad\frac{\partial g_i}{\partial\phi}(Y_i,\hat\Lambda_n,\hat\beta_n,\gamma_0)\nabla_\gamma\phi(\gamma_0,X_i)'\left\{\hat\gamma_n-\gamma_0\right\}\\
&\quad-\frac1n\sum_{i=1}^n (1-\Delta_i)\1_{\{Y_i\leq\tau_0\}}e^{\beta'_0Z_i}\left\{\int_0^{Y_i}h_1(s)\dd\Lambda_0(s)+\Lambda_0(Y_i)h'_2Z_i
\right\}\\
&\quad\qquad\qquad\qquad\qquad\qquad\frac{\partial g_i}{\partial\phi}(Y_i,\Lambda_0,\beta_0,\gamma_0)\nabla_\gamma\phi(\gamma_0,X_i)'\left\{\hat\gamma_n-\gamma_0\right\}+o_P(n^{-1/2})\\
\end{split}
\]
Using the boundedness in probability of $\hat\beta_n$ and $\hat\Lambda_n(\tau_0)$, the boundedness of the covariates, $\beta_0$, $\Lambda_0(\tau)$, $\nabla_\gamma\phi(\gamma,x)$ and the consistency results in Theorem \ref{theo:consistency2}, it follows that 
\[
\begin{split}
&\left|\frac1n\sum_{i=1}^n (1-\Delta_i)\1_{\{Y_i\leq\tau_0\}}\left\{e^{\hat\beta'_nZ_i}\left(\int_0^{Y_i}h_1(s)\dd\hat\Lambda_n(s)+\hat\Lambda_n(Y_i)h'_2Z_i
\right)\frac{\partial g_i}{\partial\phi}(Y_i,\hat\Lambda_n,\hat\beta_n,\gamma_0)\right.\right.\\
&\left.\left.\quad-e^{\beta'_0Z_i}\left(\int_0^{Y_i}h_1(s)\dd\Lambda_0(s)+\Lambda_0(Y_i)h'_2Z_i
\right)\frac{\partial g_i}{\partial\phi}(Y_i,\Lambda_0,\beta_0,\gamma_0)\right\}\nabla_\gamma\phi(\gamma_0,X_i)\right|=o_P(1)
\end{split}
\]
As a consequence, since $\hat\gamma_n-\gamma_0=O_P(n^{-1/2})$, we obtain
\[
\sqrt{n}\left\{(\hat{S}_n-S_n)(\Upsilon_n)-(\hat{S}_n-S_n)(\Upsilon_0)\right\}=o_P(1).
\]
It remains to deal with the first term on the right hand side of \eqref{eqn:cond4}. It suffices to show that, for any sequence $\epsilon_n\to  0$,
\[
\sup_{|\Lambda-\Lambda_0|_\infty\leq\epsilon_n, \,\Vert\beta-\beta_0\Vert\leq\epsilon_n}\frac{\left|({S}_n-S)(\Upsilon)-({S}_n-S)(\Upsilon_0)\right|}{n^{-1/2}\vee \Vert\beta-\beta_0\Vert\vee |\Lambda-\Lambda_0|_\infty }=o_P(1).
\]
Let
\[
a_1(y,\delta,z,x)=\delta e^{\beta'z}\int_0^{y}h_1(s)\dd\Lambda(s)-\delta e^{\beta'_0z}\int_0^{y}h_1(s)\dd\Lambda_0(s),
\]
\[
a_2(y,\delta,z,x)=\delta e^{\beta'z}\Lambda(y)h'_2z-\delta e^{\beta'_0z}\Lambda_0(y)h'_2z,
\]
and
\[
\begin{split}
a_3(y,\delta,z,x)&=(1-\delta)\1_{\{y\leq\tau_0\}}g(y,\Lambda,\beta,\gamma_0)\left\{e^{\beta'z}\int_0^{y}h_1(s)\dd\Lambda(s)+e^{\beta'z}\Lambda(y)h'_2z\right\}\\
&\quad-(1-\delta)\1_{\{y\leq\tau_0\}}g(y,\Lambda_0,\beta_0,\gamma_0)\left\{e^{\beta'_0z}\int_0^{y}h_1(s)\dd\Lambda_0(s)+e^{\beta'_0z}\Lambda_0(y)h'_2z
\right\}
\end{split}
\]
Then, we have  
\[
\begin{split}
({S}_n-S)(\Upsilon)-({S}_n-S)(\Upsilon_0)&=-\frac{1}{n}\sum_{i=1}^n \left\{a_1(Y_i,\Delta_i,Z_i,X_i)-\E\left[a_1(Y,\Delta,Z,X)\right]  \right\}\\
&\quad-\frac{1}{n}\sum_{i=1}^n \left\{a_2(Y_i,\Delta_i,Z_i,X_i)-\E\left[a_2(Y,\Delta,Z,X)\right]  \right\}\\
&\quad-\frac{1}{n}\sum_{i=1}^n \left\{a_3(Y_i,\Delta_i,Z_i,X_i)-\E\left[a_3(Y,\Delta,Z,X)\right]  \right\} 
\end{split}
\]
Next we consider the first term. The other two can be handled similarly. From a Taylor expansion we have
\[
\begin{split}
&\delta e^{\beta'z}\int_0^{y}h_1(s)\dd\Lambda(s)-\delta e^{\beta'_0z}\int_0^{y}h_1(s)\dd\Lambda_0(s)\\
&=(\beta-\beta_0)'z\delta e^{\beta'_0z}\int_0^{y}h_1(s)\dd\Lambda(s)+\delta e^{\beta'_0z}\int_0^{y}h_1(s)\dd(\Lambda-\Lambda_0)(s)+o(\Vert\beta-\beta_0\Vert).
\end{split}
\]Hence
\begin{equation}
\label{eqn:a1}
\begin{split}
&\frac{1}{n}\sum_{i=1}^n \left\{a_1(Y_i,\Delta_i,Z_i,X_i)-\E\left[a_1(Y,\Delta,Z,X)\right]  \right\}\\
&\leq (\beta-\beta_0)\left\{\frac{1}{n}\sum_{i=1}^n Z_i\Delta_i e^{\beta'_0Z_i}\int_0^{Y_i}h_1(s)\dd\Lambda(s)-\E\left[Z\Delta e^{\beta'_0Z}\int_0^{Y}h_1(s)\dd\Lambda(s)\right]+o(1)\right\}\\
&\quad+ \frac{1}{n}\sum_{i=1}^n\left\{\Delta_i e^{\beta'_0Z_i}\int_0^{Y_i}h_1(s)\dd(\Lambda-\Lambda_0)(s)-\E\left[\Delta e^{\beta'_0Z}\int_0^{Y}h_1(s)\dd(\Lambda-\Lambda_0)(s)\right]\right\}.
\end{split}
\end{equation}
By the law of large numbers
\[
\frac{1}{n}\sum_{i=1}^n\left\{ Z_i\Delta_i e^{\beta'_0Z_i}\int_0^{Y_i}h_1(s)\dd\Lambda(s)-\E\left[Z\Delta e^{\beta'_0Z}\int_0^{Y}h_1(s)\dd\Lambda(s)\right] \right\}=o_P(1)
\]
and as a result, the first term in the right hand side of \eqref{eqn:a1} is $o_P(\Vert\beta-\beta_0\Vert)$. The second term can be rewritten as 
\[
\int_0^{\tau_0} D_n(s)h_1(s)\dd(\Lambda-\Lambda_0)(s)
\]
where
\[
D_n(s)=\frac{1}{n}\sum_{i=1}^n\left\{\Delta_i\1_{\{Y_i\geq s\}} e^{\beta'_0Z_i}-\E\left[\Delta\1_{\{Y\geq s\}} e^{\beta'_0Z}\right]\right\}. 
\]
By integration by parts and the chain rule we have
\[
\begin{split}
&\int_0^{\tau_0} D_n(s)h_1(s)\dd(\Lambda-\Lambda_0)(s)\\
&=  D_n(\tau_0)h_1(\tau_0)(\Lambda-\Lambda_0)(\tau_0)- \int_0^{\tau_0} (\Lambda-\Lambda_0)(s)\dd \left[D_n(s)h_1(s)\right]\\
&=D_n(\tau_0)h_1(\tau_0)(\Lambda-\Lambda_0)(\tau_0)-\int_0^{\tau_0} (\Lambda-\Lambda_0)(s) D_n(s)\dd h_1(s)\\
&\quad +\frac{1}{n}\sum_{i=1}^n\left\{ \Delta_i (\Lambda-\Lambda_0)(Y_i)h_1(Y_i) e^{\beta'_0Z_i}-\E\left[ (\Lambda-\Lambda_0)(Y)h_1(Y)\Delta e^{\beta'_0Z}\right]\right\}
\end{split} 
\]
Note that 
\[
\begin{split}
\E\left[ (\Lambda-\Lambda_0)(Y)h_1(Y)\Delta e^{\beta'_0Z}\right]&=\E\left[e^{\beta'_0Z}\int_0^{\tau_0}(\Lambda-\Lambda_0)(s)h_1(s)\dd H_1(s|X,Z)\bigg|X,Z\right]\\
&=-\int_0^{\tau_0}(\Lambda-\Lambda_0)(s)h_1(s)\dd\E\left[\Delta\1_{\{Y\geq s\}}e^{\beta'_0Z}\right].
\end{split}
\]
It can be shown that $\sqrt{n}D_n$ converges weakly to a tight, mean zero Gaussian process $D$ in $l^\infty([0,\tau_0])$. Since $h_1$ is bounded, it follows 
\[
\frac{\left|D_n(\tau_0)h_1(\tau_0)(\Lambda-\Lambda_0)(\tau_0)\right|}{\Vert\Lambda-\Lambda_0\Vert_{\infty}}=o_P(1)
\]
Moreover, since $D_n\to 0$ and $h_1$ is of bounded variation 
\[
\frac{ \left|\int_0^{\tau_0} (\Lambda-\Lambda_0)(s) D_n(s)\dd h_1(s)\right|}{\Vert\Lambda-\Lambda_0\Vert_{\infty}}\leq \sup_{t\in[0,\tau_0]}|D_n(s)| \int_0^{\tau_0}\,\left|\dd h(s)\right|\to 0.
\] 
Finally, since  $\left\{g_{\Lambda}(y,\delta,z)=\delta(\Lambda-\Lambda_0)(y)h_1(y)e^{\beta'_0Z}:\,\Vert\Lambda-\Lambda_0\Vert_{\infty}\leq \epsilon_n\right\}$ is a Donsker class (product of bounded variation functions, uniformly bounded) and 
\[
\E\left[(\Lambda-\Lambda_0)(Y)^2h_1(Y)^2\Delta e^{2\beta'_0Z}\right]=O(\epsilon_n^2)=o(1)
,
\] 
we have that
\[
\sqrt{n}\frac{1}{n}\sum_{i=1}^n\left\{ \Delta_i (\Lambda-\Lambda_0)(Y_i)h_1(Y_i) e^{\beta'_0Z_i}-\E\left[ (\Lambda-\Lambda_0)(Y)h_1(Y)\Delta e^{\beta'_0Z}\right]\right\}
\]
converges to zero in probability. So we obtain 
\[
\frac{ \int_0^{\tau_0} D_n(s)h_1(s)\dd(\Lambda-\Lambda_0)(s)}{n^{-1/2}\vee \Vert\Lambda-\Lambda_0\Vert_{\infty} }=o_P(1)
\] 
The other two terms related to $a_2$ and $a_3$ can be treated similarly. 

This concludes the verification of conditions of Theorem 4 in \cite{Lu} (or Theorem 3.3.1 in \cite{VW}). Hence, the weak convergence of $\sqrt{n}(\Upsilon_n-\Upsilon_0)$ to a tight, mean zero Gaussian process $G$ follows. Next we compute the covariance process of $G$.  
From Theorem 3.3.1 in \cite{VW} we have
\begin{equation}
\label{eqn:asymptotic_relation}
-\sqrt{n}\dot{S}(\Upsilon_0)(\Upsilon_n-\Upsilon_0)(h)=\sqrt{n}(\hat{S}_n(\Upsilon_0)-S(\Upsilon_0))(h)+o_P(1).
\end{equation}
Moreover, in \cite{Lu} it is computed that 
\begin{equation}
\label{eqn:derivative}
\dot{S}(\Upsilon_0)(\Upsilon_n-\Upsilon_0)(h)=\int_0^{\tau_0}\sigma_1(h)(t)\dd(\hat\Lambda_n(t)-\Lambda_0(t))+(\hat\beta_n-\beta_0)'\sigma_2(h)
\end{equation}
where $\sigma=(\sigma_1,\sigma_2)$ is a continuous linear operator from $\mathcal{H}_\mathfrak{m}$ to $\mathcal{H}_\mathfrak{m}$ of the form 
\[
\begin{split}
\sigma_1(h)(t)&=\E\left[\1_{\{Y\geq t\}}V(t,\Upsilon_0)(h)g(t,\Upsilon_0)e^{\beta'_0Z}\right]\\
&\quad-\E\left[\int_t^{\tau_0}\1_{\{Y\geq s\}}V(t,\Upsilon_0)(h)g(s,\Upsilon_0)\{1-g(s,\Upsilon_0)\}e^{2\beta'_0Z}\dd\Lambda_0(s)\right]
\end{split}
\]
and
\[
\sigma_2(h)(t)=\E\left[\int_0^{\tau_0}\1_{\{Y\geq t\}}W(t,\Upsilon_0)V(t,\Upsilon_0)(h)g(t,\Upsilon_0)e^{\beta'_0Z}\dd\Lambda_0(t)\right]
\]
where
\[
V(t,\Upsilon_0)(h)=h_1(t)-\left\{1-g(t,\Upsilon_0)\right\}e^{\beta'_0Z}\int_0^th_1(s)\dd\Lambda_0(s)+h'_2W(t,\Upsilon_0)
\]
and
\[
W(t,\Upsilon_0)=\left[1-\left\{1-g(t,\Upsilon_0)\right\}e^{\beta'_0Z}\Lambda_0(t)\right]Z
\]
In \cite{Lu}, it is also shown that $\sigma$ is invertible with inverse $ \sigma^{-1}=(\sigma_1^{-1},\sigma_2^{-1})$.  Hence, for all $g\in\mathcal{H}_\mathfrak{m}$, let $h=\sigma^{-1}(g)$. If in \eqref{eqn:derivative} we replace $h$ by $\sigma^{-1}(g)$ and use \eqref{eqn:asymptotic_relation}, we obtain
\[
\begin{split}
&\int_0^{\tau_0}g_1(t)\dd\sqrt{n}(\Lambda_n(t)-\Lambda_0(t))+\sqrt{n}(\hat\beta_n-\beta_0)'g_2\\&=-\sqrt{n}(\hat{S}_n(\Upsilon_0)-S(\Upsilon_0))(\sigma^{-1}(g))+o_P(1)\xrightarrow{d}- G^*(\sigma^{-1}(g)).
\end{split}
\]
Since the previous results holds for all $g\in\mathcal{H}_\mathfrak{m}$, it follows that $(\sqrt{n}(\hat\Lambda_n-\Lambda_0),\sqrt{n}(\hat\beta_n-\beta_0))$ converges to a tight mean zero Gaussian process $G$ with covariance 
\begin{equation}
\label{def:cov_G}
Cov(G(g),G(\tilde{g}))=Cov\left(G^*(\sigma^{-1}(g)),G^*(\sigma^{-1}(\tilde{g}))\right)
\end{equation}
and the covariance of $G^*$ is given in \eqref{eqn:cov_G*}.
\end{proof}
\section{Additional simulation results}
In this section we report the simulation results for scenario 2 of  the models~{1-4}, $n=200, 400$, that were omitted from the main paper and the results for sample size $n=1000$ (all models and scenarios). In addition, Table \ref{tab:results3_4_beta} complements Tables~\ref{tab:results3} and~\ref{tab:results4} containing results for $\hat\beta$.

\clearpage
\begin{table}
	\caption{\label{tab:results1_2:2}Bias, variance and MSE of $\hat\gamma$ and $\hat\beta$ for \texttt{smcure} (second rows) and our approach (first rows) in Model 1 and 2 (scenario 2).}
	\centering
	\scalebox{0.85}{
		\fbox{
			\begin{tabular}{ccccrrrrrrrrr}
				&	&&&\multicolumn{3}{c}{Cens. level 1}&\multicolumn{3}{c}{Cens. level 2}&\multicolumn{3}{c}{Cens. level 3}\\
				Mod.&	n&  scen. & Par. &  Bias & Var. & MSE & Bias & Var. & MSE & Bias & Var. & MSE\\[2pt]
				\hline
				&	& & & & & & & & & && \\[-8pt]
	1&	$200$ &	 $2 $ & $\gamma_1 $ & $0.005 $  & $0.032 $ & $0.032  $ & $-0.005  $ & $0.037 $ & $0.037 $ & $ -0.005 $ & $0.045 $ & $0.045 $\\
						&	& &  & $0.013 $  & $0.032 $ & $0.032  $ & $ 0.009 $ & $0.038 $ & $ 0.038$ & $ 0.015 $ & $0.046 $ & $0.046 $\\
						&	& & $\gamma_2 $ & $-0.024 $  & $0.093 $ & $ 0.094 $ & $-0.035  $ & $0.116 $ & $0.118 $ & $-0.018  $ & $ 0.137$ & $0.137 $\\
						&	& & & $0.017 $  & $0.096 $ & $ 0.096 $ & $ 0.016 $ & $0.123 $ & $ 0.123$ & $ 0.046 $ & $0.144 $ & $0.146 $\\
						&	& & $\beta $ & $0.019 $  & $0.033 $ & $0.033  $ & $0.025  $ & $0.035 $ & $0.036 $ & $ 0.011 $ & $0.041 $ & $0.041 $\\
						&	& &  & $0.018 $  & $0.033 $ & $ 0.033 $ & $0.023  $ & $0.036 $ & $0.036 $ & $0.007  $ & $0.042 $ & $0.042 $\\
							\cline{2-13}
							& & & & & & & & & && \\[-8pt]
						&	$400$ & $2 $ & $\gamma_1 $ & $0.002 $  & $ 0.016$ & $0.016  $ & $0.003  $ & $0.018 $ & $0.018 $ & $0.009  $ & $ 0.021$ & $ 0.021$\\
						&	& &  & $0.008 $  & $0.016 $ & $0.016  $ & $ 0.011 $ & $0.019 $ & $0.019 $ & $ 0.020 $ & $0.021 $ & $0.022 $\\
						&	& & $\gamma_2 $ & $-0.017 $  & $0.046 $ & $ 0.047 $ & $-0.016  $ & $ 0.054$ & $0.054 $ & $ -0.013 $ & $0.070 $ & $ 0.070$\\
						&	& & & $0.012 $  & $ 0.046$ & $0.047  $ & $0.018  $ & $0.055 $ & $0.055 $ & $0.025  $ & $0.068 $ & $0.068 $\\
						&	& & $\beta $ & $ 0.001$  & $0.014 $ & $0.014  $ & $0.019  $ & $0.019 $ & $0.019 $ & $0.001  $ & $0.020 $ & $ 0.020$\\
						&	& &  & $0.000 $  & $0.014 $ & $ 0.014 $ & $ 0.017 $ & $0.019 $ & $0.019 $ & $ -0.001 $ & $0.020 $ & $0.020 $\\
							\hline
						&	& & & & & & & & & && \\[-8pt]
						2&	$200$ & $2 $ & $\gamma_1 $ & $-0.004 $  & $0.033 $ & $0.033  $ & $-0.016  $ & $0.039 $ & $0.039 $ & $-0.032  $ & $0.050 $ & $0.051 $\\
						&	& &  & $0.020 $  & $0.034 $ & $0.034  $ & $ 0.037 $ & $0.042 $ & $0.044 $ & $ 0.079 $ & $0.064 $ & $0.070 $\\
						&	& & $\gamma_2 $ & $-0.016 $  & $0.041 $ & $ 0.041 $ & $-0.037  $ & $0.044 $ & $0.045 $ & $-0.054  $ & $0.057 $ & $0.060 $\\
						&	& & & $0.029 $  & $0.044 $ & $ 0.045 $ & $0.031  $ & $0.052 $ & $ 0.053$ & $0.064  $ & $ 0.075$ & $0.079 $\\
						&	& & $\beta $ & $ 0.004$  & $ 0.016$ & $ 0.016 $ & $0.008  $ & $0.017 $ & $0.017 $ & $ 0.008 $ & $0.018 $ & $0.018 $\\
						&	& &  & $0.002 $  & $0.016 $ & $ 0.016 $ & $ 0.003 $ & $0.018 $ & $0.018 $ & $-0.005  $ & $0.019 $ & $ 0.019$\\
							\cline{2-13}
						&	& & & & & & & & & && \\[-8pt]
							&	$400$ &  $2 $ & $\gamma_1 $ & $0.000 $  & $0.018 $ & $ 0.018 $ & $-0.002  $ & $ 0.022$ & $0.022 $ & $-0.015  $ & $0.028 $ & $0.028 $\\
						&	& &  & $ 0.012$  & $0.018 $ & $ 0.018 $ & $ 0.027 $ & $0.022 $ & $ 0.023$ & $ 0.043 $ & $0.030 $ & $0.032 $\\
						&	& & $\gamma_2 $ & $-0.019 $  & $0.021 $ & $ 0.021 $ & $ -0.022 $ & $0.025 $ & $ 0.026$ & $-0.045  $ & $0.031 $ & $0.033 $\\
						&	& & & $ 0.007$  & $0.022 $ & $0.022  $ & $ 0.018 $ & $0.027 $ & $0.027 $ & $ 0.023 $ & $0.033 $ & $0.033 $\\
						&	& & $\beta $ & $0.005 $  & $0.007 $ & $ 0.007 $ & $ 0.010 $ & $0.008 $ & $0.008 $ & $ 0.009 $ & $0.010 $ & $0.010 $\\
						&	& &  & $0.004 $  & $ 0.007$ & $0.007  $ & $ 0.007 $ & $ 0.008$ & $0.008 $ & $ 0.001 $ & $ 0.011$ & $ 0.011$\\
						\end{tabular}
				}}
			\end{table}
				
\begin{table}
	\caption{	\label{tab:results3_4_beta}Bias, variance and MSE of $\hat\beta$ for \texttt{smcure} (second rows) and our approach (first rows) in Model 3 and 4.}
	\centering
	\scalebox{0.85}{
		\fbox{
			\begin{tabular}{ccccrrrrrrrrr}
				&	& & & & & & & & & && \\[-8pt]
				&	&&&\multicolumn{3}{c}{Cens. level 1}&\multicolumn{3}{c}{Cens. level 2}&\multicolumn{3}{c}{Cens. level 3}\\
				Mod.&	n&  scen. & Par. &  Bias & Var. & MSE & Bias & Var. & MSE & Bias & Var. & MSE\\[2pt]
				\hline
				&	& & & & & & & & & && \\[-8pt]
			3	&$200$	&1 & $\beta_1 $ & $-0.013 $  & $0.007 $ & $  0.007$ & $ -0.007 $ & $ 0.006$ & $0.006 $ & $-0.009  $ & $0.008 $ & $0.008 $\\
		&		& &  & $-0.014 $  & $0.007 $ & $0.007  $ & $ -0.009 $ & $0.006 $ & $0.007 $ & $-0.014  $ & $0.008 $ & $0.008 $\\
		&		& & $\beta_2 $ & $0.007 $  & $0.003 $ & $ 0.003 $ & $0.006 $ & $0.003 $ & $ 0.003 $ & $0.005 $ & $0.004 $& $ 0.004 $ \\
		&		& &  & $0.007 $  & $0.003 $ & $0.003  $ & $ 0.008 $ & $0.003 $ & $0.003 $ & $ 0.007 $ & $0.004 $ & $0.004 $\\
		&		& & $\beta_3 $ & $0.016 $  & $0.043 $ & $  0.043$ & $0.004  $ & $ 0.043$ & $0.043 $ & $0.007  $ & $0.053 $ & $0.053 $\\
		&		& &  & $ 0.017$  & $0.043 $ & $ 0.043 $ & $ 0.006 $ & $ 0.043$ & $0.043 $ & $ 0.015 $ & $0.053 $ & $0.054 $\\
					\cline{3-13}
						&		& & & & & & & & & && \\[-8pt]
		&		& 3& $\beta_1 $ & $0.024 $  & $0.014 $ & $0.014  $ & $0.024  $ & $0.014 $ & $0.015 $ & $0.017  $ & $0.017 $ & $ 0.017$\\
		&		& &  & $0.023 $  & $ 0.014$ & $0.014  $ & $ 0.025 $ & $ 0.014$ & $ 0.015$ & $ 0.024 $ & $0.017 $ & $0.018 $\\
		&		& & $\beta_2 $ & $ -0.005$  & $0.004 $ & $0.004  $ & $-0.002  $ & $0.005 $ & $0.005 $ & $-0.003  $ & $0.006 $ & $0.006 $\\
		&		& &  & $-0.005 $  & $0.004 $ & $0.004  $ & $-0.003  $ & $0.005 $ & $0.005 $ & $-0.004  $ & $0.006 $ & $ 0.006$\\
		&		& & $\beta_3 $ & $0.018 $  & $ 0.066$ & $ 0.066 $ & $ -0.002 $ & $0.081 $ & $ 0.081$ & $ 0.021 $ & $ 0.099$ & $ 0.100$\\
		&		& &  & $0.018 $  & $0.066 $ & $ 0.066 $ & $0.001  $ & $0.083 $ & $0.083 $ & $ 0.033 $ & $0.104 $ & $0.105 $\\
					\cline{2-13}
		&		& & & & & & & & & && \\[-8pt]
		& $400$		&1 & $\beta_1 $ & $-0.007 $  & $0.003 $ & $ 0.003 $ & $-0.004  $ & $0.003 $ & $0.003 $ & $ -0.002 $ & $0.003 $ & $0.003 $\\
		&		& &  & $-0.007 $  & $0.003 $ & $ 0.003 $ & $ -0.005 $ & $0.003 $ & $0.003 $ & $ -0.006 $ & $0.003 $ & $0.003 $\\
		&		& & $\beta_2 $ & $0.004 $  & $0.002 $ & $  0.002$ & $0.003  $ & $ 0.002$ & $0.002 $ & $ 0.002 $ & $ 0.002$ & $0.002 $\\
		&		& &  & $0.004 $  & $0.002 $ & $0.002  $ & $  0.003$ & $0.002 $ & $0.002 $ & $ 0.003 $ & $0.002 $ & $0.002 $\\
		&		& & $\beta_3 $ & $ 0.006$  & $ 0.020$ & $0.020  $ & $0.007  $ & $0.022 $ & $0.022 $ & $ 0.002 $ & $0.024 $ & $0.024 $\\
		&		& &  & $ 0.006$  & $0.020 $ & $0.020  $ & $ 0.008 $ & $ 0.022$ & $0.022 $ & $ 0.006 $ & $ 0.024$ & $0.024 $\\
			\cline{2-13}
		& & & & & & & & & && \\[-8pt]
	& 	& 3& $\beta_1 $ & $0.013 $  & $ 0.006$ & $ 0.006 $ & $ 0.009 $ & $0.007 $ & $0.007 $ & $ 0.003 $ & $0.008 $ & $0.008 $\\
	&	& &  & $0.012 $  & $0.006 $ & $ 0.006 $ & $ 0.009 $ & $0.007 $ & $0.007 $ & $ 0.007 $ & $ 0.008$ & $0.008 $\\
	&	& & $\beta_2 $ & $ -0.002$  & $0.002 $ & $0.002  $ & $-0.001  $ & $ 0.003$ & $0.003 $ & $ 0.000 $ & $0.003 $ & $0.003 $\\
	&	& &  & $-0.002 $  & $0.002 $ & $ 0.002 $ & $ -0.001 $ & $0.003 $ & $0.003 $ & $-0.001  $ & $0.003 $ & $0.003 $\\
	&	& & $\beta_3 $ & $ 0.008$  & $0.030 $ & $ 0.030 $ & $0.000  $ & $0.035 $ & $0.035 $ & $ -0.007 $ & $0.045 $ & $ 0.045$\\
	&	& &  & $0.007 $  & $0.030 $ & $ 0.030 $ & $ 0.001 $ & $ 0.035$ & $0.035 $ & $ 0.001 $ & $ 0.047$ & $0.047 $\\
				\hline
	&	& & & & & & & & & && \\[-8pt]
	4	&$200$	&1 & $\beta_1 $ & $-0.012  $  & $ 0.005$ & $  0.005 $ & $ -0.010  $ & $ 0.005 $ & $  0.005$ & $  -0.006 $ & $ 0.005 $ & $ 0.005 $\\
	&		& &  & $ -0.013 $  & $0.005 $ & $ 0.005  $ & $-0.013   $ & $0.005  $ & $ 0.005 $ & $ -0.013  $ & $ 0.005 $ & $ 0.006 $\\
	&		& & $\beta_2 $ &$0.003  $  & $ 0.003$ & $ 0.003  $ & $ 0.002  $ & $ 0.003 $ & $ 0.003 $ & $0.001   $ & $ 0.003 $ & $ 0.003 $\\
	&		& &  &$ 0.003 $  & $ 0.003$ & $ 0.003  $ & $ 0.002  $ & $ 0.003 $ & $ 0.003 $ & $0.003   $ & $ 0.003 $ & $  0.003$\\
	&		& & $\beta_3 $ &$ 0.002 $  & $0.032 $ & $ 0.032  $ & $ 0.000  $ & $ 0.036 $ & $0.036  $ & $ 0.000  $ & $0.040  $ & $ 0.040 $\\
	&		& &    & $0.002 $ & $ 0.032  $ & $ 0.032  $ & $  0.002$ & $ 0.037 $ & $ 0.037  $& $ 0.005 $ & $ 0.041 $ & $  0.041$\\
	\cline{3-13}
	&		& & & & & & & & & && \\[-8pt]
	&		& 3& $\beta_1 $ & $ 0.013 $  & $0.005 $ & $  0.005 $ & $  0.010 $ & $ 0.005 $ & $ 0.005 $ & $ 0.002  $ & $ 0.007 $ & $  0.007$\\
	&		& &  & $ 0.014 $  & $ 0.005$ & $  0.005 $ & $ 0.013  $ & $ 0.005 $ & $  0.006$ & $0.010   $ & $ 0.007 $ & $0.007  $\\
	&		& & $\beta_2 $ &$- 0.003 $  & $0.004 $ & $ 0.004  $ & $-0.002   $ & $ 0.005 $ & $  0.005$ & $ -0.002  $ & $ 0.006 $ & $ 0.006 $\\
	&		& &  &$ -0.004 $  & $ 0.004$ & $ 0.004  $ & $ -0.003  $ & $ 0.005 $ & $0.005  $ & $ -0.004  $ & $  0.006$ & $ 0.006 $\\
		&		& & $\beta_3 $ &$ -0.001 $  & $ 0.057$ & $ 0.057  $ & $ -0.008  $ & $ 0 .067$ & $ 0.067 $ & $ -0.023  $ & $  0.085$ & $  0.085$\\
	&		& &  & $ 0.001 $  & $ 0.057$ & $  0.057 $ & $ -0.001  $ & $ 0.068 $ & $0.068  $ & $ -0.005  $ & $ 0.089 $ & $ 0.089 $\\
	\cline{2-13}
	&		& & & & & & & & & && \\[-8pt]
		& $400$		& 1& $\beta_1 $ &$  -0.005$  & $0.002 $ & $ 0.002  $ & $ -0.003  $ & $ 0.002 $ & $ 0.002 $ & $ 0.001  $ & $ 0.003 $ & $ 0.003 $\\
		&		& &  & $-0.005  $  & $ 0.002$ & $  0.002 $ & $-0.005   $ & $ 0.002 $ & $ 0.002 $ & $ -0.003  $ & $ 0.003 $ & $ 0.003 $\\
		&		& & $\beta_2 $ &$0.000  $  & $0.001 $ & $ 0.001  $ & $  0.000 $ & $ 0.001 $ & $ 0.001 $ & $ -0.001  $ & $ 0.001 $ & $ 0.001 $\\
		&		& &  &$ 0.000 $  & $0.001 $ & $ 0.001  $ & $ 0.000  $ & $  0.001$ & $ 0.001 $ & $ -0.001  $ & $ 0.001 $ & $ 0.001 $\\
		&		& & $\beta_3 $ &$0.002  $  & $ 0.015$ & $ 0.015  $ & $ 0.000  $ & $ 0.017 $ & $0.017  $ & $ -0.002  $ & $ 0.018 $ & $ 0.018 $\\
	&		& &  & $0.002  $  & $ 0.015$ & $  0.015 $ & $ 0.000  $ & $ 0.017 $ & $ 0.017 $ & $ -0.001  $ & $ 0.018 $ & $ 0.018 $\\
	\cline{3-13}
	& & & & & & & & & && \\[-8pt]
	& 	& 3& $\beta_1 $ & $ 0.006 $  & $0.002 $ & $  0.002 $ & $0.004   $ & $ 0.003 $ & $  0.003$ & $0.000   $ & $ 0.003 $ & $  0.003$\\
	&		& &  & $0.006  $  & $0.002 $ & $ 0.002  $ & $ 0.006  $ & $  0.003$ & $ 0.003 $ & $ 0.006  $ & $ 0.004 $ & $ 0.004 $\\
	&		& & $\beta_2 $ &$ -0.004 $  & $0.002 $ & $ 0.002  $ & $ -0.004  $ & $ 0.003 $ & $  0.003$ & $ -0.004  $ & $ 0.003 $ & $ 0.003 $\\
	&		& &  &$ -0.005 $  & $0.002 $ & $  0.002 $ & $ -0.005  $ & $ 0.003 $ & $0.003  $ & $ -0.005  $ & $ 0.003 $ & $ 0.003 $\\
	&		& & $\beta_3 $ &$ 0.001 $  & $0.028 $ & $ 0.028  $ & $  -0.004 $ & $ 0.033 $ & $ 0.033 $ & $ -0.013  $ & $ 0.046 $ & $ 0.046 $\\
&		& &  & $ 0.001 $  & $ 0.028$ & $ 0.028  $ & $ -0.001  $ & $ 0.034 $ & $ 0.034 $ & $ 0.000  $ & $ 0.048 $ & $ 0.048 $\\
			\end{tabular}
	}}
\end{table}
		
	\begin{table}
		\caption{	\label{tab:results3:2}Bias, variance and MSE of $\hat\gamma$ and $\hat\beta$ for \texttt{smcure} (second rows) and our approach (first rows) in Model 3 (Scenario 2).}
		\centering
		\scalebox{0.85}{
			\fbox{
				\begin{tabular}{ccccrrrrrrrrr}
				&	& & & & & & & & & && \\[-8pt]
				&	&&&\multicolumn{3}{c}{Cens. level 1}&\multicolumn{3}{c}{Cens. level 2}&\multicolumn{3}{c}{Cens. level 3}\\
				Mod.&	n&  scen. & Par. &  Bias & Var. & MSE & Bias & Var. & MSE & Bias & Var. & MSE\\[2pt]
					\hline
				&	& & & & & & & & & && \\[-8pt]
			3&$200$		& $2 $ & $\gamma_1 $ & $-0.056 $  & $ 0.261$ & $0.264  $ & $-0.215  $ & $0.344 $ & $0.390 $ & $-0.379  $ & $0.422 $ & $0.566 $\\
			&		& &  & $0.091 $  & $ 0.299$ & $0.308  $ & $ 0.090 $ & $0.454 $ & $ 0.462$ & $ 0.127 $ & $0.755 $ & $ 0.771$\\
			&		& & $\gamma_2 $ & $-0.055 $  & $0.124 $ & $ 0.127 $ & $ -0.123 $ & $0.137 $ & $0.152 $ & $-0.198  $ & $0.165 $ & $0.204 $\\
			&		& & & $0.183 $  & $0.178 $ & $ 0.212 $ & $ 0.275 $ & $0.268 $ & $ 0.344$ & $ 0.390 $ & $ 0.485$ & $0.638 $\\
			&		& & $\gamma_3 $ & $0.023 $  & $0.401 $ & $0.402  $ & $ 0.087 $ & $0.512 $ & $0.519 $ & $ 0.162 $ & $ 0.612$ & $ 0.639$\\
		&			& & & $0.161 $  & $0.488 $ & $ 0.514 $ & $ 0.278 $ & $0.729 $ & $0.806 $ & $ 0.427 $ & $ 1.159$ & $1.341 $\\
		&			& & $\gamma_4 $ & $-0.049 $  & $ 0.321$ & $ 0.324 $ & $ -0.093 $ & $ 0.421$ & $ 0.429$ & $ -0.164 $ & $0.543 $ & $0.570 $\\
		&			& & & $0.036 $  & $0.367 $ & $0.369  $ & $ 0.086 $ & $ 0.547$ & $ 0.554$ & $ 0.062 $ & $ 0.795$ & $ 0.799$\\
		&			& & $\beta_1 $ & $ 0.017$  & $0.008 $ & $0.009  $ & $ 0.017 $ & $ 0.009$ & $0.009 $ & $ 0.013 $ & $0.010 $ & $ 0.010$\\
		&			& &  & $0.016 $  & $ 0.008$ & $ 0.009 $ & $0.018  $ & $0.009 $ & $0.009 $ & $0.015  $ & $0.010 $ & $0.011 $\\
		&			& & $\beta_2 $ & $0.004 $  & $0.004 $ & $ 0.004 $ & $0.007  $ & $ 0.004$ & $0.004 $ & $0.005  $ & $0.005 $ & $0.005 $\\
		&			& &  & $ 0.004$  & $0.004 $ & $ 0.004 $ & $0.009  $ & $0.004 $ & $0.004 $ & $  0.008$ & $ 0.005$ & $0.005 $\\
		&			& & $\beta_3 $ & $0.022 $  & $0.059 $ & $0.059  $ & $0.008  $ & $0.059 $ & $0.060 $ & $-0.003  $ & $ 0.074$ & $0.074 $\\
		&			& &  & $ 0.024$  & $0.059 $ & $ 0.060$ & $ 0.015 $ & $0.060 $ & $0.060 $ & $ 0.013 $ & $0.074 $ & $0.075 $\\
					\cline{2-13}
		&			& & & & & & & & & && \\[-8pt]
		&$400$			& $2 $ & $\gamma_1 $ & $-0.037 $  & $0.120 $ & $0.121  $ & $-0.145  $ & $ 0.174$ & $0.195 $ & $ -0.255 $ & $0.240 $ & $0.305 $\\
		&			& &  & $0.047 $  & $0.122 $ & $0.124  $ & $ 0.029 $ & $ 0.179$ & $0.179 $ & $ 0.042 $ & $0.265 $ & $ 0.267$\\
		&			& & $\gamma_2 $ & $-0.059 $  & $0.056 $ & $ 0.060 $ & $-0.123  $ & $0.075 $ & $0.090 $ & $ -0.177 $ & $0.077 $ & $0.0108 $\\
		&			& & & $0.090 $  & $ 0.068$ & $0.076  $ & $ 0.110 $ & $ 0.098$ & $ 0.110$ & $0.173  $ & $0.125 $ & $0.155 $\\
		&			& & $\gamma_3 $ & $-0.020 $  & $0.181 $ & $ 0.182 $ & $ 0.026 $ & $ 0.245$ & $0.246 $ & $0.080  $ & $0.337 $ & $0.344 $\\
		&			& & & $0.066 $  & $ 0.194$ & $ 0.198 $ & $ 0.124 $ & $0.261 $ & $ 0.276$ & $0.212  $ & $0.380 $ & $ 0.425$\\
		&			& & $\gamma_4 $ & $-0.021 $  & $ 0.151$ & $ 0.152 $ & $ -0.069 $ & $ 0.220$ & $ 0.225$ & $ -0.122 $ & $ 0.315$ & $0.329 $\\
		&			& & & $0.023 $  & $0.156 $ & $ 0.157 $ & $ 0.022 $ & $0.218 $ & $ 0.218$ & $ 0.033 $ & $ 0.309$ & $ 0.310$\\
		&			& & $\beta_1 $ & $ 0.011$  & $ 0.004$ & $0.004  $ & $ 0.008 $ & $0.004 $ & $0.004 $ & $0.013  $ & $0.005 $ & $ 0.005$\\
		&			& &  & $0.011 $  & $0.004 $ & $0.004  $ & $ 0.007 $ & $ 0.004$ & $0.004 $ & $ 0.013 $ & $0.005 $ & $ 0.005$\\
		&			& & $\beta_2 $ & $0.003 $  & $ 0.002$ & $0.002  $ & $ 0.003 $ & $ 0.002$ & $ 0.002$ & $0.001  $ & $ 0.002$ & $ 0.002$\\
		&			& &  & $ 0.003$  & $ 0.002$ & $  0.002$ & $ 0.004 $ & $0.002 $ & $ 0.002$ & $ 0.003 $ & $0.002 $ & $0.002 $\\
		&			& & $\beta_3 $ & $0.019 $  & $0.028 $ & $0.029 $ & $0.002  $ & $ 0.029$ & $0.029 $ & $-0.001  $ & $ 0.033$ & $0.033 $\\
				&	& &  & $0.020 $  & $ 0.028$ & $ 0.029 $ & $0.006  $ & $ 0.029$ & $ 0.029$ & $0.009  $ & $0.034 $ & $0.034 $\\
								\end{tabular}
			}	}
		\end{table}
	
	\begin{table}
		\caption{	\label{tab:results4:2}Bias, variance and MSE of $\hat\gamma$ and $\hat\beta$ for \texttt{smcure} (second rows) and our approach (first rows) in Model 4 (Scenario 2).}
		\centering
		\scalebox{0.85}{
			\fbox{
				\begin{tabular}{ccccrrrrrrrrr}
					&	& & & & & & & & & && \\[-8pt]
					&	&&&\multicolumn{3}{c}{Cens. level 1}&\multicolumn{3}{c}{Cens. level 2}&\multicolumn{3}{c}{Cens. level 3}\\
					Mod.&	n&  scen. & Par. &  Bias & Var. & MSE & Bias & Var. & MSE & Bias & Var. & MSE\\[2pt]
								\hline
				&	& & & & & & & & & && \\[-8pt]
				4&$200$		& $2 $ & $\gamma_1 $ & $ -0.121 $  & $0.117  $ & $ 0.131  $ & $ -0.240  $ & $ 0.140 $ & $ 0.198$ & $ -0.375 $ & $0.179  $ & $ 0.320 $\\
				&		& &  & $ 0.006 $  & $0.134  $ & $ 0.134  $ & $ -0.011  $ & $ 0.170 $ & $ 0.170$ & $ -0.035 $ & $0.231  $ & $ 0.232 $\\
				&		& & $\gamma_2 $ & $ 0.053 $  & $ 0.014 $ & $ 0.017  $ & $ 0.091  $ & $0.017  $ & $ 0.026$ & $ 0.122 $ & $ 0.018 $ & $  0.033$\\
				&		& & &$ 0.041 $  & $ 0.016 $ & $ 0.018  $ & $ 0.056  $ & $0.020  $ & $0.024 $ & $ 0.077 $ & $0.027 $ & $ 0.033 $\\
				&		& & $\gamma_3 $ & $ -0.195 $  & $ 0.137 $ & $ 0.175  $ & $ -0.348  $ & $ 0.156 $ & $ 0.277$ & $  -0.511$ & $0.165  $ & $  0.426$\\
				&			& & & $ 0.119 $  & $0.169  $ & $  0.183 $ & $  0.122 $ & $ 0.216 $ & $ 0.231$ & $ 0.139 $ & $ 0.273 $ & $ 0.293 $\\
				&			& & $\gamma_4 $ & $ 0.096 $  & $ 0.156 $ & $ 0.165  $ & $ 0.157  $ & $ 0.189 $ & $ 0.213$ & $ 0.240 $ & $0.229  $ & $0.287  $\\
				&			& & & $ 0.072 $  & $ 0.182 $ & $ 0.187  $ & $ 0.096  $ & $  0.216$ & $0.226 $ & $ 0.137 $ & $0.288  $ & $ 0.307 $\\
					&			& & $\gamma_5 $ & $ -0.038 $  & $ 0.161 $ & $  0.163 $ & $ -0.078  $ & $ 0.185 $ & $0.191 $ & $  -0.149$ & $ 0.218 $ & $  0.240$\\
				&			& & & $0.046  $  & $ 0.175 $ & $  0.177 $ & $  0.058 $ & $ 0.220 $ & $ 0.223$ & $ 0.047 $ & $0.254  $ & $0.257  $\\
					&			& & $\beta_1 $ & $ 0.014 $  & $  0.007$ & $  0.007 $ & $ 0.009  $ & $ 0.008 $ & $0.008 $ & $  0.001$ & $ 0.008 $ & $ 0.008 $\\
				&			& &  &$ 0.017 $  & $ 0.007 $ & $ 0.007  $ & $ 0.016  $ & $ 0.008 $ & $0.008 $ & $ 0.014 $ & $  0.008$ & $ 0.008 $\\
					&			& & $\beta_2 $ & $ 0.005 $  & $  0.004$ & $ 0.004  $ & $0.003   $ & $ 0.004 $ & $0.004 $ & $  0.002$ & $0.004  $ & $  0.005$\\
				&			& &  &$ 0.005 $  & $ 0.004 $ & $0.004   $ & $  0.005 $ & $ 0.004 $ & $ 0.004$ & $0.005  $ & $  0.004$ & $0.005  $\\
					&			& & $\beta_3 $ & $ 0.010 $  & $0.056  $ & $ 0.056  $ & $ -0.003  $ & $ 0.063 $ & $0.063 $ & $-0.022  $ & $ 0.071 $ & $ 0.072 $\\
				&			& &  &$  0.015$  & $ 0.057 $ & $ 0.057  $ & $ 0.011  $ & $ 0.064 $ & $ 0.064$ & $ 0.006 $ & $  0.072$ & $  0.072$\\
				\cline{2-13}
				&			& & & & & & & & & && \\[-8pt]
				&$400$			& $2 $ & $\gamma_1 $ & $ -0.068 $  & $  0.057$ & $ 0.062  $ & $  -0.169 $ & $0.073  $ & $ 0.102$ & $ -0.279 $ & $ 0.089 $ & $ 0.167 $\\
				&			& &  & $ 0.015 $  & $ 0.061 $ & $ 0.061  $ & $ 0.001  $ & $ 0.078 $ & $ 0.078$ & $ -0.022 $ & $ 0.097 $ & $ 0.098 $\\
				&			& & $\gamma_2 $ & $ 0.033 $  & $ 0.007 $ & $  0.008 $ & $ 0.068  $ & $ 0.008 $ & $ 0.013$ & $ 0.099 $ & $  0.010$ & $  0.019$\\
				&			& & & $ 0.019 $  & $ 0.008 $ & $  0.008 $ & $ 0.030  $ & $ 0.010 $ & $ 0.011$ & $ 0.044 $ & $ 0.012 $ & $ 0.014 $\\
				&			& & $\gamma_3 $ & $ -0.180 $  & $ 0.073 $ & $ 0.105  $ & $  -0.311 $ & $ 0.084 $ & $ 0.181$ & $ -0.450 $ & $ 0.093 $ & $ 0.296 $\\
				&			& & &$ 0.043 $  & $ 0.079 $ & $ 0.081  $ & $ 0.038 $  & $0.096  $ & $  0.097 $  & $ 0.030 $ & $  0.119$ & $ 0.120 $\\
				&			& & $\gamma_4 $ & $ 0.050 $  & $0.078  $ & $ 0.080  $ & $ 0.119  $ & $ 0.096 $ & $0.110 $ & $ 0.177 $ & $ 0.118 $ & $ 0.149 $\\
				&			& & &  $ 0.028  $ & $ 0.082 $ & $0.083 $   & $ 0.049  $ & $ 0.101 $ & $ 0.104$ & $  0.066$ & $ 0.124 $ & $ 0.129 $\\
					&			& & $\gamma_5 $ & $ -0.060 $  & $ 0.077 $ & $0.081   $ & $-0.106   $ & $0.096  $ & $0.107 $ & $ -0.150 $ & $  0.101$ & $ 0.124 $\\
				&			& & &$ 0.003 $  & $0.082  $ & $ 0.082  $ & $ 0.001  $ & $0.099  $ & $0.099 $ & $ 0.003 $ & $ 0.116 $ & $ 0.116 $\\
				&			& & $\beta_1 $ & $0.006  $  & $ 0.003 $ & $ 0.003  $ & $ 0.002  $ & $ 0.004 $ & $ 0.004$ & $ -0.004 $ & $ 0.004 $ & $ 0.004 $\\
				&			& &  &$ 0.008 $  & $ 0.003 $ & $ 0.003  $ & $  0.008 $ & $ 0.004 $ & $0.004 $ & $0.007  $ & $ 0.004 $ & $ 0.004 $\\
				&			& & $\beta_2 $ & $ 0.002 $  & $ 0.002 $ & $ 0.002  $ & $0.000   $ & $  0.002$ & $ 0.002$ & $ 0.000 $ & $ 0.002 $ & $0.002  $\\
			&			& &  &$ 0.002 $  & $ 0.002 $ & $ 0.002  $ & $0.001   $ & $ 0.002 $ & $0.002 $ & $0.002  $ & $ 0.002 $ & $ 0.002 $\\
				&			& & $\beta_3 $ & $ 0.003 $  & $0.026  $ & $ 0.026  $ & $ -0.010  $ & $0.030  $ & $ 0.030$ & $ -0.023 $ & $  0.034$ & $0.034  $\\
			&			& &  &$ 0.007 $  & $0.026  $ & $0.026   $ & $ 0.002  $ & $0.031  $ & $0.031 $ & $ -0.001 $ & $0.034  $ & $ 0.034 $\\
					\end{tabular}
		}	}
	\end{table}

\begin{table}
	\caption{\label{tab:results1_2_1000}Bias, variance and MSE of $\hat\gamma$ and $\hat\beta$ for \texttt{smcure} (second rows) and our approach (first rows) in Model 1 and 2 ($n=1000$).}
	\centering
	\scalebox{0.85}{
	\fbox{
		\begin{tabular}{ccccrrrrrrrrr}
		&	&&&\multicolumn{3}{c}{Cens. level 1}&\multicolumn{3}{c}{Cens. level 2}&\multicolumn{3}{c}{Cens. level 3}\\
		Mod.&	n&  scen. & Par. &  Bias & Var. & MSE & Bias & Var. & MSE & Bias & Var. & MSE\\[2pt]
			\hline
		&	& & & & & & & & & && \\[-8pt]
		1&	$1000$ & $1 $ & $\gamma_1 $ & $ 0.007$  & $ 0.011$ & $ 0.011 $ & $  0.008$ & $ 0.013$ & $0.013 $ & $ 0.010 $ & $ 0.015$ & $0.015 $\\
		&	& &  & $0.016 $  & $0.011 $ & $  0.012$ & $  0.019$ & $0.013 $ & $ 0.013$ & $0.022  $ & $0.015 $ & $0.016 $\\
		&	& & $\gamma_2 $ & $-0.006 $  & $ 0.032$ & $ 0.032 $ & $-0.006  $ & $0.037 $ & $0.037 $ & $-0.006  $ & $0.041 $ & $0.041 $\\
		&	& & & $0.023 $  & $ 0.033$ & $ 0.034 $ & $0.027  $ & $ 0.038$ & $0.039 $ & $ 0.028 $ & $ 0.043$ & $0.044 $\\
		&	& & $\beta $ & $0.001 $  & $ 0.005$ & $0.005  $ & $ 0.002 $ & $0.006 $ & $ 0.006$ & $ 0.002 $ & $ 0.006$ & $ 0.006$\\
		&	& &  & $ 0.000$  & $ 0.005$ & $ 0.005 $ & $ 0.001 $ & $0.006 $ & $0.006 $ & $ 0.000 $ & $0.006 $ & $0.006 $\\
			\cline{3-13}
		&	& & & & & & & & & && \\[-8pt]
		&	& $2 $ & $\gamma_1 $ & $ 0.005$  & $ 0.007$ & $0.007  $ & $ 0.006 $ & $0.008 $ & $0.008 $ & $ 0.007 $ & $ 0.010$ & $0.010 $\\
		&	& &  & $ 0.008$  & $0.007 $ & $  0.007$ & $0.011  $ & $0.008 $ & $ 0.008$ & $ 0.014 $ & $0.010 $ & $0.010 $\\
		&	& & $\gamma_2 $ & $-0.008 $  & $ 0.019$ & $0.019  $ & $ -0.006 $ & $0.022 $ & $ 0.022$ & $ -0.007 $ & $0.026 $ & $0.026 $\\
		&	& & & $0.012 $  & $ 0.019$ & $ 0.020 $ & $ 0.016 $ & $ 0.023$ & $ 0.023$ & $ 0.018 $ & $0.027 $ & $0.027 $\\
		&	& & $\beta $ & $ 0.000$  & $ 0.006$ & $0.006  $ & $ -0.001 $ & $0.007 $ & $ 0.007$ & $ 0.000 $ & $ 0.008$ & $0.008 $\\
		&	& &  & $-0.001 $  & $0.006 $ & $ 0.006 $ & $-0.002  $ & $0.007 $ & $ 0.007$ & $ -0.002 $ & $ 0.008$ & $ 0.008$\\
			\cline{3-13}
		&	& & & & & & & & & && \\[-8pt]
		&	& $3 $ & $\gamma_1 $ & $0.001 $  & $0.011 $ & $0.011  $ & $ -0.008 $ & $ 0.012$ & $ 0.012$& $ -0.003 $ & $ 0.018$ & $0.018 $ \\
		&	& &  & $ 0.004$  & $0.011 $ & $  0.011$ & $-0.002  $ & $0.013 $ & $0.013 $ & $0.009  $ & $ 0.019$ & $ 0.019$\\
		&	& & $\gamma_2 $ & $ -0.064$  & $ 0.093$ & $0.097  $ & $ -0.076 $ & $ 0.123$ & $0.129 $ & $ -0.107 $ & $0.163 $ & $ 0.175$\\
			&& & & $ 0.046$  & $0.099 $ & $ 0.102 $ & $0.058  $ & $0.130 $ & $0.134 $ & $ 0.072 $ & $0.167 $ & $ 0.173$\\
		&	& & $\beta $ & $ 0.011$  & $ 0.018$ & $0.018  $ & $0.005  $ & $0.022 $ & $ 0.022$ & $ 0.004 $ & $0.026 $ & $0.026 $\\
		&	& &  & $ 0.009$  & $0.018 $ & $0.018  $ & $ 0.001 $ & $0.022 $ & $0.022 $ & $ -0.005 $ & $0.026 $ & $ 0.026$\\
			\hline
		&	& & & & & & & & & && \\[-8pt]
		2&	$1000$ & $1 $ & $\gamma_1 $ & $0.006 $  & $ 0.008$ & $ 0.008 $ & $ 0.000 $ & $0.009 $ & $ 0.009$ & $0.007  $ & $ 0.012$ & $0.012 $\\
		&	& &  & $0.009 $  & $0.008 $ & $ 0.008 $ & $ 0.007 $ & $0.009 $ & $0.009 $ & $0.022  $ & $0.013 $ & $0.013 $\\
		&	& & $\gamma_2 $ & $-0.010 $  & $ 0.007$ & $ 0.007 $ & $ -0.011 $ & $0.009 $ & $0.009 $ & $ -0.018 $ & $0.011 $ & $0.011 $\\
		&	& & & $-0.002 $  & $0. 007$ & $0.007  $ & $-0.002  $ & $0.009 $ & $ 0.009$ & $ -0.005 $ & $0.012 $ & $0.012 $\\
		&	& & $\beta $ & $-0.001 $  & $0.002 $ & $ 0.002 $ & $0.000  $ & $ 0.003$ & $ 0.003$ & $ 0.001 $ & $0.003 $ & $0.003 $\\
		&	& &  & $-0.001 $  & $0.002 $ & $ 0.002 $ & $-0.001  $ & $0.003 $ & $0.003 $ & $ -0.001 $ & $0.003 $ & $ 0.003$\\
			\cline{3-13}
		&	& & & & & & & & & && \\[-8pt]
		&	& $2 $ & $\gamma_1 $ & $-0.004 $  & $ 0.006$ & $ 0.006 $ & $-0.010  $ & $0.007 $ & $ 0.007$ & $-0.004  $ & $0.010 $ & $0.010 $\\
			&& &  & $0.002 $  & $0.006 $ & $ 0.006 $ & $0.003  $ & $0.008 $ & $ 0.008$ & $ 0.022 $ & $0.010 $ & $0.011 $\\
	&		& & $\gamma_2 $ & $ -0.012$  & $0.007 $ & $ 0.008 $ & $ -0.023 $ & $ 0.009$ & $ 0.010$ & $ -0.031 $ & $ 0.012$ & $0.012 $\\
	&		& & & $0.004 $  & $0.008 $ & $0.008  $ & $0.000  $ & $0.009 $ & $0.009 $ & $ 0.002 $ & $0.012 $ & $ 0.012$\\
	&		& & $\beta $ & $0.005 $  & $0.003 $ & $0.003  $ & $0.006  $ & $0.003 $ & $0.003 $ & $  0.005$ & $0.004 $ & $ 0.004$\\
	&		& &  & $0.004 $  & $0.003 $ & $0.003  $ & $0.005 $ & $0.003 $ & $0.003 $ & $  0.001$ & $0.004 $ & $ 0.004$\\
			\cline{3-13}
	&		& & & & & & & & & && \\[-8pt]
		&	& $3 $ & $\gamma_1 $ & $-0.011 $  & $0.014 $ & $0.014  $ & $ -0.029 $ & $ 0.014$ & $0.015 $ & $ -0.054 $ & $0.019 $ & $ 0.022$\\
		&	& &  & $ 0.002$  & $0.016 $ & $0.016  $ & $ 0.015 $ & $ 0.018$ & $0.018 $ & $ 0.028 $ & $0.029 $ & $0.029 $\\
		&	& & $\gamma_2 $ & $-0.314 $  & $0.144 $ & $ 0.242 $ & $ -0.583 $ & $0.189 $ & $0.529 $ & $ -0.841 $ & $ 0.237$ & $0.945 $\\
		&	& & & $0.045 $  & $ 0.148$ & $ 0.151 $ & $0.068  $ & $0.175 $ & $ 0.180$ & $0.137  $ & $0.223 $ & $0.242 $\\
		&	& & $\beta $ & $ 0.004$  & $0.006 $ & $ 0.006 $ & $ 0.004 $ & $0.008 $ & $ 0.008$ & $ 0.006 $ & $0.010 $ & $0.010 $\\
		&	& &  & $ 0.002$  & $0.006 $ & $ 0.006 $ & $ 0.001 $ & $0.008 $ & $ 0.008$ & $ 0.003 $ & $0.010 $ & $0.010 $\\
		\end{tabular}
	}}	
\end{table}

\begin{table}
	\caption{	\label{tab:results3_1000}Bias, variance and MSE of $\hat\gamma$ and $\hat\beta$ for \texttt{smcure} and our approach in Model 3 ($n=1000$).}
	\centering
\scalebox{0.85}{
	\fbox{
		\begin{tabular}{cccrrrrrrrrr}
			& & & & & & & & & && \\[-8pt]
			&&&\multicolumn{3}{c}{Cens. level 1}&\multicolumn{3}{c}{Cens. level 2}&\multicolumn{3}{c}{Cens. level 3}\\
			n&  scen. & Par. &  Bias & Var. & MSE & Bias & Var. & MSE & Bias & Var. & MSE\\[2pt]
			\hline
			& & & & & & & & & && \\[-8pt]
			$1000$ & $1 $ & $\gamma_1 $ & $ 0.000$  & $ 0.026$ & $ 0.026 $ & $  0.010$ & $0.035 $ & $0.035 $ & $  0.013$ & $0.050 $ & $ 0.050$\\
			& &  & $0.001 $  & $0.026 $ & $0.026  $ & $0.003  $ & $0.031 $ & $0.031 $ & $ 0.006 $ & $0.040 $ & $ 0.040$\\
			& & $\gamma_2 $ & $-0.009 $  & $0.008 $ & $0.008  $ & $ -0.021 $ & $ 0.001$ & $0.011 $ & $-0.051  $ & $ 0.014$ & $0.017 $\\
			& & & $-0.018 $  & $0.008 $ & $0.008  $ & $-0.022  $ & $0.010 $ & $0.010 $ & $ -0.036 $ & $0.013 $ & $0.014 $\\
			& & $\gamma_3 $ & $ 0.019$  & $ 0.071$ & $ 0.072 $ & $ 0.022 $ & $0.094 $ & $ 0.094$ & $ 0.022 $ & $ 0.128$ & $0.128 $\\
			& & & $0.041 $  & $0.069 $ & $ 0.070 $ & $ 0.060 $ & $ 0.078$ & $ 0.082$ & $  0.062$ & $0.105 $ & $ 0.109$\\
			& & $\gamma_4 $ & $0.004 $  & $ 0.061$ & $ 0.061 $ & $ -0.030 $ & $0.087 $ & $0.088 $ & $ -0.095 $ & $0.124 $ & $ 0.133$\\
			& & & $0.026 $  & $ 0.058$ & $ 0.058 $ & $ -0.002 $ & $0.077 $ & $ 0.078$ & $ 0.033 $ & $ 0.098$ & $0.099 $\\
			& & $\beta_1 $ & $ -0.002$  & $0.001 $ & $ 0.001 $ & $-0.001  $ & $ 0.001$ & $ 0.001$ & $0.001  $ & $0.001 $ & $0.001 $\\
			& &  & $ -0.002$  & $ 0.001$ & $  0.001$ & $-0.002  $ & $0.001 $ & $0.001 $ & $-0.002  $ & $0.001 $ & $0.001 $\\
			& & $\beta_2 $ & $ 0.002$  & $0.001 $ & $ 0.001 $ & $0.000  $ & $ 0.001$ & $ 0.001$ & $0.001  $ & $0.001 $ & $0.001 $\\
			& &  & $ 0.002$  & $ 0.001$ & $  0.001$ & $0.000 $ & $0.001 $ & $0.001 $ & $0.001 $ & $0.001 $ & $0.001 $\\
			& & $\beta_3 $ & $ 0.001$  & $0.007 $ & $ 0.007 $ & $-0.006  $ & $ 0.008$ & $ 0.008$ & $-0.003  $ & $0.009 $ & $0.009 $\\
			& &  & $ 0.001$  & $ 0.007$ & $  0.007$ & $-0.006  $ & $0.0018$ & $0.008 $ & $-0.001  $ & $0.009 $ & $0.009 $\\
			\cline{2-12}
			& & & & & & & & & && \\[-8pt]
			& $2 $ & $\gamma_1 $ & $ -0.025$  & $0.046 $ & $ 0.047 $ & $ -0.071 $ & $ 0.077$ & $ 0.082$ & $ -0.142 $ & $0.117 $ & $0.137 $\\
			& &  & $ 0.020$  & $ 0.046$ & $0.047  $ & $  0.007$ & $ 0.066$ & $ 0.066$ & $0.005  $ & $0.090 $ & $0.090 $\\
			& & $\gamma_2 $ & $-0.057 $  & $ 0.022$ & $ 0.025 $ & $-0.097  $ & $0.030 $ & $0.040 $ & $ -0.108 $ & $ 0.038$ & $0.050 $\\
			& & & $0.032 $  & $0.023 $ & $0.024  $ & $0.037  $ & $0.031 $ & $ 0.032$ & $ 0.056 $ & $ 0.038$ & $0.041 $\\
			& & $\gamma_3 $ & $-0.041 $  & $0.068 $ & $0.070  $ & $-0.024  $ & $ 0.125$ & $ 0.126$ & $  0.015$ & $0.169 $ & $0.170 $\\
			& & & $0.009 $  & $0.066 $ & $0.066  $ & $0.044  $ & $0.104 $ & $0.106 $ & $ 0.059 $ & $0.132 $ & $0.136 $\\
			& & $\gamma_4 $ & $-0.019 $  & $ 0.059$ & $ 0.060 $ & $ -0.025 $ & $0.093 $ & $ 0.093$ & $ -0.037 $ & $0.127 $ & $0.128 $\\
			& & & $0.005 $  & $0.058 $ & $ 0.058 $ & $0.020  $ & $ 0.081$ & $0.081 $ & $0.039  $ & $0.101 $ & $0.102 $\\
			& & $\beta_1 $ & $0.003 $  & $0.001 $ & $0.001  $ & $  0.002$ & $ 0.002$ & $ 0.002$ & $ 0.002 $ & $0.002 $ & $0.002 $\\
			& &  & $0.003 $  & $0.001 $ & $0.001  $ & $  0.001$ & $ 0.002$ & $ 0.002$ & $ 0.002 $ & $0.002 $ & $0.002 $\\
			& & $\beta_2 $ & $0.001 $  & $0.001 $ & $0.001  $ & $  0.002$ & $0.001 $ & $ 0.001$ & $-0.002  $ & $0.001 $ & $0.001 $\\
			& &  & $0.001 $  & $0.001 $ & $0.001  $ & $  0.002$ & $0.001 $ & $ 0.001$ & $-0.002  $ & $0.001 $ & $0.001 $\\
			& & $\beta_3 $ & $ 0.003$  & $ 0.010$ & $ 0.010 $ & $ 0.004 $ & $0.011 $ & $0.011 $ & $ -0.010 $ & $0.013 $ & $0.013 $\\
			& &  & $ 0.003$  & $ 0.010$ & $ 0.010 $ & $ 0.005 $ & $0.011 $ & $0.011 $ & $ -0.004 $ & $0.013 $ & $0.013 $\\
			\cline{2-12}
			& & & & & & & & & && \\[-8pt]
			& $3 $ & $\gamma_1 $ & $0.004 $  & $ 0.030$ & $0.031  $ & $-0.008  $ & $0.050 $ & $0.050 $ & $ -0.058 $ & $0.093 $ & $ 0.096$\\
			& &  & $ -0.001$  & $0.029 $ & $0.029  $ & $-0.008  $ & $0.043 $ & $ 0.043$ & $ -0.022 $ & $ 0.069$ & $ 0.070$\\
			& & $\gamma_2 $ & $-0.014 $  & $0.009 $ & $ 0.009 $ & $ -0.011 $ & $ 0.014$ & $0.014 $ & $ 0.031 $ & $ 0.022$ & $0.023 $\\
			& & & $0.015 $  & $0.009 $ & $0.009  $ & $0.018  $ & $0.012 $ & $0.012 $ & $0.045  $ & $0.018 $ & $0.020 $\\
			& & $\gamma_3 $ & $ 0.000$  & $0.044 $ & $ 0.044 $ & $ -0.004 $ & $0.072 $ & $ 0.072$ & $ 0.049 $ & $0.123 $ & $ 0.125$\\
			& & & $ 0.017$  & $0.043 $ & $ 0.043 $ & $0.016  $ & $0.058 $ & $ 0.058$ & $0.043  $ & $0.089 $ & $0.091 $\\
			& & $\gamma_4 $ & $ -0.002$  & $ 0.037$ & $ 0.037 $ & $0.002  $ & $0.060 $ & $ 0.060$ & $  -0.020$ & $ 0.099$ & $0.099 $\\
			& & & $-0.007 $  & $0.036 $ & $0.036  $ & $0.003  $ & $0.051 $ & $0.051 $ & $  -0.012$ & $0.071 $ & $0.072 $\\
			& & $\beta_1 $ & $ 0.007$  & $ 0.002$ & $ 0.002 $ & $ 0.008 $ & $0.003 $ & $0.003 $ & $0.005  $ & $0.003 $ & $0.003 $\\
			& &  & $ 0.006$  & $0.002 $ & $0.002  $ & $ 0.007 $ & $0.003 $ & $ 0.003$ & $  0.006$ & $0.003 $ & $0.003 $\\
			& & $\beta_2 $ & $ 0.000$  & $ 0.001$ & $ 0.001 $ & $ -0.002 $ & $ 0.001$ & $ 0.001$ & $ -0.001 $ & $0.001 $ & $0.001 $\\
			& &  & $ 0.000$  & $ 0.001$ & $ 0.001 $ & $ -0.002 $ & $ 0.001$ & $ 0.001$ & $ -0.001 $ & $0.001 $ & $0.001 $\\
			& & $\beta_3 $ & $ 0.000$  & $ 0.001$ & $ 0.001 $ & $ 0.014 $ & $ 0.013$ & $ 0.013$ & $ -0.004 $ & $0.017 $ & $0.017 $\\
			& &  &  $ 0.000$  & $ 0.001$ & $ 0.001 $ & $ 0.014 $ & $ 0.012$ & $ 0.013$ & $ -0.001 $ & $0.016 $ & $0.016 $\\
		\end{tabular}
	}}
\end{table}

\begin{table}
	\caption{	\label{tab:results4_1000}Bias, variance and MSE of $\hat\gamma$ and $\hat\beta$ for \texttt{smcure} and our approach in Model 4 ($n=1000$).}
	\centering
\scalebox{0.85}{
	\fbox{
		\begin{tabular}{cccrrrrrrrrr}
			& & & & & & & & & && \\[-8pt]
			&&&\multicolumn{3}{c}{Cens. level 1}&\multicolumn{3}{c}{Cens. level 2}&\multicolumn{3}{c}{Cens. level 3}\\
			n&  scen. & Par. &  Bias & Var. & MSE & Bias & Var. & MSE & Bias & Var. & MSE\\[2pt]
			\hline
			& & & & & & & & & && \\[-8pt]
			$1000$ & $1 $ & $\gamma_1 $ & $-0.002  $  & $ 0.027$ & $ 0.027  $ & $  0.007 $ & $  0.033$ & $ 0.033 $ & $ 0.010 $ & $0.038  $ & $ 0.038 $\\
			& & & $  0.004$  & $0.027 $ & $  0.027 $ & $ 0.008  $ & $ 0.032 $ & $ 0.032 $ & $ 0.012 $ & $0.038  $ & $ 0.038 $\\
			& & $\gamma_2 $ & $ 0.013 $  & $ 0.007$ & $ 0.007  $ & $ -0.001  $ & $ 0.008 $ & $ 0.008 $  & $ -0.005 $ & $ 0.009 $& $ 0.009 $\\
			& & & $ -0.015 $  & $ 0.007$ & $ 0.007  $ & $ -0.024  $ & $0.009  $ & $ 0.009 $ & $ -0.029 $ & $ 0.011 $ & $0.012  $\\
			& & $\gamma_3 $ & $-0.192  $  & $0.041 $ & $ 0.078  $ & $  -0.149 $ & $ 0.037 $ & $ 0.059 $ & $-0.196  $ & $  0.046$ & $0.084  $\\
			& & & $  0.020$  & $ 0.037$ & $0.037   $ & $ 0.029  $ & $ 0.047 $ & $ 0.048 $ & $ 0.035 $ & $ 0.058 $ & $ 0.059 $\\
			& & $\gamma_4 $ & $ -0.054 $  & $ 0.063$ & $  0.066 $ & $ -0.086  $ & $  0.076$ & $ 0.084 $ & $ -0.155 $ & $ 0.090 $ & $  0.114$\\
		& &  & $ 0.030 $  & $ 0.067$ & $ 0.068  $ & $  0.047 $ & $ 0.083 $ & $ 0.085 $ & $ 0.053 $ & $ 0.104 $ & $ 0.107 $\\
			& & $\gamma_5 $ & $ -0.043 $  & $ 0.060$ & $  0.062 $ & $  -0.108 $ & $ 0.075 $ & $0.087  $ & $ -0.188 $ & $ 0.083 $ & $ 0.118 $\\
			& &  & $ 0.025 $  & $ 0.061$ & $  0.062 $ & $ 0.035  $ & $ 0.075 $ & $0.076  $ & $0.039  $ & $0.092  $ & $ 0.093 $\\
			& & $\beta_1 $ & $ -0.001 $  & $ 0.001$ & $   0.001$ & $ 0.000  $ & $0.001  $ & $ 0.001 $ & $ 0.002 $ & $ 0.001 $ & $ 0.001 $\\
			& &  &$ -0.002 $  & $ 0.001$ & $ 0.001  $ & $   -0.002$ & $ 0.001 $ & $0.001  $ & $ -0.001 $ & $ 0.001 $ & $  0.001$\\
			& & $\beta_2 $ & $ 0.000 $  & $0.000 $ & $ 0.000  $ & $ 0.000  $ & $0.001  $ & $ 0.001 $ & $ 0.000 $ & $0.001  $ & $ 0.001 $\\
			& &  & $0.000  $  & $0.000 $ & $ 0.000  $ & $ 0.000  $ & $ 0.001 $ & $ 0.001 $ & $ 0.000 $ & $ 0.001 $ & $ 0.001 $\\
			& & $\beta_3 $ & $ 0.001 $  & $0.006 $ & $   0.006 $ & $ 0.002  $ & $ 0.006 $ & $ 0.006 $ & $ 0.002 $ & $ 0.007 $ & $  0.007$\\
			& &  & $ 0.001$  & $0.006 $ & $ 0.006  $ & $ 0.002  $ & $ 0.007 $ & $ 0.007 $ & $  0.003$ & $0.007  $ & $ 0.007 $\\
			\cline{2-12}
			& & & & & & & & & && \\[-8pt]
			& $2 $ & $\gamma_1 $ & $ -0.049 $  & $ 0.024$ & $  0.027 $ & $  -0.116 $ & $ 0.031 $ & $0.044  $ & $  -0.207$ & $0.040  $ & $ 0.083 $\\
			& &  & $0.006  $  & $0.024 $ & $  0.024 $ & $ -0.004  $ & $  0.031$ & $ 0.031 $ & $ -0.020 $ & $ 0.039 $ & $ 0.039 $\\
			& & $\gamma_2 $ & $0.019  $  & $ 0.003$ & $  0.003 $ & $  0.049 $ & $ 0.003 $ & $0.006  $ & $ 0.079 $ & $ 0.004 $ & $ 0.011 $\\
			& &  & $0.006  $  & $0.003 $ & $ 0.003  $ & $  0.011 $ & $ 0.004 $ & $ 0.004 $ & $ 0.018 $ & $ 0.005 $ & $ 0.005 $\\
			& & $\gamma_3 $ & $ -0.142 $  & $0.031 $ & $ 0.051  $ & $  -0.229 $ & $ 0.039 $ & $ 0.091 $ & $ -0.330 $ & $ 0.045 $ & $ 0.154 $\\
			& & & $ 0.017 $  & $ 0.032$ & $  0.033 $ & $  0.021 $ & $  0.040$ & $ 0.040 $ & $0.017  $ & $ 0.046 $ & $0.046  $\\
			& & $\gamma_4 $ & $ 0.029 $  & $0.033 $ & $ 0.034  $ & $ 0.079  $ & $ 0.041 $ & $0.047  $ & $ 0.144 $ & $ 0.051 $ & $ 0.072 $\\
			& & & $ 0.012 $  & $ 0.033$ & $ 0.033  $ & $  0.026 $ & $0.040  $ & $ 0.040 $ & $ 0.039 $ & $  0.047$ & $ 0.049 $\\
			& & $\gamma_5 $ & $ -0.034 $  & $0.033 $ & $0.034   $ & $ -0.068  $ & $0.039  $ & $ 0.044 $ & $  -0.103$ & $0.048  $ & $ 0.058 $\\
			& & & $  0.005$  & $0.032 $ & $  0.032 $ & $ 0.006  $ & $ 0.036 $ & $ 0.036 $ & $ 0.007 $ & $ 0.041 $ & $0.042  $\\
			& & $\beta_1 $ & $ 0.001 $  & $ 0.001$ & $   0.001$ & $ -0.003  $ & $ 0.001 $ & $ 0.001 $ & $ -0.008 $ & $ 0.001 $ & $ 0.001 $\\
			& &  &$ 0.002 $  & $0.001 $ & $  0.001 $ & $  0.002 $ & $  0.001$ & $ 0.001 $ & $ 0.001 $ & $  0.001$ & $ 0.001 $\\
			& & $\beta_2 $ & $0.000  $  & $ 0.001$ & $ 0.001  $ & $-0.001   $ & $ 0.001 $ & $ 0.001 $ & $  -0.001$ & $0.001  $ & $ 0.001 $\\
			& &  & $ 0.001 $  & $ 0.001$ & $ 0.001  $ & $  0.000 $ & $ 0.001 $ & $ 0.001 $ & $ 0.001 $ & $  0.001$ & $ 0.001 $\\
			& & $\beta_3 $ & $0.000  $  & $0.010 $ & $0.010   $ & $  -0.009 $ & $ 0.011 $ & $  0.011$ & $ -0.019 $ & $0.012  $ & $ 0.012 $\\
			& &  & $0.002  $  & $0.010 $ & $  0.010 $ & $0.000   $ & $ 0.011 $ & $ 0.011 $ & $  -0.001$ & $ 0.012 $ & $ 0.012 $\\
			\cline{2-12}
			& & & & & & & & & && \\[-8pt]
			& $3 $ & $\gamma_1 $ & $ -0.002 $  & $ 0.015$ & $  0.015 $ & $  -0.012 $ & $ 0.019 $ & $ 0.019 $ & $ -0.029 $ & $ 0.026 $ & $ 0.027 $\\
			& &  & $ 0.004 $  & $ 0.015$ & $ 0.015  $ & $  0.001 $ & $  0.018$ & $ 0.018 $ & $0.001  $ & $ 0.023 $ & $ 0.023 $\\
			& & $\gamma_2 $ & $ -0.005 $  & $0.001 $ & $  0.001$&$-0.001 $ & $ 0.002  $ & $ 0.002 $ & $  0.006$ & $ 0.002 $ & $ 0.002 $ \\
			& &  & $ 0.007 $  & $0.002 $ & $ 0.002  $ & $  0.004 $ & $ 0.002 $ & $ 0.002 $ & $ 0.002 $ & $ 0.002 $ & $ 0.002 $\\
			& & $\gamma_3 $ & $ 0.051 $  & $0.016 $ & $ 0.019  $ & $  0.034 $ & $  0.015$ & $ 0.017 $ & $ 0.041 $ & $  0.021$ & $  0.022$\\
			& & &  $-0.009  $ & $ 0.014 $ & $0.014  $& $ -0.006  $ & $ 0.016 $ & $ 0.016 $&$ -0.006 $  & $0.021 $ & $ 0.021  $  \\
			& & $\gamma_4 $ & $ 0.000 $  & $0.020 $ & $ 0.020  $ & $ 0.009  $ & $ 0.024 $ & $0.024  $ & $ 0.018 $ & $ 0.031 $ & $0.031  $\\
			& & & $ -0.001 $  & $0.019 $ & $  0.019 $ & $  0.004 $ & $  0.023$ & $ 0.023 $ & $ 0.004 $ & $  0.028$ & $ 0.028 $\\
			& & $\gamma_5 $ & $ -0.006 $  & $0.019 $ & $  0.019 $ & $ -0.011  $ & $  0.024$ & $ 0.024 $ & $ -0.015 $ & $ 0.033 $ & $ 0.033 $\\
			& & & $ -0.005 $  & $ 0.019$ & $  0.019 $ & $ -0.007  $ & $ 0.022 $ & $0.022  $ & $ -0.011 $ & $ 0.028 $ & $ 0.028 $\\
			& & $\beta_1 $ & $0.003  $  & $ 0.001$ & $  0.001 $ & $   0.001$ & $ 0.001 $ & $ 0.001 $ & $ -0.001 $ & $ 0.001 $ & $  0.001$\\
			& &  &$ 0.002 $  & $0.001 $ & $  0.001 $ & $ 0.001  $ & $  0.001$ & $ 0.001 $ & $0.001  $ & $  0.001$ & $ 0.001 $\\
			& & $\beta_2 $ & $ 0.000 $  & $ 0.001$ & $ 0.001  $ & $ 0.000  $ & $ 0.001 $ & $ 0.001 $ & $  0.000$ & $ 0.001 $ & $ 0.001 $\\
			& &  & $ -0.001 $  & $ 0.001$ & $ 0.001  $ & $0.000   $ & $  0.001$ & $ 0.001 $ & $0.000  $ & $ 0.001 $ & $ 0.001 $\\
			& & $\beta_3 $ & $ 0.003 $  & $0.010 $ & $ 0.010  $ & $  0.001 $ & $ 0.012 $ & $ 0.012 $ & $ -0.005 $ & $ 0.016 $ & $ 0.016 $\\
			& &  & $ 0.003 $  & $ 0.010$ & $  0.010 $ & $ 0.002  $ & $ 0.012 $ & $ 0.012 $ & $ -0.001 $ & $0.016  $ & $0.016  $\\
		\end{tabular}
	}}
\end{table}

\end{document}